\documentclass{article}

\usepackage{amsfonts}
\usepackage{amssymb}
\usepackage{amsthm}
\usepackage{amsmath}

\usepackage{tikz,graphicx,subfigure}
\usetikzlibrary{arrows}
\usepgflibrary{decorations.markings}
\usetikzlibrary{decorations.markings}
\usepackage{overpic}
\usepackage{psfrag}
\usepackage{comment}

\newtheorem{theorem}{Theorem}[section]
\newtheorem{lemma}[theorem]{Lemma}

\newtheorem{RHproblem}[theorem]{RH problem}
\newenvironment{rhproblem}{\begin{RHproblem}\rm}{\end{RHproblem}}

\DeclareMathOperator{\Tr}{Tr}
\DeclareMathOperator{\Res}{Res}
\DeclareMathOperator{\Ai}{Ai}

\DeclareMathOperator{\area}{area}
\DeclareMathOperator{\diag}{diag}
\DeclareMathOperator{\supp}{supp}
\renewcommand{\Re}{\mathop{\mathrm{Re}}}
\renewcommand{\Im}{\mathop{\mathrm{Im}}}

\newcommand{\ds}{\displaystyle}

\theoremstyle{definition}

\newtheorem{definition}[theorem]{Definition}
\newtheorem{remark}[theorem]{Remark}

\numberwithin{equation}{section}

\begin{document}
\pagenumbering{arabic}
\title{Orthogonal polynomials in the normal matrix model with a cubic potential}
\author{Pavel M. Bleher\footnote{Department of Mathematical Sciences, Indiana University-Purdue University Indianapolis, 402 N. Blackford
St., Indianapolis, IN 46202, U.S.A. email: bleher@math.iupui.edu.} \ and \  
Arno B.J. Kuijlaars\footnote{Department of Mathematics, Katholieke Universiteit Leuven, Celestijnenlaan 200B bus 2400, 3001~Leuven, Belgium.
email: arno.kuijlaars@wis.kuleuven.be.}}
\date{\today}
\maketitle

\begin{abstract}
We consider the 
normal matrix model with a cubic potential. The model is ill-defined, 
and in order to reguralize it, Elbau and Felder introduced a model with a 
cut-off and corresponding system of orthogonal polynomials with respect to a varying exponential
weight on the {\it cut-off region} on the complex plane. 
In the present paper we show how to define  orthogonal polynomials 
on a specially chosen {\it system of infinite contours} on the complex plane, 
 without any cut-off, which satisfy the same
recurrence algebraic identity  that is asymptotically valid
for the orthogonal polynomials of Elbau and Felder. 

The main goal of this paper is to develop the Riemann-Hilbert (RH) approach to the 
orthogonal polynomials under consideration
and to obtain their asymptotic behavior on the complex plane as the 
degree $n$ of the polynomial goes to infinity. 
As the first step in the RH approach, we introduce an auxiliary vector equilibrium 
problem for a pair of measures $(\mu_1,\mu_2)$ on the complex plane. We then 
formulate a $3\times 3$ matrix valued RH problem for the orthogonal polynomials 
in hand, and we apply the nonlinear steepest descent method of Deift-Zhou to 
the asymptotic analysis of the RH problem. The central steps in our study 
are a sequence of transformations of the RH problem, based on the 
equilibrium vector measure $(\mu_1,\mu_2)$, and the construction of a global parametrix.

The main result of this paper is a derivation of the large $n$ asymptotics of 
the orthogonal polynomials on the whole complex plane.  
We prove that the distribution of zeros of the orthogonal polynomials converges to 
the measure $\mu_1$, the first component of the equilibrium measure. We also 
obtain analytical results for the measure $\mu_1$ relating it to the 
distribution of eigenvalues in the normal matrix model which is uniform 
in a domain bounded by a simple closed curve.

\end{abstract}

\section{Introduction}

The normal matrix model is a probability measure on the space of
$n \times n$ normal matrices $M$ of the form
\begin{equation} \label{NMM1} 
	\frac{1}{Z_n} \exp \left( - n \Tr \mathcal{V}(M) \right) dM.
\end{equation}
A typical form for $\mathcal V$ is
\begin{equation} \label{NMM2} 
	\mathcal V(M) = \frac{1}{t_0} \left(M M^* - V(M) - \overline{V}(M^*)) \right), \qquad t_0 > 0,
	\end{equation}
where $V$ is a polynomial and $\overline{V}$ is the polynomial obtained from $V$ by conjugating
the coefficients. In this case, the model may alternatively be defined
on general $n \times n$ complex matrices $M$, see e.g.\ \cite{Zab2}. In this paper we study in particular
\eqref{NMM1}--\eqref{NMM2} with a cubic potential 
\begin{equation} \label{cubicV}  
	V(M) = \frac{t_3}{3} M^3,  \qquad t_3 > 0.
	\end{equation}

The main feature of the normal matrix model is that the eigenvalues of $M$ fill out a 
bounded two-dimensional domain $\Omega$ as $n \to \infty$ with a uniform density. 
Wiegmann and Zabrodin \cite{WieZab} showed that if
\begin{equation} \label{generalV} 
	V(M) = \sum_{k=1}^{\infty} \frac{t_k}{k} M^k
	\end{equation}
then $\Omega = \Omega(t_0; t_1, t_2, \ldots)$ is such that 
\begin{equation} \label{harmonicmoments} 
	t_0 = \frac{1}{\pi} \area(\Omega), \qquad t_k = - \frac{1}{\pi} \iint_{\mathbb C \setminus \Omega} \frac{dA(z)}{z^k}, \quad k=1,2,3, \ldots. 
	\end{equation}
	where $dA$ denotes the two-dimensional Lebesgue measure in the complex plane. The relations \eqref{harmonicmoments}
characterize the domain $\Omega$ by means of its area $\pi t_0$ and its exterior harmonic moments $t_k$ for $k \geq 1$
(it is assumed that $0 \in \Omega$, and the integrals in \eqref{harmonicmoments} 
need to be regularized for $k\leq 2$). 
An important fact, first shown in \cite{KKMWZ}, is that the boundary of $\Omega$ as a function of $t_0 > 0$ (which is 
seen as a time parameter) evolves  according to the model of Laplacian growth. The Laplacian growth is unstable 
and singularities such as boundary cusps may appear in finite time. 
See also \cite{MiPuTe,TeBeAgZaWi,Zab1,Zab2} for related work and surveys.

However, from a mathematical point of view, the model \eqref{NMM1}--\eqref{NMM2} is
not well-defined if $V$ is a polynomial of degree $\geq 3$, since then the integral
\begin{equation} \label{Zn} 
	Z_n = \int \exp(-n \Tr \mathcal V(M)) dM 
	\end{equation}
diverges, and one cannot normalize the measure \eqref{NMM1} to make it a probability measure.
To make the model well-defined, Elbau and Felder \cite{Elb,ElbFel} propose
to use a cut-off procedure. Instead of considering all normal matrices $M$, they restrict
to normal matrices with spectrum in a fixed bounded domain $D$. The integral \eqref{Zn} 
restricted to all such matrices is convergent and the model is well-defined.  
As it is the case for unitary random matrices, the eigenvalues of $M$ are then
distributed according to a determinantal point process with a kernel that is built out of
orthogonal polynomials with respect to the scalar product 
\begin{equation} \label{Dscalarproduct} 
	\langle f, g \rangle_D = \iint_D f(z) \overline{g(z)}  e^{-n \mathcal V(z)} dA(z) 
	\end{equation}
(which depends on $n$), with 
\begin{equation} \label{weightV} 
	\mathcal V(z) = \frac{1}{t_0} \left( |z|^2 - V(z) - \overline{V(z)}\right),
	\end{equation}
see \cite{Elb}. 
For each $n$, we have the sequence $(P_{k,n})_{k=0}^{\infty}$ of monic polynomials (i.e., $P_{k,n}(z) = z^k + \cdots$) 
such that
\[ \langle P_{k,n}, P_{j,n} \rangle_{D} = h_{k,n} \delta_{j,k}, \]
and then the correlation kernel for the determinantal point process is
\[ K_n(w,z) = e^{-\frac{n}{2}(\mathcal V(w) + \mathcal V(z))} \sum_{k=0}^{n-1} \frac{P_{k,n}(z) \overline{P_{k,n}(w)}}{h_{k,n}}. \]
Elbau and Felder \cite{Elb,ElbFel} prove that for a polynomial $V$ as in \eqref{generalV} with $t_1 =0$, $|t_2| < 1$,
and for $t_0$ small enough it is possible to find a suitable domain $D$ such that indeed the
eigenvalues of the normal matrix model with cut-off $D$ accumulate on a domain $\Omega$ as $n \to \infty$.
The domain $\Omega$ is characterized by \eqref{harmonicmoments} and so in particular evolves according to Laplacian
growth in the time parameter $t_0$. Note that the cut-off approach works fine for $t_0$ small enough but
fails to capture important features of the normal matrix model such as the formation
of cusp singularities at a critical value of $t_0$.

Elbau \cite{Elb} also discusses the zeros of the orthogonal polynomials $P_{n,n}$ as $n \to \infty$. 
For the cubic potential \eqref{cubicV} and again for $t_0$ sufficiently small, 
he shows that these zeros accumulate on a starlike set
\begin{equation} \label{Sigma1Elbau} 
	\Sigma_1 = [0,x^*] \cup [0, \omega x^*] \cup [0, \omega^2 x^*], \qquad \omega = e^{2\pi i/3},
	\end{equation}
for some explicit value of $x^* >0$. The set $\Sigma_1$ is contained in $\Omega$. 
The limiting distribution of zeros is a probability measure $\mu_1^*$ on $\Sigma_1$ 
satisfying
\begin{equation} \label{mu1balayage} 
	\int \log|z-\zeta| \, d\mu_1^*(\zeta) = \frac{1}{\pi t_0} \iint_{\Omega} \log |z-\zeta| \, dA(\zeta),
	\qquad z \in \mathbb C \setminus \Omega. 
	\end{equation}
For general polynomial $V$ and $t_0$ small enough, Elbau conjectures
that the zeros of the orthogonal polynomials
accumulate on a tree-like set strictly contained in $\Omega$ with a limiting distribution whose
logarithmic potentials outside of $\Omega$ agrees with that of the normalized Lebesgue measure on $\Omega$.

In this paper we want to analyze the orthogonal polynomials for the normal matrix model without
making a cut-off. We cannot use the scalar product
\begin{equation} \label{Cscalarproduct} 
	\langle f, g \rangle = \iint_{\mathbb C} f(z) \overline{g(z)} e^{-n \mathcal V(z)} dA(z) 
	\end{equation}
defined on $\mathbb C$ since the integral \eqref{Cscalarproduct} diverges if
$f$ and $g$ are polynomials and  $V$ is a polynomial of degree $\geq 3$.
Our approach is to replace \eqref{Cscalarproduct} by a Hermitian form defined on
polynomials that is a priori not given by any integral, but that should satisfy
the relevant algebraic properties of the scalar product \eqref{Cscalarproduct}. 
We define and classify these Hermitian forms. See the next section for precise
statements.

It turns out that there is more than one possibility for such a Hermitian form.
We conjecture that for any polynomial $V$ it is possible to choose the Hermitian
form in such a way that that the corresponding orthogonal 
polynomials have the same asymptotic behavior as the 
orthogonal polynomials in the cut-off approach of Elbau and Felder. That is,
for small values of $t_0$ the zeros of the polynomials accumulate on a tree like set $\Sigma_1$
with a limiting distribution $\mu_1^*$ satisfying \eqref{mu1balayage}.

We are able to establish this for the cubic case \eqref{cubicV} and this is the
main result of the paper. We recover
the same set $\Sigma_1$ as in \eqref{Sigma1Elbau} and also the domain $\Omega$ that
evolves according to the Laplace growth, as we will show. In our approach we do not
have to restrict to $t_0$ sufficiently small. We can actually take any $t_0$ up to 
the critical time $t_{0,crit}$. For this value the endpoints of $\Sigma_1$ come to the boundary
of $\Omega$. Then three cusps are formed on the boundary and the Laplacian growth
breaks down. An asymptotic analysis at the critical time would involve the
Painlev\'e I equation, which is what we see in our model and that we will address
in a future paper.

We note that Ameur, Hedenmalm and Makarov \cite{AmHeMa, HedMak} do not use a cut-off. Instead 
they consider cases where $\mathcal V(z)$ tends to $+\infty$ as $z \to \infty$ in all directions of
the complex plane. This approach does not include cases \eqref{weightV} with a 
polynomial $V$ of degree $\geq 3$.

\section{Statement of results} \label{statement}

\subsection{Hermitian forms} \label{subsecHermitian}

We propose to consider orthogonal polynomials with respect to Hermitian forms
that share the algebraic properties of the scalar product \eqref{Cscalarproduct}
with $\mathcal V(z)$ given by \eqref{weightV}.
The Hermitian form is a sesquilinear form $\langle \cdot, \cdot \rangle$ defined
on the vector space of complex polynomials in one variable that satisfies
the Hermitian condition
\begin{equation} \label{hermitian} 
	\langle f, g \rangle = \overline{ \langle g, f \rangle }.
	\end{equation}
We use the convention that $\langle \cdot, \cdot \rangle$ is linear in the
first argument and conjugate-linear in the second.

To see what kind of condition we want to put on the Hermitian form
we look at the scalar product \eqref{Dscalarproduct} on the cut-off region $D$.
From the complex version of Green's theorem we obtain for polynomials $f$ and $g$
\begin{multline}
	\frac{t_0}{2i} \oint_{\partial D} f(z) \overline{g(z)} e^{-n \mathcal V(z)} dz = 
		t_0 \iint_D \frac{\partial}{\partial \overline{z}} \left[f(z) \overline{g(z)} e^{-n \mathcal V(z)} \right] dA(z)  \\
		= t_0 \iint_D f(z) \overline{g'(z)} e^{-n \mathcal V(z)} dA(z) 
	  -  n \iint_D z f(z) \overline{g(z)} e^{-n \mathcal V(z)} dA(z) \\
		 			+  n  \iint_D f(z) \overline{V'(z) g(z)} e^{-n \mathcal V(z)} dA(z),
			\end{multline}
which can be written as 
\begin{equation} \label{structureD0} t_0 \langle f, g' \rangle_D - n  \langle zf, g \rangle_D 
	+ n  \langle f, V' g \rangle_D
	= \frac{t_0}{2i} \oint_{\partial D} f(z) \overline{g(z)} e^{-n \mathcal V(z)} dz. 
\end{equation}
The cut-off approach works if the domain $D$ is chosen such that 
the effect of the boundary $\partial D$ becomes small when $n$ is large. 
This means that the boundary integral in \eqref{structureD0} is exponentially small 
compared to the  other terms in \eqref{structureD0} in the large $n$ limit.

Inspired by \eqref{structureD0} our idea is to ignore the right-hand side of \eqref{structureD0} 
and require that the Hermitian form  satisfies the identity
\begin{equation} \label{structure} 
	t_0 \langle f, g' \rangle - n \langle zf, g \rangle + n \langle f, V' g \rangle = 0 
	\end{equation}
for all polynomials $f$ and $g$. We call \eqref{structure} the structure relation. 
Combining \eqref{structure} with \eqref{hermitian} we also have the dual structure relation
\begin{equation} \label{structuredual} 
	t_0 \langle f' , g \rangle - n \langle f, zg \rangle + n \langle V' f,  g \rangle = 0. 
	\end{equation}
We do not impose the condition that the Hermitian form is positive definite, and therefore
it may not be a scalar product.

Our first result is a classification of Hermitian forms satisfying \eqref{hermitian}
and \eqref{structure}.
Sesquilinear forms that satisfy \eqref{structure} can be found as follows.
We assume $V$ is a polynomial and 
\[ d = \deg V - 1. \]
There are $d+1$ directions in the complex plane, given by $\arg z = \theta_j$, $j=0,1, \ldots, d$,
such that 
\[ \arg(- V(z)) \to 0 \qquad \text{as } |z| \to \infty \text{ with } \arg z = \theta_j. \]
We choose the directions in such a way that   
\begin{equation} \label{thetaj} 
	\theta_j = \theta_0 + \frac{2\pi j}{d+1}, \qquad j=0,1,\ldots, d+1.
	\end{equation}
Let $\Gamma_j$ for $j=0, \ldots, d$, be a simple smooth curve in the complex plane, starting at infinity 
in the direction $\arg z = \theta_j$ and ending at infinity in the direction $\arg z = \theta_{j+1}$.
See Figure \ref{fig:contoursGamma} for possible contours $\Gamma_j$ in case $d =2$ and $t_3 > 0$.
Then define for polynomials $f$ and $g$, and for $j,k=0,1, \ldots, d$,
\begin{equation} \label{PhijkV} \Phi_{j,k}(f,g) = \int_{\Gamma_j} dz \int_{\overline{\Gamma}_k} dw
	f(z) \overline{g}(w)  e^{- \frac{n}{t_0} \left(wz - V(z) - \overline{V}(w) \right)},
		\end{equation}
where $\overline{g}$ is the polynomial obtained from $g$ by conjugating the coefficients.
The integrals in \eqref{PhijkV} are convergent because of our choice of contours.
	
It is an easy exercise, based on integration by parts, that the sesquilinear forms $\Phi_{j,k} (\cdot, \cdot)$
indeed satisfy \eqref{structure} and \eqref{structuredual}. Then so does any linear combination.
We prove the following.

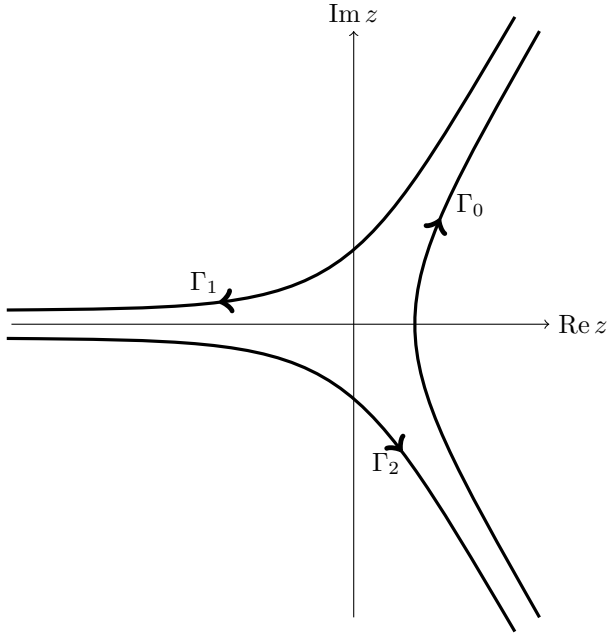
\begin{figure}[t]
\begin{center}
\begin{tikzpicture}[scale=1.3,decoration={markings,mark=at position .67 with {\arrow[black,line width=0.8mm]{>};}}]
\draw[->](-3.5,0)--(2,0) node[right]{$\Re z$};
\draw[->](0,-3)--(0,3) node[above]{$\Im z$};

\draw[postaction={decorate}, very thick] (1.9,-3)..controls (0.2,0)..(1.9,3) node[near end, right]{$\Gamma_0$};
\draw[postaction={decorate}, very thick, rotate around={120:(0,0)}] (1.9,-3)..controls (0.2,0)..(1.9,3) node[near end, above]{$\Gamma_1$};
\draw[postaction={decorate}, very thick, rotate around={-120:(0,0)}] (1.9,-3)..controls (0.2,0)..(1.9,3) node[near end, left]{$\Gamma_2$};
\end{tikzpicture}
\end{center}
\caption{Contours $\Gamma_0$, $\Gamma_1$, $\Gamma_2$ for the case of a cubic potential ($d=2$)}
\label{fig:contoursGamma}
\end{figure}

\begin{theorem} \label{theorem1}
For every $n \in \mathbb N$, $t_0 > 0$, and polynomial $V$ of degree $\deg V = d+1$, 
the real vector space of Hermitian forms satisfying \eqref{hermitian} and
\eqref{structure}  is $d^2$ dimensional.

Every such Hermitian form can be uniquely represented as
\begin{equation} \label{HformV}
 \langle f, g \rangle = \sum_{j=0}^d \sum_{k=0}^d C_{j,k} \Phi_{j,k}(f,g)
 \end{equation}
where $\Phi_{j,k}$ is given by \eqref{PhijkV} and
where $C = \begin{pmatrix} C_{j,k} \end{pmatrix}_{j,k=0}^d$ is a Hermitian matrix having  zero row and column sums.
\end{theorem} 
The proof of Theorem \ref{theorem1} is in Section \ref{proofTheo1}.

Now the following obvious question arises. Given a polynomial $V$, is there a Hermitian matrix
$C = (C_{j,k})$ such that the Hermitian form \eqref{HformV} captures the main
features of the normal matrix model? We answer this question for the cubic case. 

\subsection{The cubic potential and rotational symmetry}

We consider from now on the case
\begin{equation} \label{cubicV2} 
	V(z) = \frac{t_3}{3} z^3, \qquad t_3 > 0. 
	\end{equation}
Then $d = \deg V - 1 = 2$ and so by Theorem \ref{theorem1} the
space of Hermitian forms \eqref{HformV} is four-dimensional.
For the case \eqref{cubicV2} it is natural to require 
an additional rotational symmetry
\begin{equation} \label{rotation} 
	\langle f(\omega z), g(\omega z) \rangle = \langle f, g \rangle, \qquad \omega = e^{2\pi i/3}.
	\end{equation}
The condition \eqref{rotation} corresponds to the fact that $V(\omega z) = V(z)$ for the cubic potential \eqref{cubicV2}
so that the integral in \eqref{Cscalarproduct}
is invariant under the change of variables $z \mapsto \omega z$.
Thus \eqref{rotation} is specific for the cubic potential and will have to be modified for other potentials.

One may verify from the definition \eqref{PhijkV} with $V$ given by \eqref{cubicV2} that
\[ \Phi_{j,k} (f(\omega \cdot), g(\omega \cdot)) = \Phi_{j+1,k+1} (f,g), \]
where the indices are taken modulo $3$. 
Thus
\[ \sum_{j=0}^2 \sum_{k=0}^2 C_{j,k} \Phi_{j,k}(f(\omega \cdot), g(\omega \cdot))
	= \sum_{j=0}^2 \sum_{k=0}^2 C_{j-1,k-1} \Phi_{j,k}(f,g), \]
where again the indices are taken modulo $3$ so that for example $C_{-1,-1} = C_{2,2}$.
It follows that \eqref{HformV} satisfies the symmetry condition \eqref{rotation} if and only if
\begin{equation} \label{circulant} 
	C_{j,k} = C_{j-1,k-1}, \qquad \text{(indices modulo $3$)}. 
	\end{equation}	
The condition \eqref{circulant} means that $C = (C_{j,k})$ is a circulant matrix. 
As dictated by Theorem \ref{theorem1}, we also require that $C$ is Hermitian with zero row 
and column sums.

The real vector space of circulant Hermitian matrices of size $3 \times 3$ with zero row
and column sums is two-dimensional. A basis is given by the two matrices
\begin{equation} \label{basismatrix} 
	\begin{pmatrix} 2 & - 1 & - 1 \\ - 1 & 2 & -1 \\ -1 & - 1 & 2 \end{pmatrix},
	\qquad \begin{pmatrix} 0 & i & -i \\ -i & 0 & i \\ i & -i & 0 \end{pmatrix}. 
	\end{equation}
It turns out that we are able to do asymptotic analysis on the orthogonal polynomials
only if we choose for $C$ a multiple of the second basis matrix in \eqref{basismatrix}.
However, we have no a priori reason to prefer this matrix above the other one, or above 
a linear combination of the two. It is only because of our ability to do 
large $n$ asymptotics that we choose
\begin{equation} \label{Cmatrix} 
	C = \frac{1}{2\pi i} \begin{pmatrix} 0 & -1 & 1 \\ 1 & 0 & -1 \\ -1 & 1 & 0 \end{pmatrix}. 
	\end{equation}
This leads to the following definition of the Hermitian form and the corresponding orthogonal
polynomials.	
\begin{definition}
Given $n$, $t_0, t_3 > 0$ we define the Hermitian form $\langle \cdot, \cdot \rangle$ on
the vector space of polynomials by
\begin{equation} \label{Hform3} 
	\langle f, g \rangle = \frac{1}{2\pi i}
	\sum_{j=0}^2 \sum_{k=0}^2 \epsilon_{j,k} 
		\int_{\Gamma_j} dz \int_{\overline{\Gamma}_k} dw 
			f(z)  \overline{g}(w)  e^{- \frac{n}{t_0} \left( wz -  \frac{t_3}{3}(w^3 + z^3) \right)}
			\end{equation}  
where
\begin{equation} \label{epsilonjk}
	\left( \epsilon_{j,k} \right)_{j,k=0}^2 = \begin{pmatrix} 0 & -1 & 1 \\ 1 & 0 & -1 \\ -1 & 1 & 0 \end{pmatrix}.
	\end{equation}

Note that the Hermitian form \eqref{Hform3} depends on $n$ even though we do not emphasize
this in the notation. 
For each $n$, we denote by $(P_{k,n})_{k=0,1,2,\ldots}$ the sequence of monic orthogonal polynomials for
the Hermitian form \eqref{Hform3}. That is, $P_{k,n}(z) = z^k + \cdots$ is a polynomial of
degree $k$ such that
\begin{equation} \label{orthogonal} 
	\langle P_{k,n}, z^j \rangle = 0, \qquad \text{for } j = 0, 1, \ldots, k-1. 
	\end{equation}
\end{definition}

\subsection{Existence of orthogonal polynomials $P_{n,n}$}

The Hermitian form $\langle \cdot, \cdot \rangle$ may not be positive definite,
and therefore the  existence and uniqueness 
of the orthogonal polynomials is not guaranteed. This will be our next
main result. We focus on the diagonal polynomials $P_{n,n}$
and for these polynomials we also determine the limiting behavior
of the zeros.

\begin{theorem} \label{theorem2}
Let $t_3 > 0$ and define
\begin{equation} \label{t0crit} 
	t_{0,crit} = \frac{1}{8 t_3^2}.
	\end{equation}
Then for every $t_0 \in (0, t_{0,crit})$ the orthogonal polynomials $P_{n,n}$ 
for the Hermitian form \eqref{Hform3} exist 
if $n$ is sufficiently large.
In addition, the zeros of $P_{n,n}$ 
accumulate as $n \to \infty$ on the set
\begin{equation} \label{Sigma1}
	\begin{aligned}
	\Sigma_1 & = \{ z \in \mathbb C \mid z^3 \in [0, (x^*)^3] \} =  \bigcup_{j=0}^2 [0, \omega^j x^*],
		 \end{aligned}
		 \end{equation}
where $\omega = e^{2\pi i/3}$ and 
\begin{equation} \label{xstar} 
	x^* = \frac{3}{4 t_3}  \left( 1- \sqrt{1-8 t_0 t_3^2}\right)^{2/3}.
	\end{equation}
\end{theorem}

Theorem \ref{theorem2} will follow from a strong asymptotic formula for
the orthogonal polynomials, see Lemma \ref{lemmaRH}.

\subsection{Limiting zero distribution}

For a polynomial $P$ of degree $n$ with zeros $z_1, \ldots, z_n$ in the
complex plane, we use
\begin{equation} \label{zerocounting} 
	\nu(P) = \frac{1}{n} \sum_{j=1}^n  \delta_{z_j} 
	\end{equation}
to denote the normalized zero counting measure. 

The normalized zero counting measures $\nu(P_{n,n})$ of the 
orthogonal polynomials $P_{n,n}$ have a limit that we 
characterize in terms of the solution of a vector equilibrium problem
from logarithmic potential theory.
We use the following standard notation. 
For a measure $\mu$, we define the logarithmic energy
\begin{equation} \label{energy} 
	I(\mu) = \iint \log \frac{1}{|x-y|} d\mu(x) d\mu(y)
	\end{equation}
and for two measures $\mu$ and $\nu$, we define the mutual energy
\begin{equation} \label{mutualenergy} 
	I(\mu,\nu) = \iint \log \frac{1}{|x-y|} d\mu(x) d\nu(y). 
	\end{equation}
In addition to the set $\Sigma_1$ from \eqref{Sigma1} we also need
\begin{equation} \label{Sigma2}
	\begin{aligned}
	\Sigma_2 & = \{ z \in \mathbb C \mid z^3 \in \mathbb R^- \} =  \bigcup_{j=0}^2 [0, -\omega^j \infty). 
		 \end{aligned}
		 \end{equation}
		 
\begin{definition}
Given $t_0, t_3 > 0$,  we define the energy functional
\begin{multline} \label{energyE}
	E(\mu_1, \mu_2) = 
	 I(\mu_1) + I(\mu_2) - I(\mu_1, \mu_2) \\
	 + \frac{1}{t_0} \int 
	 \left( \frac{2}{3\sqrt{t_3}} |z|^{3/2} - \frac{t_3}{3} z^3 \right) d\mu_1(z) \end{multline}
The vector equilibrium problem is to minimize \eqref{energyE} 	 
among all measures $\mu_1$ and $\mu_2$ such that
\begin{align} \label{mu1mu2}
	 \int d\mu_1 & = 1, \quad  \supp(\mu_1) \subset \Sigma_1, \qquad
	 \int d\mu_2 = \frac{1}{2}, \quad \supp(\mu_2) \subset \Sigma_2,
	 \end{align}	 		
where $\Sigma_1$ and $\Sigma_2$ are given by \eqref{Sigma1} and \eqref{Sigma2}.
\end{definition}

We prove the following.

\begin{theorem} \label{theorem3}
Let $t_3 > 0$, $0 < t_0 \leq t_{0,crit}$ and let $x^* > 0$ be given by \eqref{xstar}.
Then there is a unique minimizer $(\mu_1^*, \mu_2^*)$ for $E(\mu_1, \mu_2)$
among all vectors of measures $(\mu_1, \mu_2)$ satisfying \eqref{mu1mu2}.
If $t_0 < t_{0,crit}$ then the first component $\mu_1^*$ of the minimizer is the weak limit of 
the normalized zero counting measures  of the polynomials $P_{n,n}$ as $n \to \infty$.
\end{theorem}

\subsection{Two dimensional domain}

We finally make the connection to the domain $\Omega$ that contains
the eigenvalues in the normal matrix model, and that is characterized
by the relations \eqref{harmonicmoments}.
In the cubic model that we are considering we are able to construct 
the domain $\Omega$ in terms of the solution $(\mu_1^*, \mu_2^*)$ 
of the vector equilibrium problem for the 
energy functional \eqref{energyE} as follows.

\begin{figure}
\centering
\includegraphics[width=6cm]{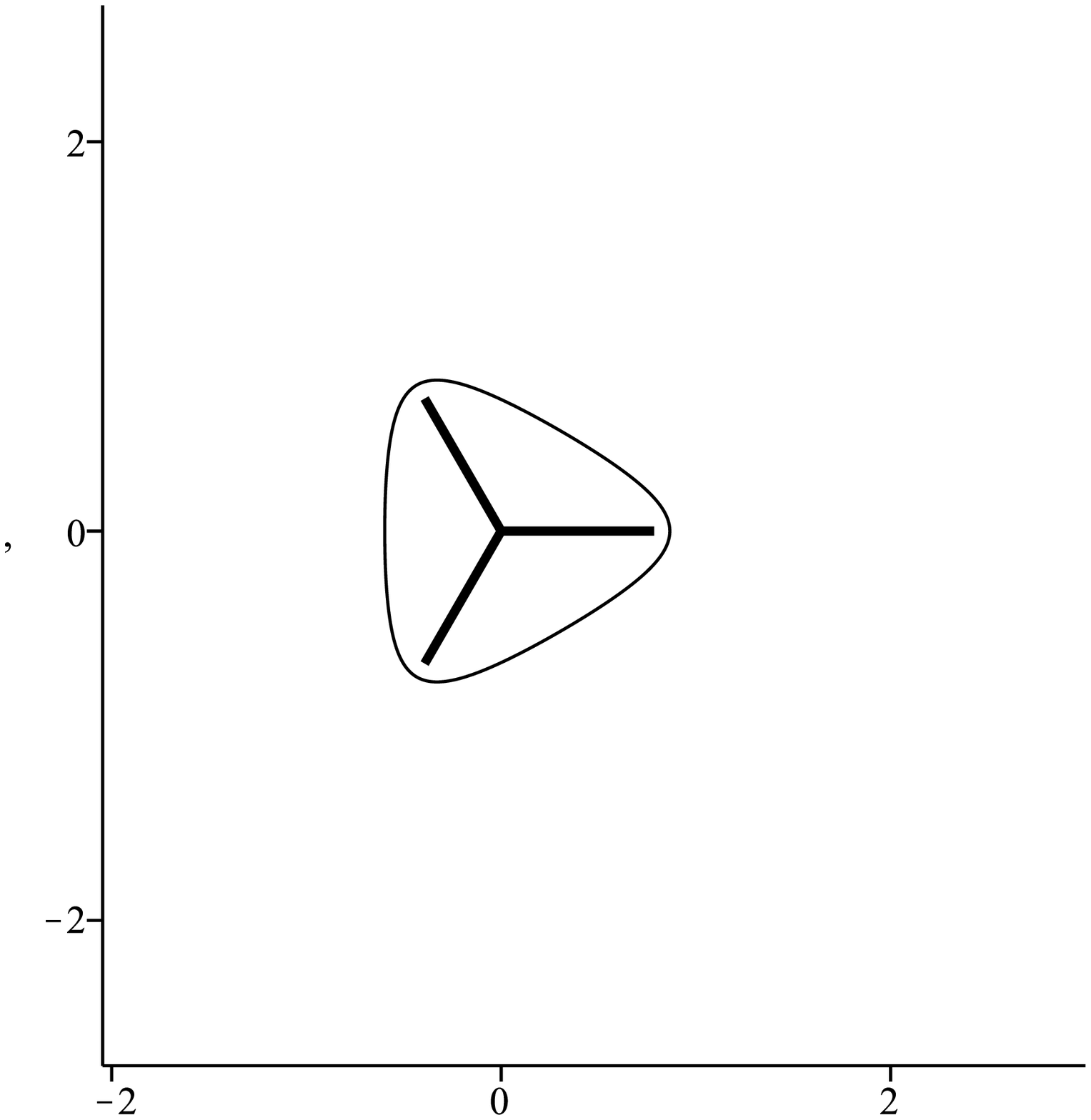} 
\includegraphics[width=6cm]{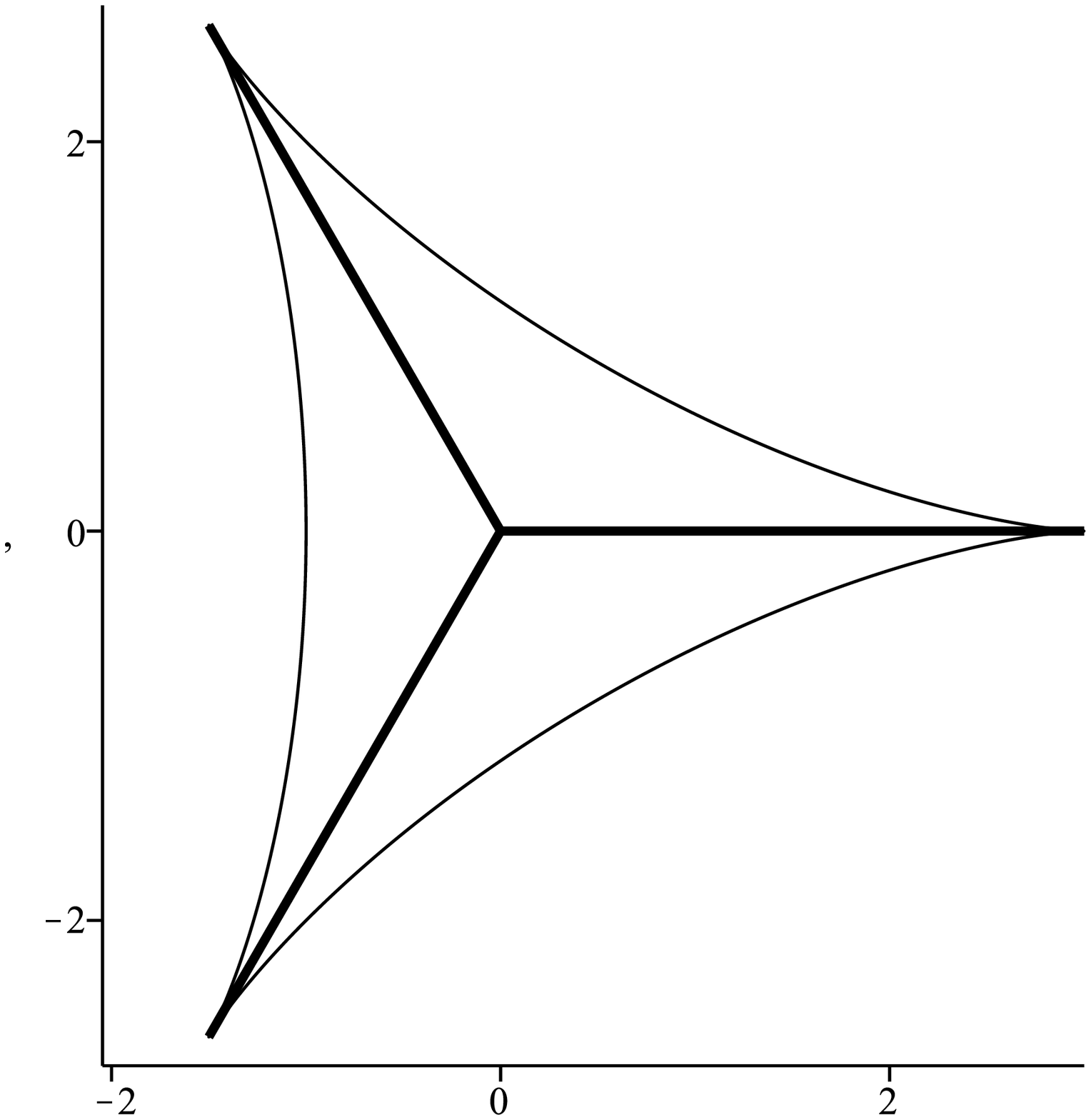}
\caption{The set $\Sigma_1$ where the zeros of $P_{n,n}$ accumulate and the 
boundary of the domain $\Omega$ for the values $t_0 = 1/2$ and $t_3 =1/4$ (left) 
and for the critical values $t_0=2$ and $t_3 =1/4$ (right). }
\label{Domains}
\end{figure}

\begin{theorem} \label{theorem4}
Let $t_3 > 0$ and $0 < t_0 \leq t_{0,crit}$. Let $\mu_1^*$ be the first
component of the minimizer of the vector equilibrium problem, as described
in Theorem~\ref{theorem3}. 
Then the equation
\begin{equation} \label{Schwarz}
	t_3 z^2 + t_0 \int \frac{d\mu_1^*(\zeta)}{z-\zeta} = \overline{z} 
	\end{equation}
defines a simple curve $\partial \Omega$ that is the boundary of a
domain $\Omega$ containing $\Sigma_1$.

The domain $\Omega$ satisfies \eqref{harmonicmoments}
and is such that
\begin{equation} \label{potential} 
	\int \frac{d\mu_1^*(\zeta)}{z-\zeta} = \frac{1}{\pi t_0} \iint_{\Omega} \frac{dA(\zeta)}{z-\zeta}.
	\qquad z \in \mathbb C \setminus \overline{\Omega}
	\end{equation}
\end{theorem}
The identity \eqref{Schwarz} means that the left-hand side of \eqref{Schwarz}
is the Schwarz function of $\partial \Omega$. See Figure \ref{Domains} for
the domain $\Omega$ 

The proof of Theorem \ref{theorem4} is in section \ref{proofTheo4}.
It follows from an analysis of the vector equilibrium problem and an associated
three sheeted Riemann surface.
This Riemann surface will also be important in the proofs of 
Theorems \ref{theorem2} and \ref{theorem3} that are based on a steepest
descent analysis of a $3 \times 3$ matrix valued Riemann-Hilbert problem. The
connection to the RH problem is explained in  section \ref{preliminary} and we 
refer to the discussion in subsection \ref{subsec:discussion} for motivation 
and connections with previous work. The actual steepest descent analysis is done
in section \ref{RHanalysis}. It is quite involved and ultimately leads to a
strong asymptotic formula for the polynomials $P_{n,n}$ as $n \to \infty$,
see Lemma \ref{lemmaRH}, from which the Theorems \ref{theorem2} and \ref{theorem3}
will follow. 

\section{Proof of Theorem \ref{theorem1}} \label{proofTheo1}

We have already seen that any linear combination \eqref{HformV} of the basic
forms \eqref{PhijkV} gives a sesquilinear 
form that satisfies \eqref{structure}. 
We first show that any such sesquilinear form is characterized by a unique
matrix $C = (C_{j,k})$ with zero row and column sums. 

\begin{lemma}  \label{formunique}
Any linear combination of the basic forms \eqref{PhijkV} can be written
in the form \eqref{HformV} with a unique matrix $C = \begin{pmatrix} C_{j,k} \end{pmatrix}_{j,k=0}^d$ 
having zero row and column sums. 
\end{lemma}

\begin{proof}
The basic forms  \eqref{PhijkV} are such that
\[ \sum_{k=0}^d \Phi_{j,k}(f,g) = 0, \qquad \text{for every } j \]
and 
\[ \sum_{j=0}^d \Phi_{j,k}(f,g) = 0, \qquad \text{for every } k. \]
Thus one may add a constant to a row or to a column of $C$
and obtain the same form \eqref{HformV}. Hence we may assume
that $C$ has zero row sums
\begin{equation} \label{Crowsum} 
	\sum_{k=0}^d C_{j,k} = 0, \qquad \text{for every } j 
	\end{equation}
and column sums
\begin{equation} \label{Ccolumnsum} 
	\sum_{j=0}^d C_{j,k} = 0, \qquad \text{for every } k. 
	\end{equation}

To show that the conditions \eqref{Crowsum} and \eqref{Ccolumnsum} 
determine the sesquilinear form \eqref{HformV}, we suppose that 
\begin{equation} \label{formunique1}
	\forall f, g : \, \langle f, g \rangle = \sum_{j=0}^d \sum_{k=0}^d C_{j,k} \Phi_{j,k}(f,g) = 0
	\end{equation}
where $C = \begin{pmatrix} C_{j,k} \end{pmatrix}$ has zero row and column sums,
and we prove that $C=O$.

Recall that $V$ is a polynomial of degree $d+1$ which we write as
\begin{equation} \label{V1z}
	V(z) = \frac{t_{d+1}}{d+1} z^{d+1} + V_1(z), \qquad \deg V_1 \leq d. 
	\end{equation}
The directions $\theta_j$ from \eqref{thetaj} are such that
\[ \frac{t_{d+1}}{d+1} z^{d+1} \in \mathbb R^-, \qquad \arg z = \theta_j. \]

We deform the contour $\Gamma_j$ in the definition \eqref{PhijkV} of $\Phi_{j,k}(f,g)$
so that it consists of the two rays $\arg z = \theta_j$ and $\arg z = \theta_{j+1}$.
Then it follows from \eqref{PhijkV} and \eqref{formunique1} that for every polynomial $f$,
\begin{align} \nonumber 
	0 & =  \langle f, 1 \rangle  \\
	& = \label{formunique2}
	\sum_{j=0}^d \int_0^{\infty e^{i \theta_j}} dz  f(z) e^{\frac{n}{t_0} V(z)} 
	\sum_{k=0}^d (C_{j-1,k} -C_{j,k}) \int_{\overline{\Gamma}_k} dw e^{-\frac{n}{t_0} (wz - \overline{V}(w))}.
		\end{align}
where $C_{-1,k} = C_{d,k}$.
Making the substitution $z = x e^{i \theta_j}$ in the integral over $z$ in \eqref{formunique2} 
and using \eqref{V1z} we obtain
\begin{align} \label{formunique3}
	 \sum_{j=0}^d \int_0^{\infty} f(e^{i \theta_j} x) e^{- c x^{d+1}} \phi_j(x) dx = 0
	\end{align}
where $c = \frac{n}{t_0} \frac{|t_{d+1}|}{d+1} > 0$, and
\begin{equation} \label{formunique4} 
	\phi_j(x) = e^{i \theta_j}    e^{\frac{n}{t_0} V_1(e^{i \theta_j} x)}
	\sum_{k=0}^d (C_{j-1,k} -C_{j,k}) \int_{\overline{\Gamma}_k} dw e^{-\frac{n}{t_0} ( e^{i \theta_j} w x - \overline{V}(w))}.
\end{equation}

Taking $f(x) = x^{(d+1)l+r}$ with non-negative integers $l$ and $r$, we find from \eqref{thetaj} and \eqref{formunique3}
\begin{equation} \label{formunique5} \sum_{j=0}^d \omega_{d+1}^{jr} \int_0^{\infty} x^{(d+1) l + r} e^{- c x^{d+1}}   \phi_j(x)  dx = 0, \qquad
	\omega_{d+1} = e^{\frac{2\pi i}{d+1}}. 
	\end{equation}
Since this holds for every $l =0,1,2 \ldots$, and since the functions $\phi_j$ grow at most
as $O(e^{c_2 x^d}))$ as $ x \to +\infty$ for some $c_2 > 0$, we find from \eqref{formunique5}
that
\[ \sum_{j=0}^d  \omega_{d+1}^{jr} \phi_j(x) \equiv 0, \qquad r = 0, 1, \ldots, d \]
which in turn implies that
\begin{equation} \label{formunique6} 
	\phi_j(x) \equiv 0, \qquad j=0, \ldots, d. 
	\end{equation}
From \eqref{formunique4} we know that the function $x \mapsto \phi_j(x) e^{-\frac{n}{t_0} V_1(e^{i \theta_j} x)}$ is 
a linear combination of the entire functions
\[ x \mapsto \int_{\overline{\Gamma}_k} dw e^{-\frac{n}{t_0} ( e^{i \theta_j} w x - \overline{V}(w))},
\qquad k=0, \ldots, d. \]
The only linear relation between these function is that they add up to zero.
Because of \eqref{formunique4} and \eqref{formunique6} it follows that 
$C_{j-1,k} - C_{j,k}$ is independent of $k$. Summing over $k=0,1, \ldots, d$ and using
the fact that the matrix $C$ has zero row sums, we find that
\[ C_{j-1,k} - C_{j,k} = 0, \qquad \text{for } j,k = 0,1, \ldots, d, \]
i.e., $C$ has constant columns. Since the column sums are zero as well, we get $C = O$ as claimed.
\end{proof}

Now we give the proof of Theorem \ref{theorem1}.

\begin{proof}[Proof of Theorem \ref{theorem1}]
Using the rule that
\[ \overline{\int_{\Gamma} f(z) dz} = \int_{\overline{\Gamma}} \overline{f(\overline{z})} dz \]
it is easy to obtain from \eqref{PhijkV} that
\[ \Phi_{j,k}(f,g) = \overline{\Phi_{k,j}(g,f)}. \]
It follows that \eqref{HformV} with the normalization \eqref{Crowsum}--\eqref{Ccolumnsum}
satisfies the Hermitian condition \eqref{hermitian} if and only if  $C_{j,k} = \overline{C_{k,j}}$ for every $j,k$,
that is, if and only if 
\[ C = C^*	 \qquad \text{($C$ is Hermitian matrix)}. \]
By Lemma \ref{formunique} any sesquilinear form \eqref{HformV} is characterized by a unique
matrix $C$ with zero row and column sums. 
The space of all $(d+1) \times (d+1)$ Hermitian matrices with the zero row/column sum property
is isomorphic to the space of all $d \times d$ Hermitian matrices. Indeed, the restriction of $C$
to the first $d$ rows and columns provides an isomorphism. The dimension of this real vector space is $d^2$.
Hence the real dimension of the vector space of Hermitian forms satisfying
\eqref{structure} and \eqref{hermitian} is at least $d^2$.

To complete the proof we show that the dimension is at most $d^2$.
To that end, we  consider the moments
\[ \mu_{j,k} = \langle z^j, z^k \rangle \]
where $\langle \cdot, \cdot \rangle$ is a Hermitian form satisfying \eqref{structure} and \eqref{hermitian}.
Suppose $V(z) = \sum\limits_{l=1}^{d+1} \frac{t_l}{l} z^l$ with $t_{d+1} \neq 0$. Then \eqref{structure}
with $f(z) = z^j$ and $g(z) = z^k$ implies that
\begin{equation} \label{mustructure} 
	k t_0 \mu_{j,k-1} - n \mu_{j+1,k} + n \sum_{l=0}^d \overline{t_{l+1}} \mu_{j,k+l} = 0.
	\end{equation}
and by \eqref{hermitian}
\begin{equation} \label{muhermitian} 
	\mu_{j,k} = \overline{\mu_{k,j}}.
	\end{equation}
From \eqref{mustructure} and \eqref{muhermitian} it is easy to see that all moments are determined by the moments
$\mu_{j,k}$ with $j,k=0,\ldots, d-1$. This block of moments yields a Hermitian matrix of
size $d \times d$, which is determined by $d^2$ real parameters. Therefore the vector space
 of Hermitian forms satisfying
\eqref{structure} and \eqref{hermitian} is at most $d^2$ dimensional.
\end{proof}

\section{Proof of Theorem \ref{theorem4}}  \label{proofTheo4}.

\subsection{Vector equilibrium problem}

We begin by analyzing the vector equilibrium problem for the energy
functional \eqref{energyE}. For every choice of $x^* > 0$ there is
a unique minimizer $(\mu_1^*, \mu_2^*)$, since $\Sigma_1$ is compact. 
Both measures are symmetric under $2\pi/3$ rotations.

Given $\mu_1^*$, the measure $\mu_2^*$ minimizes the functional
\[ \mu_2 \mapsto I(\mu_2) - I(\mu_1^*, \mu_2) \]
among all measures on $\Sigma_2$ with $\int d\mu_2 = 1/2$. This
means that $\mu_2^*$ is half of the balayage
of $\mu_1^*$ onto $\Sigma_2$. Hence $\mu_2^*$ always has full support, $\supp(\mu_2^*) = \Sigma_2$,
and it is characterized by the property that
\begin{equation} \label{EL2}
	2 \int \log |z-\zeta| d\mu_2^*(\zeta) = \int \log |z-\zeta| d\mu_1^*(\zeta),
	\qquad z \in \Sigma_2. 
	\end{equation}

Given $\mu_2^*$, the measure $\mu_1^*$ minimizes the functional
\[ \mu_1 \mapsto I(\mu_1) - I(\mu_1, \mu_2^*) + \frac{1}{t_0} 
	\int \left( \frac{2}{3 \sqrt{t_3}} |z|^{3/2}  - \frac{t_3}{3} z^3\right) d\mu_1(z)
	\]
among all probability measures on $\Sigma_1$. Thus $\mu_1^*$ is a a minimizer
for an energy functional with external field, see \cite{Dei,SafTot}, and it
is characterized by the condition that there exists a constant $\ell \in \mathbb R$
such that	
\begin{multline} \label{EL1}
	2 \int \log |z-\zeta| d\mu_1^*(\zeta) - \int \log |z-\zeta| d\mu_2^*(\zeta) \\
	-
		 \frac{1}{t_0} \left( \frac{2}{3\sqrt{ t_3}} |z|^{3/2} - \frac{t_3}{3} z^3 \right) 
		 \begin{cases} = \ell, & \quad z \in \supp(\mu_1^*) \\
		 \leq \ell, & \quad z \in \Sigma_1 \setminus \supp(\mu_1^*).
		 \end{cases}
\end{multline} 
The support of $\mu_1^*$ consists of a finite union of intervals, in general. 

The number $x^*$ is at our disposal. We want to choose it an optimal
way as described in the following lemma.
\begin{lemma} \label{lemma41}
	Let $0 < t_0 \leq t_{0,crit}$. Then there is a unique value for
$x^* >0 $ such that
\begin{itemize}
\item $\mu_1^*$ has full support, i.e., $\supp(\mu_1^*) = \Sigma_1$, and
\item the density of $\mu_1^*$ vanishes at the endpoints $\omega^j x^*$, $j=0,1,2$.
\end{itemize}
If $t_0 < t_{0,crit}$ the density of $\mu_1^*$ vanishes like a square root at the endpoints,
while for $t_0 = t_{0,crit}$ it vanishes with an exponent $3/2$.
\end{lemma}

To prove Lemma \ref{lemma41} we first assume that we choose $x^*$
satisfying the conditions in the lemma. 
From that assumption we will find explicit expressions for the measures,
from which we can indeed check that the conditions are satisfied.
The proof of Lemma \ref{lemma41} is given in subsection \ref{subsec:proofLemma41}.
The proof will also give that $x^*$ is given by \eqref{xstar}.

Theorem \ref{theorem4} is proved in subsection \ref{subsec:proofTheorem4}.

\subsection{Riemann surface} \label{subsec:RiemannSurface}
The construction of the measures is based on the consideration of
a three sheeted Riemann surface $\mathcal R$ with sheets
\begin{align} \label{sheets}
	\mathcal R_1 = \overline{\mathbb C} \setminus \Sigma_1, \qquad
	\mathcal R_2 = \overline{\mathbb C} \setminus (\Sigma_1 \cup \Sigma_2), \qquad
	\mathcal R_3 = \overline{\mathbb C} \setminus \Sigma_2.
	\end{align}
For $k=1,2$, the sheet $\mathcal R_k$ is connected to $\mathcal R_{k+1}$ along $\Sigma_k$
in the usual crosswise manner. Then $\mathcal R$ is a compact Riemann surface of
genus zero.

Define the Cauchy transforms of the measures $\mu_1^*$ and $\mu_2^*$
\begin{equation} \label{Cauchytransforms} 
	F_k(z) = \int \frac{d\mu_k^*(\zeta)}{z-\zeta}, \qquad z \in \mathbb C \setminus \Sigma_k, \qquad k=1,2.
	\end{equation}	
These functions have the symmetry property
\[ F_k(\omega z) = \omega^2 F_k(z), \qquad z \in \mathbb C \setminus \Sigma_k, \qquad k=1,2. \]
The conditions \eqref{EL2}--\eqref{EL1} together with the symmetry properties, lead to the following 
relations for the Cauchy transforms 
\begin{equation} \label{ELforF} 
\begin{aligned}
	F_{1,+}(z) + F_{1,-}(z) - F_2(z) & =  \omega^{2j} \frac{1}{t_0} \left(\frac{1}{ \sqrt{t_3}} |z|^{1/2} - t_3 |z|^2\right), &&
		 z \in \Sigma_1, \\ 
	F_{2,+}(z) + F_{2,-}(z) - F_1(z) & = 0, && z \in \Sigma_2.
	\end{aligned}
	\end{equation}
These functions are used to define the $\xi$-functions that play a major role.
Throughout the paper we use the principal arguments of fractional powers, that
is, with a branch cut along the negative real axis. 
 
\begin{definition} 
We define
\begin{align} \label{xi1}
	\xi_1(z) & = t_3 z^2 + t_0 F_1(z), &&  z \in \mathcal R_1, \\
	\xi_2(z) & = \label{xi2} \begin{cases} 
		\frac{1}{\sqrt{t_3}} z^{1/2} + t_0( F_2(z) - F_1(z)), \\
		-\frac{1}{\sqrt{t_3}} z^{1/2} + t_0( F_2(z) - F_1(z)),
		\end{cases} \hspace*{-1cm} &&
		\begin{aligned} & z \in \mathcal R_2 \cap S_0, \\
		  & z \in \mathcal R_2 \cap (S_1 \cup S_2),
		  \end{aligned} \\
	\xi_3(z) & = \label{xi3} \begin{cases}
		 -\frac{1}{\sqrt{t_3}} z^{1/2} - t_0 F_2(z), \\
		 \frac{1}{\sqrt{t_3}} z^{1/2} - t_0 F_2(z),
		 \end{cases} && 
		 \begin{aligned} & z \in \mathcal R_3 \cap S_0, \\
				&  z \in \mathcal R_3 \cap (S_1 \cup S_2).
				\end{aligned}
	\end{align}
where $S_0, S_1, S_2$ denote the sectors
\begin{equation} \label{sectors}
	\begin{aligned} 
		S_0 : & \quad  - \pi/3 < \arg z < \pi/3, \\
		S_1 : & \quad \pi/3 < \arg z < \pi, \\
		S_2 : & \quad -\pi < \arg z < -\pi/3. 
		\end{aligned}
		\end{equation}
\end{definition}

Note that by \eqref{Cauchytransforms} and \eqref{xi1} we can express the densities of the measure
$\mu_1^*$ in terms of $\xi_1$.  Indeed, by \eqref{Cauchytransforms} we have by the Sokhotskii-Plemelj
formula that
\[ d\mu_1^*(x) = \frac{1}{2 \pi i} \left(F_{1,-}(x) - F_{1,+}(x) \right) dx, \qquad x \in \supp(\mu_1^*), \]
which by \eqref{xi1} leads to
\begin{align} \label{densitymu1}
	d\mu_1^*(x) & = \frac{1}{2 \pi i t_0 } \left(\xi_{1,-}(x) - \xi_{1,+}(x) \right) dx, \\
	& = \frac{1}{2 \pi i t_0} \left( \xi_{2,+}(x) - \xi_{1,+}(x) \right) dx, \qquad  x \in \supp(\mu_1^*). 
	\end{align}
Similarly
\begin{align} \label{densitymu2}
	d \mu_2^*(z) =  \frac{1}{2 \pi i t_0} \left( \pm \frac{2z^{1/2}}{\sqrt{t_3}}  + \xi_{3,+}(z) - \xi_{2,+}(z) \right) dz, \qquad  z \in \supp(\mu_2^*),
	\end{align}
	with appropriate choice of signs $\pm$ and square roots.

\begin{lemma} \label{lem:cubic}
Suppose that $\supp(\mu_1^*) = \Sigma_1$ and that the density of $\mu_1^*$ vanishes
at the endpoints $\omega^j x^*$ for $j=0,1,2$. Then
\begin{equation} \label{t0bound} 
	t_0 \leq t_{0,crit} = \frac{1}{8t_3^2}.
	\end{equation}
and the following hold.
\begin{enumerate}
\item[\rm (a)] The three functions $\xi_1, \xi_2, \xi_3$
given by \eqref{xi1}--\eqref{xi3} define a meromorphic function on the Riemann surface $\mathcal R$
whose only poles are at the points at infinity.
\item[\rm (b)]
The functions $\xi_j$, $j=1,2,3$ are the three solutions of the cubic equation
\begin{equation} \label{spectralcurve} 
	\xi^3 - t_3 z^2 \xi^2 - \left(t_0 t_3 + \frac{1}{t_3} \right) z \xi + z^3 + A = 0. 
	\end{equation}
where 
\begin{equation} \label{Aconstant} 
	A =  \frac{1 + 20 t_0 t_3^2 - 8 t_0^2 t_3^4  - (1 - 8 t_0 t_3^2)^{3/2}}{32 t_3^3}.
	\end{equation}
\item[\rm (c)] 
The algebraic equation \eqref{spectralcurve} has branch points at $z = \omega^j x^*$, $j=0,1,2$, where $x^*$ is given by \eqref{xstar}
and nodes at the values $z = \omega^j \widehat{x}$, $j=0,1,2$, where
\begin{equation} \label{xhat} 
	\widehat{x} = \frac{3 + \sqrt{1-8t_0 t_3^2}}{4t_3} > x^*.
	\end{equation}
The values $z = \omega^j x^*$ and $z = \omega^j \widehat{x}$, $j=0,1,2$, are the only zeros of the discriminant
of \eqref{spectralcurve}. 
\item[\rm (d)] We have $\widehat{x} \geq x^* > 0$ with equality only if $t_0 = t_{0,crit}$, and
\begin{equation} \label{widehatx} 
	\xi_1(\omega^j \widehat{x}) = \xi_2(\omega^j \widehat{x}) = \omega^{2j} \widehat{x}, \qquad j=0,1,2. 
\end{equation}
\end{enumerate}  
\end{lemma}
\begin{proof}
The conditions \eqref{ELforF} imply that the three functions $\xi_1$, $\xi_2$, $\xi_3$
given in \eqref{xi1}--\eqref{xi3} define a meromorphic function on the Riemann surface. 
There is a double pole at infinity on the first
sheet, and a simple pole at the other point at infinity.
There are no poles at the endpoints $\omega^j x^*$ due to the assumption that the density 
of $\mu_1^*$ vanishes at these points. This proves part (a).

It follows that any symmetric function 
of the three $\xi_j$ functions is a polynomial in $z$. From \eqref{xi1}--\eqref{xi3} and the fact that
\begin{equation} \label{F12atinfty} 
	F_1(z) = z^{-1} + O(z^{-4}), \qquad F_2(z) = \tfrac{1}{2} z^{-1} + O(z^{-5/2}) 
	\end{equation}
as $z \to \infty$, it is then easy to see that
\begin{equation} \label{symmetricxis}
\begin{cases} 
	\xi_1(z) + \xi_2(z) + \xi_3(z) =  t_3 z^2 , \\
	\xi_1(z) \xi_2(z) + \xi_1(z) \xi_3(z) + \xi_2(z) \xi_3(z)  = - \left(t_0 t_3 + \frac{1}{t_3}\right) z, \\
	\xi_1(z) \xi_2(z) \xi_3(z)  = - z^3 - A,
	\end{cases}
\end{equation}
where $A$ is some real constant.
Thus $\xi_1, \xi_2, \xi_3$ are indeed the solutions of the cubic equation \eqref{spectralcurve}.
The algebraic equation has branch points at $z = \omega^j x^*$, $j=0,1,2$, and no other
finite branch points. 
This property allows us to determine the constant $A$. 
 
The discriminant of \eqref{spectralcurve} with respect to $\xi$ is a cubic equation 
in $\zeta = z^3$ that we calculated with the help of Maple. The result is
\begin{multline} \label{discrim} 	
		4 t_3^3 \zeta^3 + 
		(t_0^2 t_3^4 + 4 A t_3^3 + 12 t_0 t_3^2 - 8) \zeta^2 \\
		+ (4 t_0^3 t_3^3 + 18 A t_0 t_3^2 + 12 t_0^2 t_3 - 36 A  + 12 t_0 t_3^{-1} + 4 t_3^{-3}) \zeta 
		- 27 A^2.
	\end{multline}	
The discriminant \eqref{discrim} has a root for the value
$\zeta^* = (x^*)^3$ that corresponds to the branch points $z= \omega^j x^*$. The fact that
there are no other finite branch points implies that either $\zeta^*$ is a triple root
of \eqref{discrim} or $\zeta^*$ is a simple root and \eqref{discrim} has a double root
as well.  The case of a triple root happens for the values (calculated with Maple)
\begin{equation} \label{critAandcritt0} 
	t_0 = \frac{1}{8 t_3^2}, \qquad A = \frac{27}{256 t_3^3}, \qquad  x^* = \frac{3}{4t_3},
	\end{equation}
and this is the  only possible combination of values with $t_0 > 0$ (one combination with $t_0  < 0$ also gives rise
to a triple root).
The values \eqref{critAandcritt0} are the values for the critical case.

We may from now on assume that \eqref{discrim} has one simple root $\zeta^* > 0$
and one double root.  Since $x^*$ is a branch point that connects the
first and second sheets, we then have that $\xi_2(z) - \xi_1(z)$ vanishes
as a square root as $z \to x^*$. For $x > x^*$, we have that $\xi_2(x) - \xi_1(x)$ is real,
while according to \eqref{densitymu1} we must have $\xi_{2,+}(x) - \xi_{1,+}(x) \in i \mathbb R^+$ for $x \in [0, x^*]$,
since the density of $\mu_1^*$ is positive on $[0, x^*]$.
This implies that $\xi_2(x) - \xi_1(x) > 0$ for $x$ in some interval $(x^*, x^* + \delta)$ to the right of $x^*$.
The definitions \eqref{xi1}--\eqref{xi2} imply that $\xi_1(x) > \xi_2(x)$ for $x > x^*$ large enough.
Therefore there is a value 
\begin{equation} \label{xhatandxstar}
	 \widehat{x} > x^* 
	 \end{equation}
such that $\xi_1(\widehat{x}) = \xi_2(\widehat{x})$.  
Then clearly $\widehat{\zeta} = (\widehat{x})^3$ is
the double root of \eqref{discrim}, and $z= \widehat{x}$, $\xi = \xi_1(\widehat{x})$ is a 
node of the spectral curve \eqref{spectralcurve}.

Because of the symmetry of the cubic equation \eqref{spectralcurve} in the variables $z$ and $\xi$, 
we can interchange the values of $z$ and $\xi$, and
we get that $z = \xi_1(\widehat{x})$, $\xi = \widehat{x}$ is also a node. Thus  $(\xi_1(\widehat{x}))^3$
is also a double root of \eqref{discrim}, which because of the uniqueness of the double root implies 
\eqref{widehatx} for $j=0$. Because of $2\pi/3$ rotational symmetry we also have \eqref{widehatx} for $j=1,2$. 
This proves part (d) of the lemma.

Taking $\xi = z$ in \eqref{spectralcurve} we find from \eqref{widehatx} that $z = \widehat{x}$ is a double root of
the quartic polynomial
\begin{equation} \label{SCondiag} 2 z^3 - t_3z^4 - \left(t_0 t_3+ \frac{1}{t_3}\right) z^2 + A. 
\end{equation}
Thus $\widehat{x}$ is one of the critical points of the polynomial \eqref{SCondiag}, 
which are easily calculated to be
$z=0$ and the two solutions of $ 2 t_3 z^2 - 3z + \left(t_0 t_3+ \frac{1}{t_3}\right) = 0$.
Since $\widehat{x} > x^* > 0$ we discard $z =0$, and find two possible values for $\widehat{x}$ 
\begin{equation} \label{hatx} 
	\widehat{x}_1 = \frac{3 + \sqrt{1-8t_0 t_3^2}}{4t_3}, \qquad
		\widehat{x}_2 = \frac{3 - \sqrt{1-8t_0 t_3^2}}{4t_3}.
	\end{equation}
This value must be real and so we  conclude that $8 t_0 t_3^2 < 1$. This proves \eqref{t0bound}.
	
The corresponding values of $A$ we now find by substituting $z = \widehat{x}_j$  into \eqref{SCondiag}
and equating to $0$. The results are
\begin{align*} 
	A_1 & = \frac{1 +20t_0 t_3^2 - 8t_0^2 t_3^4 - (1-8t_0 t_3^2)^{3/2}}{32t_3^3}, \\
	A_2 & = \frac{1 +20t_0 t_3^2 - 8t_0^2 t_3^4 + (1-8t_0 t_3^2)^{3/2}}{32t_3^3}. 
	\end{align*}
For the value $A = A_j$, the simple root of \eqref{discrim} comes out as $\zeta^* = (x_j^*)^3$ with
\begin{equation} \label{starx}  
	x_1^* = \frac{3}{4t_3} \left( 1 - \sqrt{1-8t_0 t_3^2}\right)^{2/3}, \qquad
	x_2^* = \frac{3}{4t_3} \left( 1 + \sqrt{1-8t_0 t_3^2}\right)^{2/3},
	\end{equation}	 
which is again calculated with Maple.
From \eqref{hatx} and \eqref{starx} it follows that
\[ x_1^* < \frac{3}{4t_3} < \widehat{x}_1, \qquad x_2^* > \frac{3}{4t_3} > \widehat{x}_2. \]
The inequality \eqref{widehatx} is only satisfied in the first case, and so
$x^* = x_1^*$, $A = A_1$, and $\widehat{x} = \widehat{x}_1$. This proves
the formulas  \eqref{xstar}, \eqref{Aconstant}, and \eqref{xhat}, which
establishes the  parts (b) and (c)  and completes the proof of the lemma.
\end{proof}

\subsection{Proof of Lemma \ref{lemma41}} \label{subsec:proofLemma41}

We now prove Lemma \ref{lemma41} by reversing the arguments given above.

For $t_0 \leq t_{0,crit}$, we start from the spectral curve \eqref{spectralcurve}
with the value of $A$ as in \eqref{Aconstant}. The spectral curve defines
a Riemann surface with sheet structure as in \eqref{sheets}. This defines
in particular the value of $x^*$. There is one solution of \eqref{spectralcurve}
that satisfies
\[ \xi_1(z) = t_3 z^2 + t_0 z^{-1} + O(z^{-4}) \qquad \text{as } z \to \infty \]
and $\xi_1$ is defined on the first sheet. The analytic continuations
onto the second and third sheets are then denoted by $\xi_2$ and $\xi_3$,
respectively. 

\begin{figure}
\includegraphics[width=6cm]{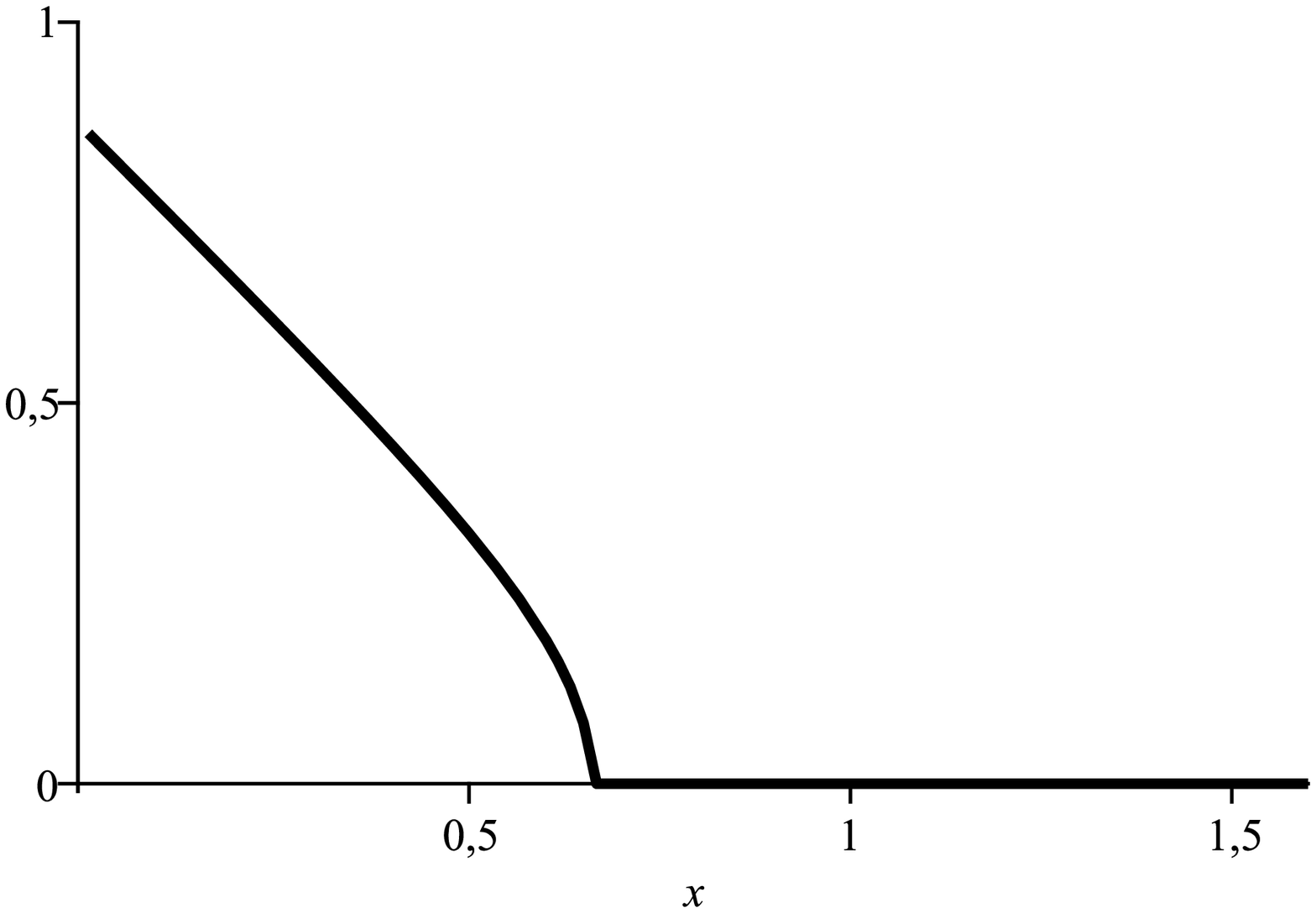} 
\includegraphics[width=6cm]{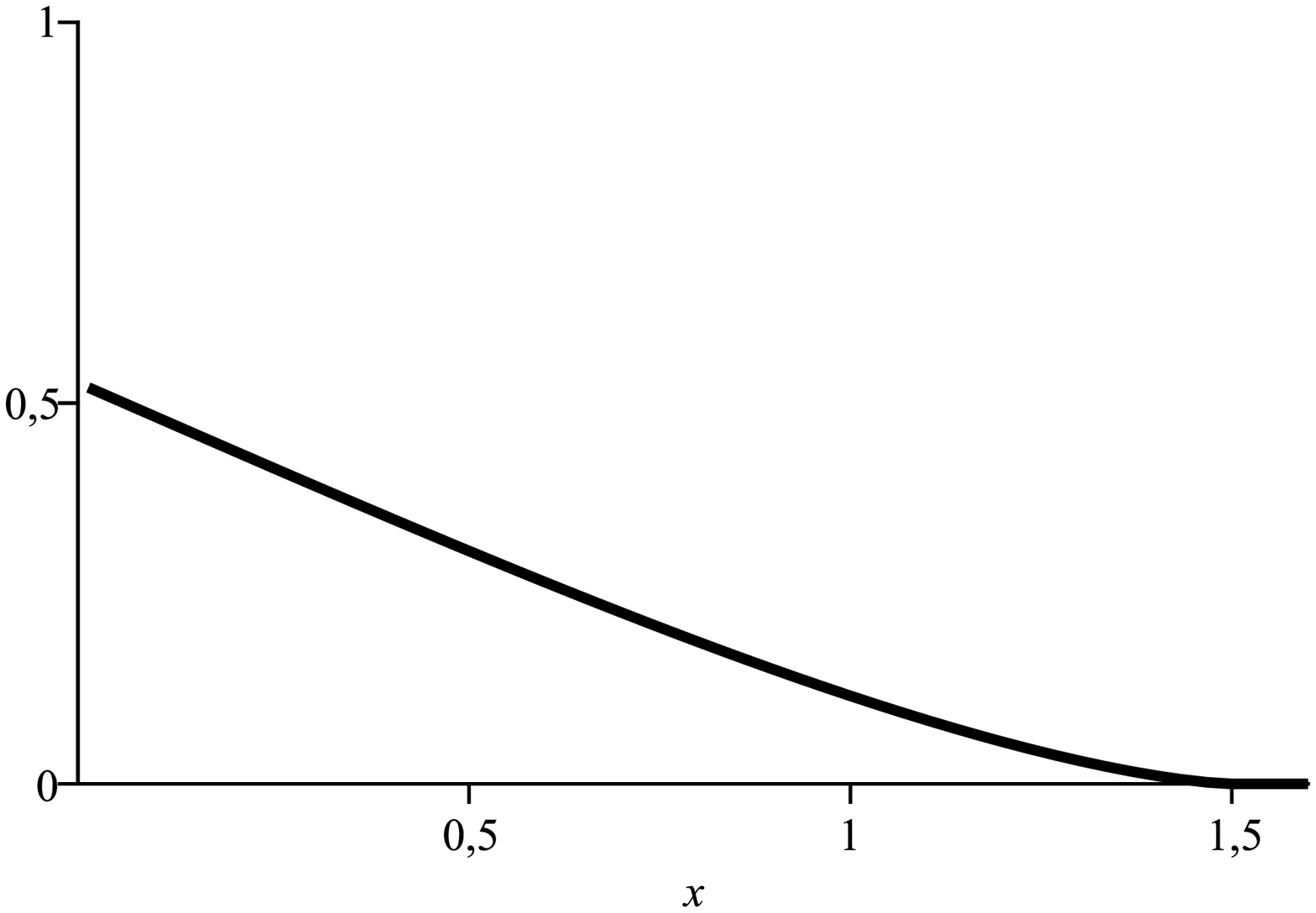}
\caption{Density of $\mu_1^*$ on $[0,x^*]$ for the values $t_0 = 1/4$ and $t_3 =1/2$ (left) 
and for the critical values $t_0=1/2$ and $t_3 =1/2$ (right).
In the non-critical case the density behaves as $c (x^*-x)^{1/2}$ as $x \to x^*$,  while in the
critical case it vanishes like $c(x^*-x)^{3/2}$.  } \label{fig:density}
\end{figure}

Then we define a measure $\mu_j^*$ on $\Sigma_j$ for $j=1,2$, by the formulas 
\eqref{densitymu1} and \eqref{densitymu2}.
All parts of $\Sigma_1$ and $\Sigma_2$ are oriented away from the origin. 
The vector of measures $(\mu_1^*, \mu_2^*)$ is the minimizer of the 
vector equilibrium problem, and $\mu_1^*$ satisfies the conditions of Lemma \ref{lemma41},
since it is positive on $[0, x^*)$ and vanishes at $x^*$, see also Figure \ref{fig:density}
for the density of $\mu_1^*$ for a non-critical value $t_0$ and for the critical value $t_{0,crit}$.
In a non-critical case the density of $\mu_1^*$ vanishes as a square root at $x^*$,
while in the critical case it vanishes like $(x^*-x)^{3/2}$ as $x \to x^*-$. 
See Figure \ref{fig:density} for plots of the density of $\mu_1^*$ in the non-critical
and critical cases.

\subsection{Proof of Theorem \ref{theorem4}} \label{subsec:proofTheorem4}

The spectral curve \eqref{spectralcurve} with $A$ given by \eqref{Aconstant}
has genus zero. The curve has a remarkable rational parametrization
\begin{equation} \label{parametrization}
   z  = h(w) = rw + aw^{-2}, \qquad \xi  = h(w^{-1}) = a w^2 + r w^{-1} 
\end{equation}
with
\begin{equation} \label{ra} 
	r = \frac{\sqrt{1-\sqrt{1-8 t_0 t_3^2}}}{2 t_3}, \qquad
	 a = \frac{1-\sqrt{1-8 t_0 t_3^2}}{4 t_3}. 
	 \end{equation}
This can easily be checked by plugging \eqref{parametrization}--\eqref{ra} into
\eqref{spectralcurve}. The calculations show that it only works if $A$ has the value 
as in \eqref{Aconstant}.

Assuming $t_0 < t_{0, crit}$ we see that $r$ and $a$ are both real and positive.
We view $z=h(w)$ as a mapping from the $w$-plane to the $z$-plane.
The mapping is conformal around infinity. For large enough $\rho$ we have
that the circle $|w| = \rho$ is mapped to a simple closed curve.
This will continue to be the case if we decrease $\rho$ until the circle
$|w| = \rho$ contains a critical point, that is, a solution of
\[ h'(w) = r - 2a w^{-3} = 0. \]
The critical points are on the circle with radius
\[ \rho_{crit} = \left(\frac{2a}{r}\right)^{1/3} = \left(1-\sqrt{1-8 t_0 t_3^2} \right)^{1/6} < 1. \]

Since $\rho_{crit} < 1$, we see that $h(|w|=1)$ is therefore a simple closed curve
which is smooth for $t_0 < t_{0,crit}$ and has three cusp points for the
critical value $t_{0,crit}$. Let $\Omega$ be the bounded domain that
is enclosed by $h(|w|=1)$. 

For $z \in \partial \Omega$, we have $z = h(w)$ for some $w$ with $|w| = 1$.
Then $\xi_1(z) = h(1/w)$, see \eqref{parametrization}, and so since $1/w = \overline{w}$,
and since the coefficient of $h$ are real numbers 
\begin{equation} \label{xi1Schwarz} 
	\xi_1(z) = h(\overline{w}) = \overline{h(w)} = \overline{z}, \qquad z \in \partial \Omega. 
	\end{equation}
This shows that \eqref{Schwarz} indeed defines a curve $\partial \Omega$, 
by virtue of  \eqref{xi1} and \eqref{Cauchytransforms}.

We compute the area of $\Omega$ by means of Green's formula
\[ \area(\Omega) =  \iint_\Omega dA = \frac{1}{2 i} \oint_{\partial \Omega} \overline{z} \, dz =
	\frac{1}{2i} \oint_{\partial \Omega} \xi_1(z) \, dz.  \]
Here we used \eqref{xi1Schwarz}. Now recall that $\xi_1$ is analytic
in the exterior of $\Omega$ with 
\begin{equation} \label{xi1atinfty} 
	\xi_1(z) = t_3 z^2 + t_0 z^{-1} + O(z^{-4}), \qquad \text{as } z \to \infty, 
	\end{equation}
see \eqref{xi1} and \eqref{F12atinfty}.
Thus we can move the contour to infinity and by doing
so we only pick up a residue contribution at $\infty$. This proves 
the formula for $t_0$ in \eqref{harmonicmoments}.

The exterior harmonic moments \eqref{harmonicmoments} are computed  in a similar way.
We have by Green's formula applied to the exterior domain
\[ - \frac{1}{\pi} \iint_{\mathbb C \setminus \overline{\Omega}} \frac{d A(z)}{z^k} = 
	\frac{1}{2\pi i} \oint_{\partial \Omega} \frac{\overline{z}}{z^k} \, dz
	= \frac{1}{2\pi i} \oint_{\partial \Omega} \frac{\xi_1(z)}{z^k} \, dz. \]
By contour deformation and using \eqref{xi1atinfty} we find that this
is $t_3$ for $k=3$ and $0$ for other $k \geq 2$. Thus \eqref{harmonicmoments} follows.

We finally prove \eqref{potential} with yet again similar arguments.
For $z \in \mathbb C \setminus \overline{\Omega}$ we have by Green's formula and \eqref{xi1Schwarz}
\[ \frac{1}{\pi} \iint_{\Omega} \frac{d A(\zeta)}{z-\zeta} = 
	\frac{1}{2\pi i} \oint_{\partial \Omega} \frac{\overline{\zeta}}{z-\zeta} d\zeta
		= 
	\frac{1}{2\pi i} \oint_{\partial \Omega} \frac{\xi_1(\zeta)}{z-\zeta} d\zeta
	\]
Moving the contour to infinity we pick up a residue contribution at $\zeta = z$, which is $\xi_1(z)$,
and at $\zeta = \infty$, which is $-t_3 z^2$ because of \eqref{xi1atinfty}.
In total we find that
\[ \frac{1}{\pi} \iint_{\Omega} \frac{d A(\zeta)}{z-\zeta} = \xi_1(z) - t_3 z^2 \]
which is equal to $t_0 F_1(z)$ by \eqref{xi1}. Then dividing by $t_0$ we obtain
\eqref{potential}. This completes the proof of Theorem \ref{theorem4}.

\section{Preliminary steps towards the proofs of Theorems \ref{theorem2} and \ref{theorem3}}
\label{preliminary}

\subsection{Discussion} \label{subsec:discussion}
The orthogonality induced by the Hermitian form \eqref{Hform3} is similar to the biorthogonality
that plays a role in the two-matrix model, see \cite{Ber}. The two-matrix model is a model for two random matrices
with two potentials $V$ and $W$ on $\mathbb R$ with sufficient increase at $\pm \infty$. 
The biorthogonal polynomials in this model 
are two sequences $(p_{j,n}(x))_j$, $\deg p_{j,n} = j$ and $(q_{k,n})_k$, $\deg q_{k,n} = k$ of monic polynomials satisfying
\begin{equation} \label{biorthogonal} 
	\int_{-\infty}^{\infty} \int_{-\infty}^{\infty} p_{j,n}(x) q_{k,n}(y) e^{-n (V(x) + W(y) - \tau xy)} dx dy = h_{k} \delta_{j,k} 
	\end{equation}
where $\tau $ is the coupling constant.
Comparing \eqref{biorthogonal} with \eqref{Hform3} one sees that \eqref{Hform3} is like the biorthogonality
in the two matrix model with equal cubic potentials $V(x) = W(x) = - \frac{t_3}{3 t_0} x^3$ and coupling
constant $\tau = - \frac{1}{t_0}$. The main difference is that integrals in \eqref{biorthogonal} are over
the real line, while integrals in \eqref{Hform3} are over combinations of the contours $\Gamma_j$, $j=0,1,2$.

This does not play a role on a formal level. The main algebraic properties that are known 
for the biorthogonal polynomials, see e.g.\ \cite{Ber, BeEyHa1, BeEyHa2},  also hold for the orthogonal polynomials
with respect to the Hermitian form \eqref{Hform3}. This includes the existence of differential and
difference equations, and an integrable structure of $\tau$-functions and Toda equations. 
However, analytic and asymptotic properties depend crucially on the precise definition of
contours. For example, it is known that the
biorthogonal polynomials characterized by \eqref{biorthogonal} have real and simple zeros  \cite{ErcMcL}, 
which is not the case for the orthogonal polynomials for \eqref{Hform3}.

Recently \cite{DuiKui,DuKuMo,Mo}, the biorthogonal polynomials $p_{n,n}$ from \eqref{biorthogonal} with an even polynomial $V$ and 
a quartic potential $W(y) = \frac{1}{4} y^4 + \frac{\alpha}{2} y^2$ were successfully analyzed in the large $n$ limit
by means of a steepest descent analysis of a Riemann-Hilbert  problem.  The Riemann-Hilbert problem 
was obtained earlier in \cite{KuiMcL} from a reformulation of the biorthogonality \eqref{biorthogonal} as
multiple orthogonality, since for multiple orthogonal polynomials a Riemann-Hilbert problem is known \cite{VAGeKu}
as a generalization of the well-known Riemann-Hilbert problem
for orthogonal polynomials \cite{FokItsKit}. 
Because of the formal similarity the same setup works in the case of orthogonal polynomials for the Hermitian
form \eqref{Hform3}. We can reformulate the orthogonality as multiple orthogonality and it leads to
a $3 \times 3$ matrix valued Riemann-Hilbert problem with jumps on the contours $\Gamma_j$.
In order to prepare for the asymptotic analysis we first adjust the contours in a suitable way.
This will be done in this section.

The systematic asymptotic analysis of Riemann-Hilbert problems is due to Deift and Zhou \cite{DeiZho} who 
developed their steepest descent analysis first in the context of  large time asymptotics of integrable systems. 
It was applied to orthogonal polynomials and random matrices in \cite{BleIts,DKMVZ1,DKMVZ2}. These papers
also emphasized the use of equilibrium measures in the asymptotic analysis. The Riemann-Hilbert problem
for orthogonal polynomials \cite{FokItsKit} is of size $2 \times 2$. Extensions of the steepest
descent analysis to larger size Riemann-Hilbert problems were first discussed in \cite{BleKui,KVAW}.
In the next section we will build on these and later works (see \cite{Kui} for an overview) 
and apply the Deift/Zhou steepest descent analysis
to the $3 \times 3$ matrix RH problem \ref{RHforY2}. A crucial role is played by the minimizer 
$(\mu_1^*, \mu_2^*)$ of the vector equilibrium problem.
This is  inspired by \cite{BlDeKu,DuiKui,DuKuMo} where the steepest descent analysis also depended
crucially on a vector equilibrium problem.

\subsection{Multiple orthogonality and Airy functions}

The orthogonality with respect to the Hermitian form \eqref{Hform3} is very 
similar to the biorthogonality 
that plays a role in the two-matrix model. We will use ideas that were developed
for the asymptotic analysis of the two-matrix model  with a quartic potential \cite{DuiKui, DuKuMo}
and apply these to the
orthogonal polynomials for the Hermitian form \eqref{Hform3}.

First of all we identify the orthogonal polynomials as multiple orthogonal polynomials.
We define 	for $j=0,1,2$ and $n \in \mathbb N$  the entire functions 
\begin{align} \label{weightw0jn} 
	w_{0,j,n}(z) & = \frac{1}{2\pi i} \sum_{k=0}^2 \epsilon_{j,k} \int_{\overline{\Gamma}_k} 
	 e^{-\frac{n}{t_0}(wz - \frac{t_3}{3} (w^3 + z^3))} dw, \\ \label{weightw1jn}
	 w_{1,j,n}(z) & = \frac{1}{2\pi i} \sum_{k=0}^2 \epsilon_{j,k} \int_{\overline{\Gamma}_k}
	 w e^{-\frac{n}{t_0}(wz - \frac{t_3}{3} (w^3 + z^3))} dw. 
	 \end{align}

\begin{lemma}
The monic orthogonal polynomial $P_{n,n}$ is characterized by the properties
\begin{equation} \label{MOPrelations}
\begin{aligned} 
	 \sum_{j=0}^2 \int_{\Gamma_j} P_{n,n}(z) z^k w_{0,j,n}(z) dz & = 0, \qquad k = 0,1, \ldots, \lceil\tfrac{n}{2}\rceil - 1, \\
	 \sum_{j=0}^2 \int_{\Gamma_j} P_{n,n}(z) z^k w_{1,j,n}(z) dz & = 0, \qquad k=  0,1, \ldots, \lfloor\tfrac{n}{2}\rfloor -1.
	 \end{aligned}
	 \end{equation}
\end{lemma}
\begin{proof}
This follows as in  \cite{KuiMcL}, but for convenience to the reader we present the argument.
Recall that $P_{n,n}$ satisfies the orthogonality conditions \eqref{orthogonal} with $k=n$.

We prove that for every $k$,
\begin{equation} \label{MOPrelations1} 
	\langle z^k P_{n,n}(z), z^j \rangle = 0, \qquad j = 0, \ldots, n-2k-1.
\end{equation} 
For $k = 0$, the condition \eqref{MOPrelations1} reduces to the orthogonality
condition \eqref{orthogonal}.

Assume that \eqref{MOPrelations1} holds for certain $k \geq 0$.
The structure relation \eqref{structure} with cubic potential \eqref{cubicV2} gives  
\begin{equation} \label{structurecubic} 
	n \langle z f, g \rangle = t_0 \langle f, g' \rangle + n t_3 \langle f, z^2 g \rangle.
	\end{equation}
Taking $f(z) = z^k P_{n,n}(z)$ and $g(z) = z^j$	
we find that both terms in the right-hand side of \eqref{structurecubic} vanish 
for $j + 2 \leq n - 2k-1 $, because of the induction hypothesis. Then the
left-hand side vanishes as well, and this gives \eqref{MOPrelations1}
with $k+1$. Thus \eqref{MOPrelations1} follows by induction.

Taking $j=0$ and $j=1$ in \eqref{MOPrelations1}, we find  
\begin{equation} \label{MOPrelations2} 
	 \begin{aligned}
	\langle z^k P_{n,n}, 1 \rangle = 0,  & \qquad  k=0, \ldots, \lceil \tfrac{n}{2} \rceil - 1, \\
	\langle z^k P_{n,n}, z \rangle = 0,  & \qquad k=0, \ldots, \lfloor \tfrac{n}{2} \rfloor -1,
	\end{aligned}
	\end{equation}
and these conditions are equal to the conditions \eqref{MOPrelations} because of the representation
\eqref{Hform3} of the Hermitian form and the definition \eqref{weightw0jn}--\eqref{weightw1jn}
of the weight functions.

The conditions \eqref{MOPrelations} in fact characterize the monic orthogonal polynomial
since we can similarly show that the conditions \eqref{MOPrelations2}
are in fact equivalent to the orthogonality conditions \eqref{orthogonal}
(again by using the structure relation \eqref{structurecubic}).
\end{proof}

The functions \eqref{weightw0jn} and \eqref{weightw1jn} can be expressed in terms
of Airy functions and their derivatives. The classical Airy differential equation  $y'' = zy$ 
has the three solutions
\begin{equation} \label{y0y1y2}
\begin{aligned} 
	y_0(z) & = \Ai(z) = \frac{1}{2\pi i} \int_{\Gamma_0} e^{\frac{1}{3}s^3 - zs} ds
		= - \frac{1}{2\pi i} \int_{\overline{\Gamma}_0} e^{\frac{1}{3} s^3 - zs} ds, \\
	y_1(z) & = \omega \Ai(\omega z)  =  \frac{1}{2\pi i} \int_{\Gamma_1} e^{\frac{1}{3}s^3 - zs} ds
		= - \frac{1}{2\pi i} \int_{\overline{\Gamma}_2} e^{\frac{1}{3} s^3 - zs} ds, \\
  y_2(z) & = \omega^2 \Ai(\omega^2 z)  =  \frac{1}{2\pi i} \int_{\Gamma_2} e^{\frac{1}{3}s^3 - zs} ds
		= - \frac{1}{2\pi i} \int_{\overline{\Gamma}_1} e^{\frac{1}{3} s^3 - zs} ds,
  \end{aligned}
\end{equation} 
that are related by the identity $y_0 + y_1 + y_2 = 0$.
Then we get from \eqref{weightw0jn} and \eqref{y0y1y2} after  the change of variables $w = \left(\frac{t_0}{nt_3}\right)^{1/3} s$, 
\begin{equation} \label{Airyw0jn}
\begin{aligned} 
	  w_{0,0,n}(z) & = d_n (y_2(c_n z) - y_1(c_n z)) e^{\frac{nt_3}{3t_0} z^3}, \\
	  w_{0,1,n}(z) & = d_n (y_1(c_n z) - y_0(c_n z)) e^{\frac{nt_3}{3t_0} z^3}, \\
	  w_{0,2,n}(z) & = d_n (y_0(c_n z) - y_2(c_n z)) e^{\frac{nt_3}{3t_0} z^3},
\end{aligned}
\end{equation}
with constants 
\begin{equation} \label{constantcn} 
	c_n = \frac{n^{2/3}}{t_0^{2/3} t_3^{1/3}} > 0, \qquad d_n = \left(\frac{t_0}{nt_3} \right)^{1/3}.
	\end{equation}
Similarly from \eqref{weightw1jn} and \eqref{y0y1y2}
\begin{equation} \label{Airyw1jn}
\begin{aligned} 
	  w_{1,0,n}(z) & = -d_n^2 (y_2'(c_n z) - y_1'(c_n z)) e^{\frac{nt_3}{3t_0} z^3}, \\
	  w_{1,1,n}(z) & = -d_n^2 (y_1'(c_n z) - y_0'(c_n z)) e^{\frac{nt_3}{3t_0} z^3}, \\
	  w_{1,2,n}(z) & = -d_n^2 (y_0'(c_n z) - y_2'(c_n z)) e^{\frac{nt_3}{3t_0} z^3}.
\end{aligned}
\end{equation}

\subsection{Riemann-Hilbert problem} 
Observe  that there is no complex conjugation in the multiple orthogonality
conditions \eqref{MOPrelations}.
Following \cite{VAGeKu} we then find a characterization of $P_{n,n}$ in terms of 
a  $3 \times 3$ matrix valued Riemann-Hilbert (RH) problem. We assume that the 
contours $\Gamma_j$ are disjoint,
and have the orientation as shown in Figure \ref{fig:contoursGamma}.
The orientation induces a $+$ and $-$ side on $\Gamma_j$, where the $+$ side ($-$ side)
is on the left (right) while traversing the contour according to its orientation.

\begin{rhproblem} \label{RHforY1}
Let $\Gamma = \bigcup_{j=0}^2 \Gamma_j$. 
We look for $Y : \mathbb C \setminus \Gamma \to \mathbb C^{3\times 3}$ satisfying
\begin{itemize}
\item $Y$ is analytic in $\mathbb C \setminus \Gamma$,
\item $Y_+(z) = Y_-(z) J_Y(z)$ for $z \in \Gamma$ with jump matrix $J_Y$ given by
\begin{equation} \label{JYdisjoint} 
	J_Y = \begin{pmatrix} 1 & w_{0,j,n} & w_{1,j,n} \\ 0 & 1 & 0 \\ 0 & 0 & 1 \end{pmatrix} \qquad \text{on } \Gamma_j, 
	\end{equation}
	(here $Y_+$ and $Y_-$ denote the limiting values of $Y$, when approaching $\Gamma$ on
	the $+$ side and $-$ side, respectively),
\item $Y(z) = (I + O(1/z)) \diag\left(z^n, z^{-\lceil n/2 \rceil}, z^{-\lfloor n/2 \rfloor} \right)$ as $z \to \infty$.	
\end{itemize}
\end{rhproblem}
The jump condition \eqref{JYdisjoint} has to be adjusted in the case of overlapping contours.

Provided that the orthogonal polynomial $P_{n,n}$ uniquely exists, the RH problem \ref{RHforY1}
has a unique solution, see \cite{VAGeKu}. The first row of $Y$ is given by
\begin{align*} 
	Y_{1,1}(z) & = P_{n,n}(z), \\
	Y_{1,2}(z) & = \frac{1}{2\pi i} \sum_{j=0}^2 \int_{\Gamma_j} \frac{P_{n,n}(x) w_{0,j,n}(x)}{x-z} \, dx, \\
	Y_{1,3}(z) & = \frac{1}{2\pi i} \sum_{j=0}^2 \int_{\Gamma_j} \frac{P_{n,n}(x) w_{1,j,n}(x)}{x-z} \, dx. 
	\end{align*}
The remaining two rows of $Y$ are built in a similar way out of certain polynomials of degrees $n-1$
(one of which is proportional to $P_{n-1,n}$, if this orthogonal polynomial exists.)

\subsection{Deformation of contours} \label{subsec:deformation}
In order to prepare for the steepest descent analysis of the RH problem, we first adjust
the contours $\Gamma_j$. Because the weights \eqref{Airyw0jn}, \eqref{Airyw1jn}
in the multiple orthogonality conditions \eqref{MOPrelations} are entire
functions, we have the freedom to make arbitrary deformation for each contour $\Gamma_j$, 
as long as the contours start and end at the same asymptotic angles at infinity. 
We are going to deform the contours such that their union contains
the set $\Sigma_1 = \bigcup_{j=0}^2 [0, \omega^j x^*]$ where the zeros of the polynomials 
are going to accumulate. The deformed contours are shown in Figure \ref{fig:deformedGamma}.

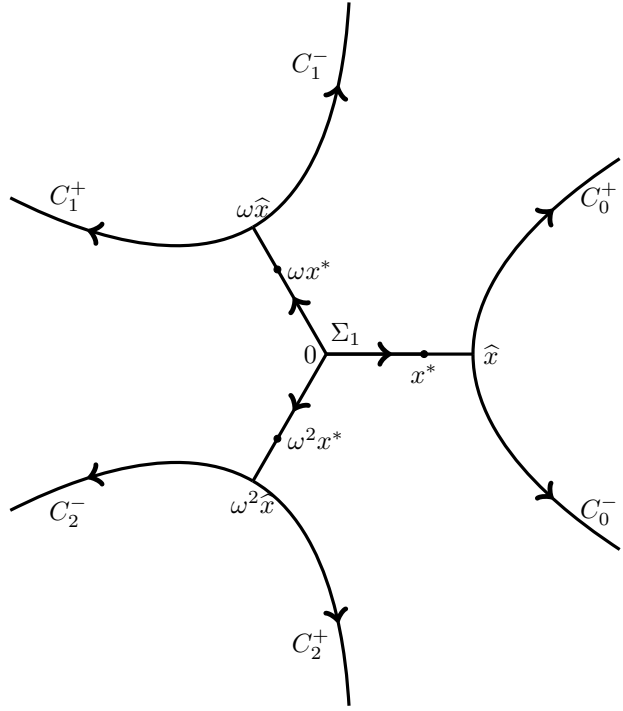
\begin{figure}[t]
\begin{center}
\begin{tikzpicture}[scale=1.3,decoration={markings,mark=at position .67 with {\arrow[black,line width=0.8mm]{>};}}]

\draw[postaction={decorate}, very thick] (0,0)--(1,0) node[below]{$x^*$};
\draw[very thick] (0,0)--(1.5,0) node[right]{$\widehat{x}$};
\draw[postaction={decorate}, very thick, rotate around={120:(0,0)}] (0,0)--(1,0) node[right]{$\omega x^*$};
\draw[very thick, rotate around ={120:(0,0)}] (1,0)--(1.5,0) node[above]{$\omega \widehat{x}$};
\draw[postaction={decorate}, very thick, rotate around={-120:(0,0)}] (0,0)--(1,0) node[right]{$\omega^2 x^*$};
\draw[very thick, rotate around ={-120:(0,0)}] (1,0)--(1.5,0) node[below]{$\omega^2 \widehat{x}$};

\draw[postaction={decorate},very thick,rotate around={-90:(1.5,0)}] (1.5,0) parabola (3.5,1.5);
\draw[postaction={decorate},very thick,rotate around={90:(1.5,0)}] (1.5,0) parabola (3.5,-1.5);

\draw[postaction={decorate},very thick,rotate around={30:(-0.75,1.3)}] (-0.75,1.3) parabola (1.25,2.8);
\draw[postaction={decorate},very thick,rotate around={210:(-0.75,1.3)}] (-0.75,1.3) parabola (1.25,-0.2);
\draw[postaction={decorate},very thick,rotate around={-210:(-0.75,-1.3)}] (-0.75,-1.3) parabola (1.25,0.2);
\draw[postaction={decorate},very thick,rotate around={-30:(-0.75,-1.3)}] (-0.75,-1.3) parabola (1.25,-2.8);

\draw(2.5,1.6) node[right]{$C_0^+$}; 
\draw(2.5,-1.6) node[right]{$C_0^-$};

\draw[rotate around={120:(0,0)}] (2.5,1.6) node[above]{$C_1^+$}; 
\draw[rotate around={120:(0,0)}] (2.5,-1.6) node[left]{$C_1^-$};

\draw[rotate around={-120:(0,0)}] (2.5,1.6) node[left]{$C_2^+$}; 
\draw[rotate around={-120:(0,0)}] (2.5,-1.6) node[below]{$C_2^-$};

\draw(0.2,0) node[above]{$\Sigma_1$};
\draw(0,0) node[left]{$0$};

\filldraw(1,0) circle (1pt);
\filldraw[rotate around={120:(0,0)}](1,0) circle (1pt);
\filldraw[rotate around={-120:(0,0)}](1,0) circle (1pt);
\end{tikzpicture}
\end{center}
\caption{Deformed contours $\Gamma_0$, $\Gamma_1$, $\Gamma_2$. All contours are oriented away from~$0$.}
\label{fig:deformedGamma}
\end{figure}

We recall that $\widehat{x} > x^*$, see \eqref{xhat}. We deform $\Gamma_0$ to contain the intervals $[0,\omega^2 \widehat{x}]$ 
and $[0, \omega \widehat{x}]$  
and we continue $\Gamma_0$ with an unbounded contour from $\omega \widehat{x}$ to infinity at asymptotic angle $5 \pi/12$,
and its mirror image in the real axis which goes from $\omega^2 \widehat{x}$ to infinity at angle $-5 \pi/12$.
The number $5 \pi/12$ could be replaced with any number between $\pi/3$ and $\pi/2$.  
We obtain $\Gamma_1$ and $\Gamma_2$ by rotating $\Gamma_0$
over angles $2\pi/3$ and $-2\pi/3$, respectively. 

It will be convenient to give new names to the unbounded contours. We define
\begin{equation} \label{contoursCjpm}
\begin{aligned}
	 C_0^+  & = (\Gamma_1 \cap S_0) \setminus [0, \widehat{x}], &  C_0^- & = (\Gamma_2 \cap S_0) \setminus [0, \widehat{x}], \\
	 C_1^+  & = (\Gamma_2 \cap S_1) \setminus [0, \omega \widehat{x}], &   C_1^- & = (\Gamma_0 \cap S_1) \setminus [0, \omega \widehat{x}], \\
	 C_2^+  & = (\Gamma_0 \cap S_2) \setminus [0, \omega^2 \widehat{x}], \qquad &   C_2^- & = (\Gamma_1 \cap S_2) \setminus [0, \omega^2 \widehat{x}],
\end{aligned}
\end{equation}
and
\begin{equation} \label{contoursCj}
	 C_j = C_j^+ \cup C_j^-, \qquad j=0,1,2.
	 \end{equation}
Then, see also Figure~\ref{fig:deformedGamma},
\begin{equation} \label{contoursGammaandC}
	\begin{aligned}
		\Gamma_0 & = C_2^+ \cup [0, \omega^2 \widehat{x}] \cup [0,\omega \widehat{x}] \cup C_1^-, \\
		\Gamma_1 & = C_0^+ \cup [0,\widehat{x}] \cup [0,\omega^2 \widehat{x}] \cup C_2^-, \\
		\Gamma_2 & = C_1^+ \cup [0,\omega \widehat{x}] \cup [0, \widehat{x}] \cup C_0^-, 
	\end{aligned}
\end{equation}	
and
\begin{equation} 
		\Gamma  = \bigcup_{j=0}^2 \left( [0,\omega^j \widehat{x}] \cup C_j \right).
\end{equation}
We redefine the orientation on the contours such that all parts are from now on oriented away from $0$ and towards  
$\infty$, as shown in Figure~\ref{fig:deformedGamma}. 
We have a freedom in the
choice of the precise location of the 
contours $C_j$ and we make use of this freedom later on. 

After this deformation of contours and partial reversion of orientation we find new expressions
for the weights. Indeed, for $x \in [0, \widehat{x}]$, we find by \eqref{Airyw0jn}
\begin{multline*} 
	- w_{0,1,n}(x) + w_{0,2,n}(x) \\
	= d_n \left(-(y_1(c_n x) - y_0(c_n x)) + (y_0(c_n x)-y_2(c_nx))\right) e^{\frac{n t_3}{3t_0} x^3} 
	\end{multline*}
Using the identity $y_0 + y_1 + y_2 =  0$ and $y_0 = \Ai$, this reduces to
\begin{equation} \label{combinedweight0} 
	-w_{0,1,n}(x) + w_{0,2,n}(x) = 3 d_n \Ai(c_n x) e^{\frac{nt_3}{3t_0} x^3}. 
	\end{equation}
Similarly we get from \eqref{Airyw1jn}
\begin{equation} \label{combinedweight1} 
	-w_{1,1,n}(x) + w_{1,2,n}(x)  = - 3 d_n^2 \Ai'(c_n x) e^{\frac{nt_3}{3t_0} x^3}.
	\end{equation}

It is thanks to our choice \eqref{Cmatrix} for the matrix $C$, that we obtain
the Airy function $\Ai$ and its derivative in \eqref{combinedweight0}--\eqref{combinedweight1},
and not some other solution of the Airy differential equation.
The Airy function $\Ai$ is  a special solution of $y'' = zy$ because of its
asymptotic behavior, see \cite[formulas 10.4.59, 10.4.61]{AbrSte},
\begin{equation} 	\label{Airyasymptotics}
	\begin{aligned} 
	\Ai(z) &= \frac{z^{-1/4}}{2\sqrt{\pi}} e^{-\frac{2}{3} z^{3/2}} (1 + O(z^{-3/2})), \\
	\Ai'(z) &= -\frac{z^{1/4}}{2\sqrt{\pi}} e^{-\frac{2}{3} z^{3/2}} (1 + O(z^{-3/2})),   
	\end{aligned}
	\end{equation}
as $z \to \infty$, $- \pi < \arg z < \pi$, which is decaying for $z$ going to infinity
on the positive real axis. Any other linearly independent solution of the Airy differential
equation increases on the positive real axis. We use \eqref{Airyasymptotics} in 
the steepest descent analysis of the RH problem, and this explains why we chose the matrix $C$
as we did in \eqref{Cmatrix}.

We now redefine the weights on the new contour $\Gamma$ by combining the
contributions of overlapping contours. We also rescale the weights by dropping
the irrelevant constant prefactor $3 d_n$ in \eqref{combinedweight0} and $-3d_n^2$ in \eqref{combinedweight1}.

\begin{definition}
Let $\Gamma = \bigcup_{j=0}^2 \left([0, \omega_j \widehat{x}] \cup C_j\right)$ be as in Figure \ref{fig:deformedGamma}.
We define two functions $w_{0,n}$ and $w_{1,n}$ on $\Gamma$ by
first defining them in the sector $S_0 : |\arg z| < \pi/3$ by putting
\begin{align} \label{weightsdef1}
	& \begin{cases}
	w_{0,n}(x) = \Ai(c_n x) e^{\frac{nt_3}{3t_0} x^3}, \\
	w_{1,n}(x) = \Ai'(c_n x) e^{\frac{nt_3}{3t_0} x^3},
	\end{cases}
	&&  \text{for } x \in [0, \widehat{x}], \\
	 \label{weightsdef2}
	& \begin{cases}
	w_{0,n}(z) = \frac{1}{3}(y_0(c_n z) - y_1(c_n z)) e^{\frac{n t_3}{3t_0} z^3}, \\
	w_{1,n}(z) = \frac{1}{3}(y_0'(c_n z) - y_1'(c_n z)) e^{\frac{n t_3}{3t_0} z^3},
	\end{cases}
	&& \text{for } z \in C_0^+, \\
	 \label{weightsdef3}
	& \begin{cases}
	w_{0,n}(z) = \frac{1}{3}(y_0(c_n z) - y_2(c_n z)) e^{\frac{n t_3}{3t_0} z^3}, \\
	w_{1,n}(z) = \frac{1}{3}(y_0'(c_n z) - y_2'(c_n z)) e^{\frac{n t_3}{3t_0} z^3}, 
	\end{cases}
	&& \text{for } z \in C_0^-,
	\end{align}
and then by extending to the  parts of $\Gamma$ in the other sectors by the property
\begin{equation} \label{weightsdef4}
	\begin{cases} 
	w_{0,n}(\omega z) = \omega^2 w_{0,n}(z), \\
	w_{1,n}(\omega z) = \omega w_{1,n}(z), 
	\end{cases}
	 \qquad \text{for } z \in \Gamma.
	\end{equation} 
\end{definition}

With the definition \eqref{weightsdef1} we have by \eqref{Airyasymptotics} that 
for $x \in (0,\widehat{x}]$,
\begin{equation} \label{w0w1asymptotics}
	\begin{cases} \ds w_{0,n}(x) = \frac{(c_n x)^{-1/4}}{2 \sqrt{\pi}}
	  \exp\left(- \frac{n}{t_0} \left(\frac{2}{3 \sqrt{t_3}} x^{3/2} - \frac{t_3}{3} x^3\right) \right) \left(1 + O(n^{-1})\right),
	\\
		\ds w_{1,n}(x) = - \frac{(c_n x)^{1/4}}{2 \sqrt{\pi}}
	  \exp\left(- \frac{n}{t_0} \left(\frac{2}{3 \sqrt{t_3}} x^{3/2} - \frac{t_3}{3} x^3\right) \right) \left(1 + O(n^{-1})\right),
	  \end{cases}
\end{equation}
as $n \to \infty$. The exponential part in the asymptotic behavior \eqref{w0w1asymptotics} of the weight functions 
is reflected in the term
\[ \frac{1}{t_0} \int \left(  \frac{2}{3 \sqrt{t_3}} |x|^{3/2} - \frac{t_3}{3} x^3\right) d\mu_1(x) \]
that appears in the energy functional \eqref{energyE}.

\subsection{Riemann-Hilbert problem after deformation}

After the contour deformation we find that the multiple orthogonality conditions \eqref{MOPrelations}
satisfied by the  monic orthogonal polynomial can also be stated as 
\begin{equation} \label{MOPrelationsdeformed}
	\begin{aligned}
	\int_{\Gamma} P_{n,n}(z) z^k w_{0,n}(z) dz &= 0, \qquad k = 0, \ldots, \lceil \tfrac{n}{2} \rceil - 1, \\
	\int_{\Gamma} P_{n,n}(z) z^k w_{1,n}(z) dz &= 0, \qquad k = 0, \ldots, \lfloor \tfrac{n}{2} \rfloor - 1.
	\end{aligned}
\end{equation}
The corresponding RH problem is as follows.
\begin{rhproblem} \label{RHforY2} We look for $Y : \mathbb C \setminus \Gamma \to \mathbb C^{3 \times 3}$ satisfying
\begin{itemize}
\item $Y : \mathbb C \setminus \Gamma \to \mathbb C^{3 \times 3}$ is analytic,
\item $Y_+(z) = Y_-(z) \begin{pmatrix} 1 & w_{0,n}(z) & w_{1,n}(z) \\ 0 & 1  & 0 \\ 0 & 0 & 1 \end{pmatrix}$ for $z \in \Gamma$
\item $Y(z) = (I + O(1/z)) \diag\left(z^n, z^{-\lceil n/2 \rceil}, z^{-\lfloor n/2 \rfloor} \right)$ as $z \to \infty$.	
\end{itemize}
\end{rhproblem}
The RH problems \ref{RHforY1} and \ref{RHforY2} are equivalent. Also RH problem \ref{RHforY2}
has a solution if and only if the monic multiple orthogonal polynomial uniquely exists and in that
case one has 
\begin{equation} \label{Y11} 
	Y_{11}(z) = P_{n,n}(z),
	\end{equation}
as before.

\section{Riemann-Hilbert steepest descent analysis and the proofs of Theorems \ref{theorem2} and \ref{theorem3}}
\label{RHanalysis} 

In this section we perform an asymptotic analysis of the RH problem \ref{RHforY2}.
For convenience we assume that $n$ is even, although this is not essential.
The asymptotic analysis will lead to the proof of the following result.

\begin{lemma} \label{lemmaRH}
Let $t_3 > 0$ and $0 < t_0 < t_{0,crit}$. Then for $n$ large enough the orthogonal
polynomial $P_{n,n}$ exists and satisfies
\begin{equation} \label{Pnnasymptotics} 
	P_{n,n}(z) = M_{1,1}(z) e^{ng_1(z)} (1 + O(1/n)), \qquad \text{ as } n \to \infty 
	\end{equation}
uniformly for $z$ in compact subsets of $\overline{\mathbb C} \setminus \Sigma_1$. Here $g_1$ is defined by
\[ g_1(z) = \int \log (z-s) d\mu_1^*(s), \qquad z \in \mathbb C \setminus \Sigma_1, \]
see also \eqref{gfunctions} below, and $M_{1,1}$ is an analytic function with no zeros 
in $\mathbb C \setminus \Sigma_1$.
\end{lemma}
We write $M_{1,1}$ since this function will arise as the (1,1) entry of a matrix-valued function $M$,
see Section \ref{subsecglobal}.

\subsection{First transformation} \label{subsecfirst}

We start from the RH problem \ref{RHforY2} for $Y$ with $n$ even. 
The first transformation $Y \mapsto X$ will have a different form in 
each of the three sectors \eqref{sectors}.

We recall the Airy functions $y_j$, $j=0,1,2$ from \eqref{y0y1y2}.
We also introduce
\begin{equation} \label{y3y4y5} 
	\begin{aligned}
	y_3(z) & = 2 \pi i(\omega^2 \Ai(\omega z) - \omega \Ai(\omega^2 z)), \\
	y_4(z) & = \omega y_3(\omega z), \\
	y_5(z) & = \omega^2 y_3(\omega^2 z).
	\end{aligned}
	\end{equation} 
These solutions of the Airy differential equation are chosen such that the Wronskians
$W(y_j,y_k) = y_j y_k' - y_j' y_k$ satisfy
\begin{equation} \label{Wy0y3} 
	\begin{aligned}
	W(y_j, y_{j+3}) & = 1, \qquad j=0,1,2, \\
	W(y_3, y_4) & = W(y_3,y_5) = W(y_4, y_5) = 0.  
	\end{aligned}
	\end{equation}
We obviously have from \eqref{y3y4y5} that $y_3'(0)=0$. Hence, by the uniqueness of the
solution of the initial value problem for the Airy equation at $z=0$,
we obtain that  $y_3(\omega z)=y_3(z)$. This implies that $y_4(z) =\omega
y_3(z)$ and $y_5(z)=\omega^2 y_3(z)$.
	
We also recall the Wronskian relations \cite[formulas 10.4.11--10.4.13]{AbrSte}
\begin{equation} \label{Wy0y1} 
	W(y_0,y_1) = W(y_1, y_2) = W(y_2,y_0) = -\frac{1}{2\pi i}. 
	\end{equation}

\begin{definition}
We define 
\begin{align} \label{defX1} 
	\widetilde{X}(z) & =  
		Y(z)  \times \begin{cases}
			\begin{pmatrix} 1 & 0 & 0 \\ 0 & y_3'(c_n z)  & - y_0'(c_n z) \\ 0  & - y_3(c_n z)  & y_0(c_n z) \end{pmatrix},
			 & \quad z \in S_0, \\
      \begin{pmatrix} 1 & 0 & 0 \\ 0 & y_5'(c_n z) & - y_2'(c_n z) \\ 0  & -y_5(c_n z)  & y_2(c_n z) \end{pmatrix},
			 & \quad z \in S_1, \\
	    \begin{pmatrix} 1 & 0 & 0 \\ 0 & y_4'(c_n z) & - y_1'(c_n z) \\ 0  & -y_4(c_n z) & y_1(c_n z) \end{pmatrix},
	     & \quad z \in S_2,
	     \end{cases}
	\end{align}
and
\begin{align} \label{defX2} 
	X(z) = \begin{pmatrix} 1 & 0 & 0 \\ 0 & (2\pi)^{-1/2} c_n^{-1/4} & 0 \\ 0 & 0 & i (2\pi)^{-1/2} c_n^{1/4} \end{pmatrix}
	\widetilde{X}(z) \begin{pmatrix} 1 & 0 & 0 \\ 0 & 1 & 0 \\ 0 & 0 & -2\pi i \end{pmatrix} 
	\end{align}
where $c_n$ is the constant given by \eqref{constantcn}
\end{definition}

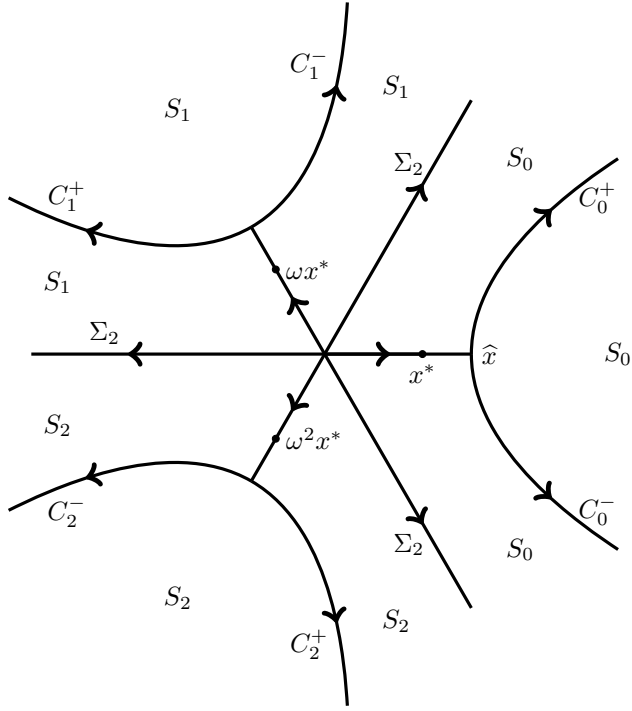
\begin{figure}[t]
\begin{center}
\begin{tikzpicture}[scale=1.3,decoration={markings,mark=at position .67 with {\arrow[black,line width=0.8mm]{>};}}]

\draw[postaction={decorate}, very thick] (0,0)--(1,0) node[below]{$x^*$};
\draw[very thick] (0,0)--(1.5,0) node[right]{$\widehat{x}$};
\draw[postaction={decorate}, very thick, rotate around={120:(0,0)}] (0,0)--(1,0) node[right]{$\omega x^*$};
\draw[very thick, rotate around ={120:(0,0)}] (1,0)--(1.5,0); 
\draw[postaction={decorate}, very thick, rotate around={-120:(0,0)}] (0,0)--(1,0) node[right]{$\omega^2 x^*$};
\draw[very thick, rotate around ={-120:(0,0)}] (1,0)--(1.5,0); 

\draw[postaction={decorate}, very thick] (0,0)--(-3,0) node[near end, above]{$\Sigma_2$};
\draw[postaction={decorate}, very thick, rotate around={120:(0,0)}] (0,0)--(-3,0) node[near end, left]{$\Sigma_2$};
\draw[postaction={decorate}, very thick, rotate around={-120:(0,0)}] (0,0)--(-3,0) node[near end, left]{$\Sigma_2$};

\draw(3,0) node{$S_0$}; \draw(2,2) node{$S_0$};  \draw(2,-2) node{$S_0$}; 
\draw(-1.5,2.5) node{$S_1$}; \draw[rotate around={120:(0,0)}] (2,2) node{$S_1$}; 
\draw[rotate around={120:(0,0)}] (2,-2) node{$S_1$};
\draw(-1.5,-2.5) node{$S_2$}; \draw[rotate around={-120:(0,0)}] (2,2) node{$S_2$}; 
\draw[rotate around={-120:(0,0)}] (2,-2) node{$S_2$};

\draw[postaction={decorate},very thick,rotate around={-90:(1.5,0)}] (1.5,0) parabola (3.5,1.5);
\draw[postaction={decorate},very thick,rotate around={90:(1.5,0)}] (1.5,0) parabola (3.5,-1.5);

\draw[postaction={decorate},very thick,rotate around={30:(-0.75,1.3)}] (-0.75,1.3) parabola (1.25,2.8);
\draw[postaction={decorate},very thick,rotate around={210:(-0.75,1.3)}] (-0.75,1.3) parabola (1.25,-0.2);
\draw[postaction={decorate},very thick,rotate around={-210:(-0.75,-1.3)}] (-0.75,-1.3) parabola (1.25,0.2);
\draw[postaction={decorate},very thick,rotate around={-30:(-0.75,-1.3)}] (-0.75,-1.3) parabola (1.25,-2.8);

\draw(2.5,1.6) node[right]{$C_0^+$}; 
\draw(2.5,-1.6) node[right]{$C_0^-$};

\draw[rotate around={120:(0,0)}] (2.5,1.6) node[above]{$C_1^+$}; 
\draw[rotate around={120:(0,0)}] (2.5,-1.6) node[left]{$C_1^-$};

\draw[rotate around={-120:(0,0)}] (2.5,1.6) node[left]{$C_2^+$}; 
\draw[rotate around={-120:(0,0)}] (2.5,-1.6) node[below]{$C_2^-$};


\filldraw(1,0) circle (1pt);
\filldraw[rotate around={120:(0,0)}](1,0) circle (1pt);
\filldraw[rotate around={-120:(0,0)}](1,0) circle (1pt);
\end{tikzpicture}
\end{center}
\caption{Contour $\Gamma_X$ for the RH problem for $X$ and the sectors $S_0$, $S_1$ and $S_2$.}
\label{fig:contourGammaX}
\end{figure}

Since the formula \eqref{defX1} is different in the three sectors, we will have that $X$
is discontinuous on the boundary of the sectors, which are the rays $\arg z = \pm \pi/3, \pi$
that form the contour $\Sigma_2$ given in \eqref{Sigma2}.
Thus $X$ is defined and analytic in $\mathbb C \setminus \Gamma_X$ where 
$	\Gamma_X = \Gamma \cup  \Sigma_2$,
see Figure \ref{fig:contourGammaX}, with a jump $X_+ = X_- J_X$ on $\Gamma_X$.	
The jump matrix on the intervals $(0, \omega^j \widehat{x}]$  simplifies to
\[ J_X(x) = \begin{pmatrix} 1 & e^{\frac{nt_3}{3t_0} x^3} & 0 \\ 0 & 1 & 0 \\ 0 & 0 & 1 \end{pmatrix}, 
	\qquad x \in \bigcup_j [0, \omega^j \widehat{x}], \]
as a consequence of the definitions \eqref{defX1}--\eqref{defX2} and the Wronskian relation \eqref{Wy0y3}.

The jump matrices on the new contours $\Sigma_2$ are piecewise constant, since they contain Wronskians $W(y_i,y_j)$
of solutions of the Airy equation. Indeed, we have by \eqref{defX1} and \eqref{Wy0y3}--\eqref{Wy0y1},
\begin{align*} 
	J_{\widetilde{X}}(z) & = \begin{pmatrix} 1 & 0 & 0 \\ 0 & y_0(c_n z) & y_0'(c_n z) \\ 0 & y_3(c_n z) & y_3'(c_n z) \end{pmatrix}
	  \begin{pmatrix} 1 & 0 & 0 \\ 0 & y_5'(c_n z) & - y_2'(c_n z) \\ 0 & - y_5(c_n z) & y_2(c_n z) \end{pmatrix} \\
	  & = \begin{pmatrix} 1 & 0 & 0 \\ 0 & W(y_0,y_5) & W(y_2,y_0) \\ 0 & W(y_3,y_5) & W(y_2, y_3) \end{pmatrix} \\
	  & = \begin{pmatrix} 1 & 0 & 0 \\ 0 & \omega^2 & -\frac{1}{2\pi i} \\ 0 & 0 & \omega \end{pmatrix},
	  	\quad \text{for } \arg z = \pi/3. 
	  	\end{align*} 
and then by \eqref{defX2}
\[ J_X(z) = \begin{pmatrix} 1 & 0 & 0  \\ 0 &  \omega^2 & 1 \\ 0 & 0 & \omega \end{pmatrix},
	  	\quad \text{for } \arg z = \pi/3.  \]
The jump matrix $J_X$ turns out to be exactly the same for $\arg z = -\pi/3$ and $\arg z = \pi$. 

On the remaining parts of $\Gamma_X$ (the unbounded contours $C_j$), the jump matrix takes
the form
\[ J_X(z) = \begin{pmatrix} 1 & \alpha e^{\frac{n t_3}{3t_0} z^3} & \beta e^{\frac{n t_3}{3t_0} z^3} \\ 0 & 1 & 0 \\ 0 & 0 & 1 \end{pmatrix}
\]
for certain constants $\alpha$ and $\beta$, that again come from  Wronskians of solutions of
the Airy equation. Let 
\[ y_6(z) = \frac{1}{3}(y_2(z) - y_1(z)), \quad y_7(z) = \omega y_6(\omega z), \quad y_8(z) = \omega^2 y_6(\omega^2 z). \]
Then the constants turn out to be 
\begin{align*} &
	\begin{cases}
	\alpha = W(-y_8,y_3) = \tfrac{1}{2} + \tfrac{1}{6} i \sqrt{3}, \\
	 \beta = 2 \pi i W(-y_8,y_0) = - \tfrac{1}{3}, 
	 \end{cases}  && \text{ on } C_0^+, \\ 
	 & 	 \begin{cases}
	\alpha = W(y_7, y_3) = \tfrac{1}{2} - \tfrac{1}{6} i \sqrt{3}, \\
	\beta  = 2\pi i W(y_7, y_0) = \tfrac{1}{3}, 
	\end{cases} && \text{ on } C_0^-,
	\end{align*}
with the same expressions on the other unbounded parts.

Thus we obtain the following RH problem for $X$.
\begin{rhproblem} \label{RHforX} 
The matrix-valued function $X$ defined by \eqref{defX1}--\eqref{defX2}
is the solution of the following RH problem.
\begin{itemize}
\item $X$ is analytic in $\mathbb C \setminus \Gamma_X$,
\item $X_+ = X_- J_X$ on $\Gamma_X$ with
\begin{equation} \label{JX} 
	J_X(z)  = \begin{cases} \begin{pmatrix} 1 & e^{\frac{n t_3}{3t_0} z^3} & 0 \\ 0 & 1 & 0 \\ 0 & 0 & 1 \end{pmatrix},
 	 & \quad \text{for } z \in \bigcup_j (0, \omega_j \widehat{x}], \\
		\begin{pmatrix} 1 & 0 & 0 \\ 0 & \omega^2 & 1 \\ 0 & 0 & \omega \end{pmatrix},
	  &	\quad \text{for } z \in \Sigma_2, \\
	  \begin{pmatrix} 1 & (\tfrac{1}{2} \pm \tfrac{1}{6} i \sqrt{3}) e^{\frac{n t_3}{3t_0} z^3} &
	  	  \mp \tfrac{1}{3} e^{\frac{n t_3}{3t_0} z^3} \\ 0 & 1 & 0 \\ 0 & 0 & 1
	  	  \end{pmatrix},
	  & \quad \text{for } z \in \bigcup_j C_j^{\pm}.
	  \end{cases} 
	  \end{equation}
\item As $z \to \infty$ we have
\begin{multline} \label{Xasymp}
	 X(z)  =  (I + O(z^{-1})) A(z)  \\
	 \times  \begin{cases}
	 	\begin{pmatrix} z^n & 0 & 0 \\ 0 & z^{-n/2} e^{\frac{2n}{3t_0 \sqrt{t_3}} z^{3/2}} & 0 \\
		 0 & 0 & z^{-n/2} e^{-\frac{2n}{3t_0\sqrt{t_3}} z^{3/2}} \end{pmatrix}, & z \in S_0, \\ 
	 	\begin{pmatrix} z^n & 0 & 0 \\ 0 & z^{-n/2} e^{-\frac{2n}{3t_0 \sqrt{t_3}} z^{3/2}} & 0 \\
		 0 & 0 & z^{-n/2} e^{\frac{2n}{3t_0\sqrt{t_3}} z^{3/2}} \end{pmatrix}, & z \in S_1 \cup S_2,
		 \end{cases}
		 \end{multline}
with
\begin{equation} \label{defAz}
	A(z) =  \begin{pmatrix} 1 & 0 & 0 \\ 0 & z^{1/4} & 0 \\ 0 & 0 & z^{-1/4} \end{pmatrix}
		\times \begin{cases}
	  \begin{pmatrix} 1 & 0 & 0 \\ 0 & \frac{1}{\sqrt{2}} & -\frac{i}{\sqrt{2}} \\ 0 & -\frac{i}{\sqrt{2}} & \frac{1}{\sqrt{2}}
	  \end{pmatrix}, & \text{ for } z  \in S_0, \\	  
	  \begin{pmatrix} 1 & 0 & 0 \\ 0 & \frac{i}{\sqrt{2}} & \frac{1}{\sqrt{2}} \\ 0 & -\frac{1}{\sqrt{2}} & -\frac{i}{\sqrt{2}}
	  \end{pmatrix}, & \text{ for } z \in S_1, \\
	  \begin{pmatrix} 1 & 0 & 0 \\ 0 & -\frac{i}{\sqrt{2}} & -\frac{1}{\sqrt{2}} \\ 0 & \frac{1}{\sqrt{2}} & \frac{i}{\sqrt{2}}
	  \end{pmatrix}, & \text{ for } z \in S_2.
	  \end{cases}
\end{equation}
\end{itemize}
\end{rhproblem}

\begin{remark} \label{remark1}
The asymptotic condition \eqref{Xasymp} for $X$ comes from the asymptotic condition in
the RH problem \ref{RHforY2} for $Y$ combined with 
the asymptotic behavior \eqref{Airyasymptotics} of the Airy functions. 
The condition \eqref{Xasymp}, however, is not valid uniformly as $z \to \infty$ for $z$ close
to $\Sigma_2$. More preceisely, \eqref{Xasymp} holds uniformly as $z \to \infty$ with 
\[ \arg z \in (-\pi + \varepsilon, - \pi/3 - \varepsilon) \cup (-\pi/3 + \varepsilon, \pi/3 - \varepsilon)
	\cup (\pi/3 + \varepsilon, \pi - \varepsilon) \]
for some $\varepsilon > 0$, but not for $\varepsilon =0$. This is due to the fact that
the functions $y_3$, $y_5$ and $y_4$ appearing in the second column of the transformation \eqref{defX1}
(which are dominant solutions of the Airy equation in the respective sectors used in \eqref{defX1})
are not the recessive solutions in the neighboring sectors.

This non-uniformity is a minor issue that will be resolved more or less automatically during
the transformations of the steepest descent analysis, see also \cite{DuiKui, DuKuMo}. To fully
justify the analysis it suffices to supplement the RH problem for $X$ with the following 
condition that controls the behavior as $z \to \infty$ near $\Sigma_2$.
\begin{itemize}
\item As $z \to \infty$ we have
\begin{multline} \label{Xasymp2}
	X(z) = O(1) A(z)  \\
	 \times  \begin{cases}
	 	\begin{pmatrix} z^n & 0 & 0 \\ 0 & z^{-n/2} e^{\frac{2n}{3t_0 \sqrt{t_3}} z^{3/2}} & 0 \\
		 0 & 0 & z^{-n/2} e^{-\frac{2n}{3t_0\sqrt{t_3}} z^{3/2}} \end{pmatrix}, & z \in S_0, \\ 
	 	\begin{pmatrix} z^n & 0 & 0 \\ 0 & z^{-n/2} e^{-\frac{2n}{3t_0 \sqrt{t_3}} z^{3/2}} & 0 \\
		 0 & 0 & z^{-n/2} e^{\frac{2n}{3t_0\sqrt{t_3}} z^{3/2}} \end{pmatrix}, & z \in S_1 \cup S_2,
		 \end{cases}
		 \end{multline}
with $O(1)$ term that is uniform in $z$.
\end{itemize}
\end{remark}

\begin{remark} \label{remark2}
If we put
\[ \Phi(z) = X(z) \begin{pmatrix} e^{\frac{nt_3}{3t_0} z^3} & 0 & 0 \\ 0 & 1 & 0 \\ 0 & 0 & 1 \end{pmatrix} \]
then it is easy to verify that $\Phi$ satisfies a RH problem with piecewise constant jumps. 
Hence $\Phi$ satisfies a differential equation
\[ \frac{d}{dz} \Phi(z) = C(z) \Phi(z) \]
with a polynomial coefficient matrix $C(z)$. In addition there are also differential
equations with respect to $t_0$ and $t_3$, as well as a difference equation in $n$. 
We will not elaborate on this here. Instead we are going to make a different transformation.
\end{remark}

\subsection{Second transformation} \label{subsecsecond}
In the next transformation  the factors $e^{\pm \frac{2n}{3 t_0 \sqrt{t_3}} z^{3/2}}$ are removed
from the asymptotic condition \eqref{Xasymp} in the RH problem for $X$. The transformation also has the effect of 
eliminating the $(1,3)$ entry $\mp \frac{1}{3} e^{\frac{nt_0}{3t_3} z^3}$ from the jump matrix \eqref{JX} 
on the unbounded contours $C_j$.

\begin{definition}
We define
\begin{equation} \label{defV1} 
	\widetilde{V}(z) = X(z) \begin{pmatrix} 1 & 0 & \frac{1}{3} e^{\frac{n t_3}{3t_0} z^3} \\ 0 & 1 & 0 \\ 0 & 0 & 1 \end{pmatrix} 
	\end{equation}
for $z$ in the domain bounded by $C_0 \cup C_1 \cup C_2$ (this is the 
unbounded domain containing $\Sigma_2$ in its interior, see Figure \ref{fig:contourGammaX}),
\begin{equation} \label{defV2} 
	\widetilde{V}(z) = X(z) \qquad \text{elsewhere}, 
	\end{equation}
and
\begin{equation} \label{defV3} 
	V(z) = \widetilde{V}(z) \times
	 \begin{cases} \begin{pmatrix} 1 & 0 & 0 \\ 0 & e^{- \frac{2n}{3t_0 \sqrt{t_3}} z^{3/2}} & 
  0 \\ 0 & 0 & e^{\frac{2}{3} \frac{n}{t_0 \sqrt{t_3}} z^{3/2}}  \end{pmatrix} & \text{ in } S_0, \\
     \begin{pmatrix} 1 & 0 & 0 \\ 0 & e^{\frac{2n}{3t_0 \sqrt{t_3}} z^{3/2}} & 
  0 \\ 0 & 0 & e^{-\frac{2}{3} \frac{n}{t_0 \sqrt{t_3}} z^{3/2}}  \end{pmatrix} & \text{ in } S_1 \cup S_2.
  \end{cases}
\end{equation}
\end{definition}

By straightforward calculations based on the RH problem \ref{RHforX} for $X$ and
the definitions \eqref{defV1}--\eqref{defV3} we find that $V$ is the solution of  the following RH problem.

\begin{rhproblem} \label{RHforV}
The matrix valued function $V : \mathbb C \setminus \Gamma_V \to \mathbb C^{3\times 3}$ satisfies
\begin{itemize}
\item $V$ is analytic in $\mathbb C \setminus \Gamma_V$ where $\Gamma_V = \Gamma_X$.
\item $V_+ = V_- J_V$ on $\Gamma_V$ with
\begin{equation} \label{JV1} 
	J_V(z)  = \begin{cases} \begin{pmatrix} 1 & e^{-\frac{n}{t_0} (\frac{2}{3 \sqrt{t_3}} |z|^{3/2} - \frac{t_3}{3} z^3) } & 0 \\ 0 & 1 & 0 \\ 0 & 0 & 1 \end{pmatrix}, \quad
 	   \hfill{z \in \bigcup_j (0, \omega^j \widehat{x}],} \\
		\begin{pmatrix} 1 & 0 &  (\frac{1}{2} - \frac{1}{6} i \sqrt{3}) e^{-\frac{n}{t_0} (\frac{2}{3 \sqrt{t_3}} z^{3/2} - \frac{t_3}{3} z^3) } \\ 0 & \omega^2 e^{ \frac{4n}{3t_0 \sqrt{t_3}} z^{3/2}}  & 1 \\ 
		0 & 0 & \omega e^{-\frac{4n}{3t_0 \sqrt{t_3}} z^{3/2}} \end{pmatrix}, \\
		\hfill{\arg z  = \pi/3,} \\ 
		\begin{pmatrix} 1 & 0 & (\frac{1}{2} - \frac{1}{6} i \sqrt{3})  e^{\frac{n}{t_0} (\frac{2}{3 \sqrt{t_3}} z^{3/2} + \frac{t_3}{3} z^3) } \\ 0 & \omega^2 e^{- \frac{4n}{3t_0 \sqrt{t_3}} z^{3/2}}  & 1 \\ 
		0 & 0 & \omega e^{\frac{4n}{3 t_0 \sqrt{t_3}} z^{3/2}} \end{pmatrix}, \\
	   \hfill{\arg z = - \pi/3 \text{ or } \arg z = \pi,} 
	 
	  \end{cases}
\end{equation}
and
\begin{equation} \label{JV2}
	J_V(z) = \begin{cases}	  \begin{pmatrix} 1 & (\tfrac{1}{2} \pm \tfrac{1}{6} i \sqrt{3}) e^{-\frac{n}{t_0} (\frac{2}{3 \sqrt{t_3}} z^{3/2} - \frac{t_3}{3} z^3) } 
	  & 0 \\ 0 & 1 & 0 \\ 0 & 0 & 1 \end{pmatrix},
	  \hfill{z \in C_0^{\pm},}  \\
	 \begin{pmatrix} 1 & (\tfrac{1}{2} \pm \tfrac{1}{6} i \sqrt{3}) e^{\frac{n}{t_0} (\frac{2}{3 \sqrt{t_3}} z^{3/2} + \frac{t_3}{3} z^3) } 
	  & 0 \\ 0 & 1 & 0 \\ 0 & 0 & 1 \end{pmatrix},
	  \hfill{z \in C_1^{\pm} \cup C_2^{\pm}.}
	  \end{cases} \end{equation}
\item $ V(z) = (I + O(1/z)) A(z) \begin{pmatrix} z^n & 0 & 0 \\ 0 & z^{-n/2} & 
	0 \\ 0 & 0 & z^{-n/2}   \end{pmatrix}$ as $z \to \infty$, 
	
	where $A(z)$ is given by \eqref{defAz}.
\end{itemize}	 
\end{rhproblem}
In \eqref{JV1} we used the convention that $\arg z^{3/2} = \frac{3}{2} \arg z $, which means that
$z^{3/2} = -i |z|^{3/2}$ for $\arg z  =\pi$. 

The entries $e^{\pm \frac{4 n}{3t_0 \sqrt{t_3}} z^{3/2}}$  in the jump matrix \eqref{JV1} on $\Sigma_2$ 
are oscillatory. Later they will be turned into exponentially decaying entries by opening up unbounded 
lenses around $\Sigma_2$. This indeed leads to exponentially decaying entries because
of the upper triangularity of the jump matrices.

\subsection{Third transformation} \label{subsecthird}
In the third transformation we make use of the minimizer $(\mu_1^*, \mu_2^*)$
of the vector equilibrium problem.
Associated with the measures $\mu_1^*$ and $\mu_2^*$ are the $g$-functions
\begin{equation} \label{gfunctions} 
	g_k(z) = \int \log(z-s) d\mu_k^*(s), \qquad k=1,2, 
	\end{equation}
with appropriately chosen branches of the logarithm. We choose the
branches in such a way that $g_1$ is defined and analytic in $\mathbb C \setminus (\Sigma_1 \cup \mathbb R^-)$
with $g_1(x)$ real for $x$ real and $x > x^*$. Then we have the symmetry
\begin{equation} \label{g1symmetry} 
	\begin{cases}
	g_1(\omega z)  = g_1(z) + 2\pi i/3, \\
	 g_1(\omega^2 z)  = g_1(z) - 2 \pi i/3,
	 \end{cases} \qquad z \in S_0. 
	\end{equation}
The branches of the logarithm in the definition \eqref{gfunctions} of $g_2$ are chosen such that
$g_2$ is defined and analytic in $ \mathbb C \setminus \Sigma_2 = S_0 \cup S_1 \cup S_2$, and $g_2(x)$ is real for real $x > 0$
with the symmetry
\begin{equation} \label{g2symmetry} 
	\begin{cases}
	g_2(\omega z)  = g_2(z) + \pi i/3, \\
	g_2(\omega^2 z) = g_2(z) - \pi i/3,
	\end{cases} \qquad z \in S_0. 
	\end{equation}

The conditions \eqref{EL2}--\eqref{EL1} lead to the following properties of the $g$-functions.
We emphasize that all contours are oriented away from $0$ and towards $\infty$, so that
for $\arg z = \pi$, we have that $g_{2,+}(z)$ is the limiting value from the lower half-plane.

\begin{lemma}
\begin{enumerate}
\item[\rm (a)]
We have, with the constant $\ell$ from \eqref{EL1},
\begin{multline} \label{EL3}
	g_{1,+}(z) + g_{1,-}(z) - g_2(z)  \\ 
		= \begin{cases} \frac{2}{3 t_0 \sqrt{t_3}} z^{3/2} - \frac{t_3}{3t_0} z^3 + \ell, & z \in [0, x^*], \\
	   -\frac{2}{3 t_0 \sqrt{t_3}} z^{3/2} - \frac{t_3}{3t_0} z^3 + \ell + \pi i,  & z \in [0, \omega x^*], \\
	    -\frac{2}{3 t_0 \sqrt{t_3}} z^{3/2} - \frac{t_3}{3t_0} z^3 + \ell - \pi i, & z \in [0, \omega^2 x^*].
	  \end{cases}
\end{multline}
\item[\rm (b)] We have 
\begin{equation} \label{EL4}
\begin{aligned} 
	g_{2,+}(z) + g_{2,-}(z) - g_1(z) & = 0, && \arg z = \pm \pi/3, \\
	g_{2,+}(z) + g_{2,-}(z) - g_{1,\pm}(z) & = \pm \pi i, && \arg z = \pi.
	\end{aligned}
	\end{equation}
\end{enumerate}
\end{lemma}
\begin{proof}
(a) The equality of the real parts follows from \eqref{EL1}, since $\supp(\mu_1^*) = \Sigma_1 = \bigcup_j [0, \omega^j x^*]$.
Both sides of \eqref{EL3} are real for $z \in [0, x^*]$ and so the identity \eqref{EL3} holds on $[0, x^*]$.
The identity on the other intervals, then follows from the symmetry \eqref{g1symmetry}. 

(b) The identities in \eqref{EL4} follow in a similar way from \eqref{EL2}.
\end{proof}

With the $g$-functions and the constant $\ell$ we make the next transformation.

\begin{definition}
We define
\begin{equation} \label{defU} 
	U(z) = \begin{pmatrix} e^{-n \ell} & 0 & 0 \\ 0 & 1 & 0 \\ 0 & 0 & 1 \end{pmatrix}
	V(z) \begin{pmatrix} e^{-n (g_1(z)-\ell)} & 0 & 0 \\ 0 & e^{n(g_1(z) - g_2(z))} & 0 \\ 0 & 0 & e^{n g_2(z)} \end{pmatrix}.
	\end{equation}
\end{definition}

It is our next task to state the RH problem for $U$. 
It will be convenient to write the jump matrices in terms of
the two functions $\varphi_1$ and $\varphi_2$ defined as follows. 
Recall that $\xi_1$, $\xi_2$ and $\xi_3$ are defined in \eqref{xi1}--\eqref{xi3}.

\begin{definition}
We define $\varphi_1 : \mathbb C \setminus (\Sigma_1 \cup \Sigma_2) \to \mathbb C$  by
\begin{equation} \label{phi1}
	\varphi_1(z) = \frac{1}{2 t_0} \int_{\omega^j x^*}^z (\xi_1(s) - \xi_2(s)) ds, 
	\quad z \in S_j \setminus [0, \omega^j x^*], \quad j=0,1,2, 
	\end{equation}
with integration along a path lying in $S_j \setminus [0,\omega^j x^*]$ if $z \in S_j \setminus [0,\omega^j x^*]$ 
for $j=0,1,2$, and
$\varphi_2 : \mathbb C \setminus \{ z \in \mathbb C \mid z^3 \in \mathbb R \} \to \mathbb C$ by
\begin{equation} \label{phi2} 
	\varphi_2(z) = \frac{1}{2 t_0} \int_0^z (\xi_2(s) - \xi_3(s)) ds 
	\begin{cases} - \pi i/6, &  0 < \arg z < \pi/3,  \\
		 	 + \pi i/6, & \pi/3 < \arg z < 2\pi/3,  \\
			+ \pi i/3, & 2\pi/3 < \arg z < \pi,  \\  
			 +\pi i/6, &  -\pi/3 <  \arg z < 0,  \\
		 	 - \pi i/6, & -2 \pi/3 < \arg z < - \pi/3,  \\
			 - \pi i/3, & -\pi < \arg z < -2\pi/3,  
			 \end{cases} \end{equation}
	with integration along a path in the sector $k \pi/3 < \arg s < (k+1) \pi/3$
	if $z$ lies in that sector, for $k=-3, \ldots, 2$. 
\end{definition}

The basic properties of $\varphi_1$ and $\varphi_2$ that connect them with the $g$-functions 
are collected in the next lemma.

\begin{lemma}
\begin{enumerate}
\item[\rm (a)]
For $z \in \Sigma_1$, we have 
\begin{equation} \label{phi1onSigma1} 
		g_{1,+}(z) - g_{1,-}(z) = \pm 2 \varphi_{1,\pm} (z).
		\end{equation}
\item[\rm (b)] For $z \in \mathbb C \setminus (\Sigma_1 \cup \Sigma_2)$ we have
\begin{align} \label{phi1inS0S1S2}
	2 g_1(z) - g_2(z) - \ell 
		=	2 \varphi_1(z) \pm \frac{2}{3t_0 \sqrt{t_3}} z^{3/2} - \frac{t_3}{3t_0} z^3
		+\begin{cases} 0 & \text{in } S_0, \\ \pi i & \text{in } S_1, \\ -\pi i & \text{in } S_2, \end{cases}
		\end{align}
		with $+$ in $S_0$ and $-$ in $S_1 \cup S_2$.
\item[\rm (c)]
For $z \in \Sigma_2$, we have 
\begin{equation} \label{phi2onSigma2} 
	\begin{aligned}
		g_{2,+}(z) - g_{2,-}(z) & =  2 \varphi_{2,+} (z) \pm \frac{4}{3t_0 \sqrt{t_3}} z^{3/2} \\
		&		 = - 2 \varphi_{2,-} (z) \pm \frac{4}{3t_0 \sqrt{t_3}} z^{3/2},
		\end{aligned}
		\end{equation}
		where we use $+$ for $\arg z = \pi/3$ and $-$ for $\arg z = -\pi/3$ or $\arg z = \pi$.
\item[\rm (d)]
For $z \in \Sigma_2$, we have 
\begin{multline} \label{phi1onSigma2} 
		g_{1,-}(z) + g_{2,+}(z)  - \ell \\
		= \begin{cases} 2 \varphi_{1,-}(z) + \frac{2}{3t_0 \sqrt{t_3}} z^{3/2} - \frac{t_3}{3t_0} z^3, & \arg z = \pi/3, \\
		2 \varphi_{1,-}(z) - \frac{2}{3t_0 \sqrt{t_3}} z^{3/2} - \frac{t_3}{3t_0} z^3, & \arg z = \pi, \\
		2 \varphi_{1,-}(z) - \frac{2}{3t_0 \sqrt{t_3}} z^{3/2} - \frac{t_3}{3t_0} z^3 - \pi i, & \arg z = -\pi/3.
		\end{cases}
		\end{multline}
\item[\rm (e)] We have for $z \in \Sigma_1$,	  	
\begin{align}	 \label{phi1andphi2onSigma1}
	  2\varphi_{1,+}(z) + 2\varphi_{2,+}(z)  =  2 \varphi_{2,-}(z) 
	  	+ \begin{cases} 0,
	  	 & \text{for } z \in [0, x^*], \\
	  	 \pi i, & \text{for } z \in [0, \omega x^*], \\
	  	- \pi i, & \text{for } z \in [0, \omega^2 x^*].
	  	\end{cases}
	  	\end{align}
and for $z \in \Sigma_2$,
\begin{align} \label{phi1andphi2onSigma2}
	  2\varphi_{1,+}(z) + 2\varphi_{2,+}(z)  =
	  	2\varphi_{1,-}(z) -
	  	\begin{cases}
	  	\pi i, & \text{for } \arg z = \pm \pi /3, \\
	  	 2\pi i, & \text{for } \arg z = \pi, 
	  	\end{cases} 
	  	\end{align}
\end{enumerate}
\end{lemma}
\begin{proof}
By  \eqref{Cauchytransforms} and \eqref{gfunctions} we have that $g_1' = F_1$ and $g_2' = F_2$, 
so that by \eqref{xi1}--\eqref{xi3} we can express the derivatives of the left-hand sides
of \eqref{phi1onSigma1}--\eqref{phi2onSigma2} in terms of the $\xi$-functions.
This yields
\begin{align*}
	g_{1,+}'(z) - g_{1,-}'(z) & = \frac{1}{t_0} (\xi_{1,+}(z) - \xi_{1,-}(z)), & z \in \Sigma_1, \\
	2g_1'(z) - g_2'(z)  & = \frac{1}{t_0} (\xi_1(z) - \xi_2(z)) \pm  \frac{2z^{1/2}}{t_0 \sqrt{t_3}} - \frac{t_3}{t_0} z^2, \\
	g_{2,+}'(z) - g_{2,-}'(z) & = \frac{1}{t_0} (\xi_{3,-}(z) - \xi_{3,+}(z)) \pm \frac{2z^{1/2}}{t_0 \sqrt{t_3}}, & z \in \Sigma_2, 
	\end{align*}  
with the same conventions on $\pm$ signs  as in parts (b) or (c) of the lemma.
In view of the definitions \eqref{phi1}--\eqref{phi2} we have
\begin{equation} \label{phijprime} 
	2 \varphi_j'(z) = \frac{1}{t_0} (\xi_j(z) - \xi_{j+1}(z)), \qquad j=1,2. 
	\end{equation}
Using also that  $\xi_{1,\pm} = \xi_{2,\mp}$ on $\Sigma_1$ and $\xi_{3,\pm} = \xi_{2,\mp}$ on $\Sigma_2$,
we then easily check that the derivatives of the two sides in \eqref{phi1onSigma1}--\eqref{phi2onSigma2}
agree on the respective contours or regions. It follows that \eqref{phi1onSigma1}--\eqref{phi2onSigma2} 
hold up to a possible constant of integration.

Both sides of \eqref{phi1onSigma1} vanish for $z = \omega^j x^*$, $j=0,1,2$, while 
the equality in \eqref{phi1inS0S1S2} for $z = \omega^j x^*$, $j=0,1,2$ follows
from the identities \eqref{EL3}. Hence the equalities in parts (a) and (b) follow. 

To establish the equality in part (c)  we examine the behavior as $z \to 0$.
Because of the symmetry \eqref{g2symmetry} we find that there exists a
constant $\gamma_2$ such that $g_2$ has the following limits a the origin
\begin{align} \label{g2at0}
  \lim_{z \to 0} g_2(z) = \begin{cases} \gamma_2, & z \in S_0, \\
  	\gamma_2 + \pi i/3, & z \in S_1, \\
  	\gamma_2 - \pi i/3, & z \in S_2.
  	\end{cases}
  	\end{align}
From \eqref{g2at0} we find the limits of $g_{2,+}(z) - g_{2,-}(z)$ as $z \to 0$ on 
the three half-rays in $\Sigma_2$. These values correspond to the limits of $\pm 2 \varphi_{2,\pm}(z)$
as $z \to 0$  that we obtain from the definition \eqref{phi2} and part (c) follows.

Using \eqref{EL4}  we find for $z \in \Sigma_2$ that
\begin{align} 
	g_{1,-}(z) +g_{2,+}(z) - \ell  
		\label{g1andg2}
		 = \begin{cases} 2g_1(z) - g_{2,-}(z) - \ell, & \arg z = \pm \pi/3, \\
		  2 g_{1,-}(z) - g_{2,-}(z)  - \ell - \pi i,  & \arg z = \pi.
		  \end{cases}
		  \end{align}
Then \eqref{phi1onSigma2} follows by letting $z$ in \eqref{phi1inS0S1S2} go to the minus-side of $\Sigma_2$,
and inserting this into \eqref{g1andg2}.

Because of \eqref{phijprime} we have
\[ 2(\varphi_1'(z) + \varphi_2'(z)) = \frac{1}{t_0} (\xi_1(z) - \xi_3(z)) \]
Using $\xi_{1,+} = \xi_{2,-}$ on $\Sigma_1$ and $\xi_{3,+} = \xi_{2,-}$ on $\Sigma_2$,
we then find that the derivatives of the two sides of 
\eqref{phi1andphi2onSigma1} and \eqref{phi1andphi2onSigma2} agree. To prove
the two equalities we again examine the behavior as $z \to 0$.
From \eqref{g2at0} and \eqref{EL4} we find the limits of $g_1$ at the origin,
\begin{align} \label{g1at0}
  \lim_{z \to 0} g_1(z) = \begin{cases} 2\gamma_2 + \pi i/3, & 0 < \arg z < 2\pi /3, \\
  	2\gamma_2 + \pi i, & 2 \pi/3 < \arg z < \pi, \\
  	2\gamma_2 - \pi i/3, & -2\pi/3 < \arg z < 0, \\
  	2 \gamma_2 - \pi i, & -\pi < \arg z < -2\pi/3,
  	\end{cases}
  	\end{align}
where $\gamma_2$ is the constant from \eqref{g2at0}. Using \eqref{EL3} we also find
\[ \ell = 3 \gamma_2. \]
Combining \eqref{g2at0} and \eqref{g1at0} with \eqref{phi1onSigma1} and \eqref{phi2onSigma2} 
we obtain the limits of $\varphi_1$  and $\varphi_2$ at the origin.
\begin{align} \label{phi1at0}
  \lim_{z \to 0} 2 \varphi_1(z) & = 
  	\begin{cases} 2\pi i/3, & \text{if } k\pi/3 < \arg z < (k+1) \pi /3
  	\text{ with } k \text{ even}, \\
    -2\pi i/3, & \text{if } k\pi/3 < \arg z < (k+1) \pi /3
  	\text{ with } k \text{ odd}, 
  	\end{cases} \\
  	\label{phi2at0}
  \lim_{z \to 0} 2 \varphi_2(z) & = 
  	\begin{cases} -\pi i/3, & 0 < \arg z < \pi/3 \text{ or } -2\pi/3 < \arg z < -\pi/3, \\
    \pi i/3, & -\pi/3 < \arg z < 0 \text{ or } \pi/3 < \arg z < 2\pi/3, \\
  	2\pi i/3, & 2\pi/3 < \arg z < \pi, \\
  	-2\pi i/3, & -\pi < \arg z < -2\pi/3.
  	\end{cases}
  	\end{align}
The values in \eqref{phi2at0} also follow from the definition \eqref{phi2} of $\varphi_2$.
From \eqref{phi1at0} and \eqref{phi2at0} we check that the left- and right-hand
 sides 	of \eqref{phi1andphi2onSigma1} and \eqref{phi1andphi2onSigma2} have the
 same values as $z \to 0$ and the equalities 
\eqref{phi1andphi2onSigma1} and \eqref{phi1andphi2onSigma2} follow.
\end{proof}
  
Then we have the following.	
\begin{rhproblem} \label{RHforU}
$U$ is the solution of the following RH problem.
\begin{itemize}
\item $U$ is analytic in $\mathbb C \setminus \Gamma_U$, where $\Gamma_U = \Gamma_V$,
\item $U_+ = U_- J_U$ on $\Gamma_U$ with
\begin{equation} \label{JU} 
	J_U(z) = \begin{cases} 
	\begin{pmatrix} e^{-2n \varphi_{1,+}(z)} & 1 & 0 \\ 
	0 & e^{-2n \varphi_{1,-}(z)} & 0 \\ 0 & 0 & 1 \end{pmatrix}, &  z \in \Sigma_1, \\
	\begin{pmatrix} 1 & e^{2n \varphi_1(z)} & 0 \\ 0 & 1 & 0 \\ 0 & 0 & 1 \end{pmatrix}, &
		z \in \bigcup_j [\omega^j x^*, \omega^j \widehat{x}], \\
		\begin{pmatrix} 1 & 0 & (\frac{1}{2} - \frac{1}{6} i \sqrt{3}) e^{2n \varphi_{1,-}(z)} \\ 0 & \omega^2 e^{-2n \varphi_{2,+}(z)}  & 1 \\ 
		0 & 0 & \omega e^{-2n \varphi_{2,-}(z)} \end{pmatrix}, &  z \in \Sigma_2, \\
	  \begin{pmatrix} 1 & (\frac{1}{2} \pm \frac{1}{6} i \sqrt{3}) e^{2n \varphi_1(z)} & 0 \\ 0 & 1 & 0 \\ 0 & 0 & 1 \end{pmatrix},
	  	& z \in \bigcup_j C_j^{\pm}.
	  \end{cases} \end{equation}
\item $U(z) = (I + O(1/z)) A(z)$ as $z\to \infty$, where $A(z)$ is given by \eqref{defAz}.
\end{itemize}
\end{rhproblem}

\begin{proof} 
From \eqref{JV1} on $\bigcup_j (0, \omega^j \widehat{x}]$ and the definition \eqref{defU} we obtain
by the fact that $g_2$ is analytic, that $J_U$ on $\bigcup_j [0,\omega_j \widehat{x}]$
takes the form
\begin{equation} \label{JUonSigma1} 
	\begin{pmatrix} e^{-n(g_{1,+}-g_{1,-})} & 
	e^{n(g_{1,+}+g_{1,-} - g_2 - \ell - \frac{1}{t_0}(\frac{2}{3\sqrt{t_3}} |z|^{3/2} - \frac{t_3}{3} z^3))} & 0 \\
	0 & e^{n(g_{1,+}-g_{1,-})} & 0 \\ 0 & 0 & 1 \end{pmatrix}. 
	\end{equation}
	The $(1,2)$ entry of \eqref{JUonSigma1} is equal to $1$ on $\Sigma_1$ because of \eqref{EL3} (recall that $n$ is even
	by our assumption at the beginning of subsection \ref{subsecfirst}, so that
  the terms $\pm \pi i$ that are in \eqref{EL3} on $[0, \omega x^*]$ and $[0,\omega^2 x^*]$ do not matter). 
  The nontrivial diagonal entries in \eqref{JUonSigma1} reduce on $\Sigma_1$ to the ones as given in \eqref{JU} on $\Sigma_1$
  because of the identities \eqref{phi1onSigma1}.
  On $\bigcup_j [\omega^j x^*, \omega^j \widehat{x}]$ the diagonal entries in \eqref{JUonSigma1} are equal to $1$,
  since $g_{1,+} = g_{1,-}$ is analytic. The $(1,2)$ entry is $e^{2n \varphi_1}$ because of 
  the identity \eqref{phi1inS0S1S2}.
  
 For the jump matrix $J_U$ on $\Sigma_2$ we obtain from \eqref{JV1}, \eqref{defU} and the fact that
 $e^{ng_1}$ is analytic on $\Sigma_2$ 
\begin{equation} \label{JUonSigma2} 
	\begin{pmatrix} 1 & 0 & (\frac{1}{2} - \frac{1}{6} i \sqrt{3}) e^{n (g_{1,-} + g_{2,+} -  \ell  + \frac{1}{t_0} ( \mp \frac{2}{3 \sqrt{t_3}} z^{3/2} + \frac{t_3}{3} z^3))} \\
	0 & \omega^2 e^{-n(g_{2,+}+g_{2,-} \pm \frac{4}{3t_0 \sqrt{t_3}} z^{3/2})} & 0 \\
	0 & 0 & \omega e^{n(g_{2,+}+g_{2,-} \pm \frac{4}{3t_0 \sqrt{t_3}} z^{3/2})} 
	\end{pmatrix} 
	\end{equation}
with $+$ for $\arg z = \pi/3$ and $-$ for $\arg z = - \pi/3$ or $\arg z = \pi$.
The nonconstant entries in \eqref{JUonSigma2} are simplified using \eqref{phi2onSigma2}
and \eqref{phi1onSigma2} which leads to the expressions given in \eqref{JU}.
Again we use the fact that $n$ is even so that the term $\pm \pi i$ in \eqref{phi1onSigma2} 
on $\arg z = -\pi/3$ does not contribute.

For $z \in \bigcup_j C_j^{\pm}$ we find from \eqref{JV2} and \eqref{defU}
that the jump matrix $J_U$ is equal to
\[ \begin{pmatrix} 1 & (\frac{1}{2} \pm \frac{1}{6} i \sqrt{3}) e^{\frac{n}{t_0} (\mp \frac{2}{3 \sqrt{t_3}} z^{3/2} + \frac{t_3}{3} z^3) + n (2g_1(z) - g_2(z) - \ell)} & 0 \\
	0 & 1 & 0 \\ 0 & 0 & 1 \end{pmatrix}. \]
Here we simplify the $(1,2)$ entry by means of \eqref{phi1inS0S1S2} and thereby obtain 
the jump matrix that is given in \eqref{JU} on $\bigcup_j C_j^{\pm}$.

The asymptotic condition in the RH problem \ref{RHforU} follows from 
the asymptotic condition in the RH problem \ref{RHforV} for $V$, the transformation \eqref{defU}
and the fact that
\[ g_1(z) = \log z + O(z^{-3}), \qquad g_2(z) = \tfrac{1}{2} \log z + O(z^{-3/2}) \]
as $z \to \infty$.
\end{proof}

The jump matrices $J_U$ in the RH problem \ref{RHforU} for $U$ have a good form for the opening
of lenses around $\Sigma_1$ and $\Sigma_2$. Indeed, the expressions $\varphi_{j,\pm}(z)$, $j=1,2$, that
appear in the non-trivial diagonal entries in the jump matrices  are purely imaginary
and the diagonal entries $e^{-2n\varphi_{j,\pm}(z)}$ are highly oscillatory if $n$ is large.
The opening of lenses will turn these entries into exponentially decaying off-diagonal entries.

The other non-constant entries in the jump matrix $J_U$ have a factor $e^{2n \varphi_1}$ or $e^{2n \varphi_{1,-}}$.
These are exponentially decaying provided that $\Re \varphi_1 < 0$ on $\bigcup_j ((\omega^j x^*, \omega^j \widehat{x}] \cup C_j)$
and $\Re \varphi_{1,-} < 0$ on $\Sigma_2$. The first of these inequalities indeed holds, but the second
one is not satisfied on the part of $\Sigma_2$ in a neighborhood of the origin. 
Fortunately, we can remove the $(1,3)$ entry in the jump matrix \eqref{JU} on $\Sigma_2$ 
by a transformation  in a larger domain  that is similar to the global opening of lenses in \cite{ApBlKu}.
We will do this in the next subsection.
The  transformation will not affect the jump matrices on $\Sigma_1$ and
$\bigcup_j [\omega^j x^*, \omega^j \widehat{x}]$.

\subsection{Fourth transformation} \label{subsecfourth}

The transformation $U \mapsto T$ is based on the following property of the set where $\Re \varphi_1 < 0$.
Recall that $\varphi_1$, see \eqref{phi1}, is defined and analytic in $\mathbb C \setminus (\Sigma_1 \cup \Sigma_2)$.
From \eqref{phi1onSigma1} and \eqref{gfunctions} it follows that
\begin{equation} \label{Rephi1onSigma1} 
	\varphi_{1,\pm}(z) = \pm \pi i \mu_1^*([z, \omega^j x^*]), \qquad \text{for } z \in [0, \omega^j x^*],
	\quad j=0,1,2,
	\end{equation}
so that  $\Re \varphi_{1,\pm} = 0$ on $\Sigma_1$.
On $\Sigma_2$ we have the identities \eqref{phi1andphi2onSigma2} with $\varphi_{2,+}$ purely
imaginary.  Therefore $\Re \varphi_1$ extends to a continuous function on $\mathbb C$,
which we also denote by $\Re \varphi_1$. With this understanding we define the domain $D$ as follows   
\begin{equation} \label{defD}
	D = \{ z \in \mathbb C  \mid \Re \varphi_1(z) < 0 \}.
	\end{equation}

Since $\Re \varphi_1$ is harmonic in $\mathbb C \setminus (\Sigma_1 \cup \Sigma_2)$,
we have that the boundary $\partial D$ consists of a finite number
of analytic arcs that start and end on $\Sigma_1 \cup \Sigma_2$ or at infinity.

From \eqref{Rephi1onSigma1} it follows that $\pm i \varphi_{1,\pm}$ is real and strictly increasing
with a positive derivative along 
each interval $[0, \omega^j x^*]$ in $\Sigma_1$. Then by the Cauchy-Riemann equations, 
it follows that the closed set $\mathbb C \setminus D$ contains the intervals $[0, \omega^j x^*)$, $j=0,1,2$ 
in its interior. Note that $0$
is also an interior point, since $\mu_1^*$ has a positive density at  $0$.
Near the branch point $x^*$ we have
\[ \varphi_1(z) = - c (z- x^*)^{3/2} \left(1+ O\left((z-x^*)^{-1}\right)\right) \qquad \text{as } z \to x^* \]
with a positive constant $c > 0$. This follows from the fact that $\mu_1^*$
vanishes as a square root at $x^*$. Thus $\Re \varphi_1 < 0$ immediately to the right
of $x^*$. It also follows that $x^* \in \partial D$ and $\partial D$
makes angles $\pm 2\pi/3$ with $[0, x^*]$. The same holds
at $\omega^j x^*$ because of the rotational symmetry.

\begin{figure}
\centering
\begin{overpic}[width=8cm]{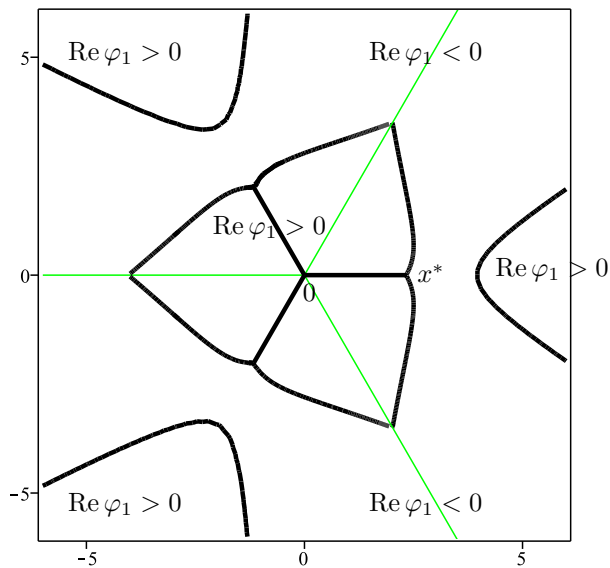}
\put(54,50){$0$}
\put(73,53){$x^*$}
\put(39,61){$\Re \varphi_1 > 0$}
\put(65,90){$\Re \varphi_1 < 0$}
\put(65,15){$\Re \varphi_1 < 0$}
\put(15,15){$\Re \varphi_1 > 0$}
\put(15,90){$\Re \varphi_1 > 0$}
\put(86,54){$\Re \varphi_1 > 0$}
\end{overpic}
\caption{Curves $\Re \varphi_1(z) = 0$ and the domains where $\Re \varphi_1(z) < 0$ and $\Re \varphi_1(z) < 0$.
The figure is for the values $t_0 = 1.8$, $t_3 = 0.25$.}
\label{fig:Rephi1}
\end{figure}

Since
\[ 2\varphi_1(z) = \frac{t_3}{3t_0} z^3 + O(z^{3/2}) \qquad \text{as } z \to \infty, \]
which follows for example from \eqref{phi1inS0S1S2},
the domain $D$ is unbounded and extends to infinity in the three directions
\[ |\arg z - (2k-1) \pi/3| < \pi/6, \qquad k =0,1,2,  \]
around the directions $\pm \pi/3$, $\pi$ of $\Sigma_2$.
The unbounded parts of $\partial D$ then extend to infinity at angles $\pm \pi/6$, $\pm \pi/2$ and
$\pm 5 \pi/6$. From what we now know about $D$ it follows that
$\partial D$ has three unbounded parts. One of them intersects the positive real line
and tends to infinity at angles $\pm \pi/6$. The other two are obtained by rotation over 
angles $\pm 2\pi /3$. 

The analytic arcs in $\partial D$ that start at the branch points $\omega^j x^*$
then necessarily end at points on $\Sigma_2$. The conclusion is that
we have a situation as shown in Figure~\ref{fig:Rephi1}. In particular the
domain $D$ is connected.

We also find that $\Re \varphi_1$ takes on a negative minimum on the interval $(x^*, \infty)$ 
which by \eqref{phi1} and parts (c) and (d) of Lemma \ref{lem:cubic} is taken at $\widehat{x}$.
Thus $\widehat{x} \in D$, and by symmetry also $\omega^j \widehat{x} \in D$ for $j=1,2$. 
We may then also assume that the unbounded contours $C_j$ from \eqref{contoursCjpm}--\eqref{contoursCj} 
(see also Figure \ref{fig:deformedGamma}) satisfy 
\[ C_j \subset D \cap S_j, \qquad j=0,1,2, \]
and that they extend to infinity at asymptotic directions $-\pi/3 + \theta_0$ and $\pi/3 - \theta_0$ (for $C_0$),
$\pi/3 + \theta_0$ and $\pi - \theta_0$ (for $C_1$) and $-\pi + \theta_0$ and $-\pi/3 - \theta_0$ (for $C_2$),
where $\theta_0 \in (0, \pi/6)$ is some fixed angle.
 The above choice makes sure that 
\begin{equation} \label{Rephi1onC1} \Re \varphi_1(z) < 0, \qquad z \in  C_j, \, j=0,1,2, 
\end{equation}
and, with positive constants $c_1, c_2 > 0$,
\begin{equation} \label{Rephi1onC2} 
	\Re \varphi_1(z) \leq -c_1 |z|^3 + c_2, 
	\end{equation}
for $z$ in the domain bounded by $C_0 \cup C_1 \cup C_2$.

After these preparations we define the transformation $U \mapsto T$ as follows.

\begin{definition}
We define $\widetilde{T}$ as 
\[ \widetilde{T}(z) = U(z) \]
for $z$ outside the domain bounded by $C_0 \cup C_1 \cup C_2$,
and for $z$ inside this domain,
\begin{align} \label{defT1}
	\widetilde T(z) = U(z)  \begin{pmatrix} 1 & ((-1)^k \tfrac{1}{2} - \tfrac{1}{6} i \sqrt{3}) e^{2n \varphi_1(z)} & 0 \\ 
	0 & 1 & 0 \\ 0 & 0 & 1 \end{pmatrix} &
		\begin{array}{l} \text{for } \arg z \in (\tfrac{k \pi}{3}, \tfrac{(k+1) \pi}{3}) \\ 
		\text{and } k = -3, \ldots, 2. \end{array} 
\end{align}
Then $T$ is defined as 
\begin{equation} \label{defT2} 
	T(z) = \begin{pmatrix} \tfrac{1}{2} & 0 & 0 \\ 0 & 1 & 0 \\ 0 & 0 & 1 \end{pmatrix}
	\widetilde T(z) \begin{pmatrix} 2 & 0 & 0 \\ 0 & 1 & 0 \\ 0 & 0 & 1 \end{pmatrix}. 
	\end{equation}
\end{definition}

The RH problem for $T$ is the following.
\begin{rhproblem} \label{RHforT}
$T$ is the solution of the following RH problem.
\begin{itemize}
\item $T$ is analytic in $\mathbb C \setminus \Gamma_T$, where $\Gamma_T = \Gamma_U$,
\item $T_+ = T_- J_T$ on $\Gamma_T$ with
\begin{equation} \label{JT} 
	J_T(z) = \begin{cases} 
	\begin{pmatrix} e^{-2n \varphi_{1,+}(z)} & 1 & 0 \\ 
	0 & e^{-2n \varphi_{1,-}(z)} & 0 \\ 0 & 0 & 1 \end{pmatrix}, &  z \in \Sigma_1, \\
		\begin{pmatrix} 1 & 0 & 0 \\ 0 & \omega^2 e^{-2n \varphi_{2,+}(z)}  & 1 \\ 
		0 & 0 & \omega e^{-2n \varphi_{2,-}(z)} \end{pmatrix}, &  z \in \Sigma_2, \\
		\begin{pmatrix} 1 & e^{2n \varphi_1(z)} & 0 \\ 0 & 1 & 0 \\ 0 & 0 & 1 \end{pmatrix},
			& z \in \bigcup_{j=0}^2 [\omega^j x^*, \omega^j \widehat{x}], \\
		\begin{pmatrix} 1 & \frac{1}{2} e^{2n \varphi_1(z)} & 0 \\ 
	0 & 1 & 0 \\ 0 & 0 & 1 \end{pmatrix}, & z \in  \bigcup_{j=0}^2 C_j.
	  \end{cases} \end{equation}
\item $T(z) = (I + O(1/z)) A(z)$ as $z\to \infty$, where $A(z)$ is given by \eqref{defAz}.
\end{itemize}
\end{rhproblem}

\begin{proof}
From \eqref{Rephi1onC2} and \eqref{defT1} we get that $\widetilde{T}(z) = U(z) \left(I + O\left(e^{-c_1|z|^3}\right)\right)$
as $z \to \infty$, which by the asymptotic condition in the RH problem \ref{RHforU} for $U$ leads to
\[ \widetilde{T}(z) = (I + O(1/z)) A(z) \]
as $z \to \infty$. The transformation \eqref{defT2} does not affect this asymptotic
behavior, because of the special form \eqref{defAz} of $A(z)$.

We next verify the jump matrices \eqref{JT}. We write  
\begin{equation} \label{defalpha} 
	\alpha = \frac{1}{2} + \frac{1}{6} i \sqrt{3}.
	\end{equation}

For the jump matrix $J_{\widetilde{T}} = \left(\widetilde{T}_-\right)^{-1} \widetilde{T}_+$ on $\Sigma_1$ we obtain by
\eqref{JU} and \eqref{defT1}
\begin{align*} 
	J_{\widetilde{T}} & = \begin{pmatrix} 1 & \alpha e^{2n \varphi_{1,-}} & 0 \\ 0 & 1 & 0 \\ 0 & 0 & 1 \end{pmatrix}
	J_U \begin{pmatrix} 1 & \bar{\alpha} e^{2n \varphi_{1,+}} &0 \\ 0 & 1 & 0 \\ 0 & 0 & 1 \end{pmatrix} \\
		& = \begin{pmatrix} e^{-2n \varphi_{1,+}} & 1 + \alpha + \bar{\alpha} & 0 \\ 0 & e^{-2n \varphi_{1,-}} & 0 \\ 0 & 0 & 1 \end{pmatrix}
		= \begin{pmatrix} e^{-2n \varphi_{1,+}} & 2 & 0 \\ 0 & e^{-2n \varphi_{1,-}} & 0 \\ 0 & 0 & 1 \end{pmatrix}.
		\end{align*}
The transformation \eqref{defT2} has the effect of dividing the $(1,2)$ entry in $J_{\widetilde{T}}$ by $2$
and we obtain the jump matrix $J_T$ on $\Sigma_1$ as given in \eqref{JT}.
Similar calculations lead to the jump matrices $J_T$  on $\bigcup_j [\omega^j x^*, \omega^j \widehat{x}]$ and  on
$\bigcup_j C_j^{\pm}$.

For the jump matrix $J_{\widetilde{T}}$ on $\Sigma_2$ we obtain by
\eqref{JU} and \eqref{defT1}
\begin{align} \nonumber 
	J_{\widetilde{T}} & =
	\begin{pmatrix} 1 & - \bar{\alpha} e^{2n \varphi_{1,-}} & 0 \\ 0 & 1 & 0 \\ 0 & 0 & 1 \end{pmatrix}
	J_U \begin{pmatrix} 1 & - \alpha e^{2n \varphi_{1,+}} &0 \\ 0 & 1 & 0 \\ 0 & 0 & 1 \end{pmatrix} \\
	& = \begin{pmatrix} 1 & - \alpha e^{2n \varphi_{1,+}} - \bar{\alpha} \omega^2 e^{2n(\varphi_{1,-} - \varphi_{2,+})} & 0 \\
		0 & \omega^2 e^{-2n \varphi_{2,+}} & 1 \\ 0 & 0 & \omega e^{-2n \varphi_{2,-}}.
		\end{pmatrix} \label{JtildeTonSigma2}
\end{align}
By \eqref{phi1andphi2onSigma2} we have $e^{2n (\varphi_{1,-} - \varphi_{2,+})} = e^{2n \varphi_{1,+}}$ on $\Sigma_2$,
and so the $(1,2)$ entry in \eqref{JtildeTonSigma2} reduces to
\[ (-\alpha - \bar{\alpha} \omega^2) e^{2n \varphi_{1,+}} = 0, \]
because of the value \eqref{defalpha} we have for $\alpha$. Thus the $(1,2)$ entry 
in \eqref{JtildeTonSigma2} vanishes. The transformation \eqref{defT2} then gives $J_T = J_{\tilde T}$
and it is equal to the jump matrix \eqref{JT}  on $\Sigma_2$. 
\end{proof}

Since $\Re \varphi_1 < 0$ on $\bigcup_j ((\omega^j x^*, \omega^j \widehat{x}] \cup C_j^{\pm})$
we see from \eqref{JT} that the jump matrices $J_T$ on these parts of $\Gamma_T$
tend to the identity matrix as $n \to \infty$.
What remains are oscillatory jump matrices on $\Sigma_1$ and $\Sigma_2$.

\subsection{Fifth transformation} \label{subsecfifth}
We open up lenses around the three intervals in $\Sigma_1$ and around $\Sigma_2$. 
We use 
\begin{equation} \label{defL1} 
	L_1 = L_1^+ \cup L_1^- \end{equation}
to denote the lens around $\Sigma_1$ where
$L_1^+$ ($L_1^-$) is the part that lies on the $+$ side ($-$ side) of $\Sigma_1$. 

The boundary $\partial L_1$ of $L_1$ meets $\Sigma_1$  at the endpoints $\omega^j x^*$, $j=0,1,2$
but not at $0$.  The boundary  meets $\Sigma_2$ at points at a positive distance $\delta_1 > 0$
from the origin, see Figure \ref{fig:lenses}, and we choose $L_1$ such that
 $\partial L_1 \setminus \{ \omega^j x^* \}$ is contained in the region  where $\Re \varphi_1 > 0$,
 see also Figures~\ref{fig:Rephi1} and \ref{fig:lenses}.

Similarly we write
\begin{equation} \label{defL2} 
	L_2 = L_2^+ \cup L_2^- \end{equation}
for the lens around $\Sigma_2$. The boundary $\partial L_2$ lies in
the region where $\Re \varphi_2 > 0$  and they are
bounded by infinite rays that meet $\Sigma_1$ at points at a positive distance $\delta_2 > 0$ from
the origin. The rays have an asymptotic angles $\pm \pi/3 \pm \varepsilon$, 
$\pi \pm \varepsilon $ for some small $\varepsilon > 0$, see also Figure \ref{fig:lenses}.
We can indeed open the lense around $\Sigma_2$ this way, since $\mu_2^*$ has a  positive
density, also at $0$.

\begin{figure}[t]
\begin{center}
\begin{tikzpicture}[scale=1.3,decoration={markings,mark=at position .67 with {\arrow[black,line width=0.8mm]{>};}}]

\draw[very thick] (0,0)--(1,0) node[below]{$x^*$};
\draw[very thick] (0,0)--(1.5,0); 
\draw[very thick, rotate around ={120:(0,0)}] (0,0)--(1.5,0); 
\draw[very thick, rotate around ={-120:(0,0)}] (0,0)--(1.5,0); 

\draw[postaction={decorate}, very thick] (0,0)--(-3,0);
\draw[postaction={decorate}, very thick, rotate around={120:(0,0)}] (0,0)--(-3,0);
\draw[postaction={decorate}, very thick, rotate around={-120:(0,0)}] (0,0)--(-3,0);

\draw[postaction={decorate},very thick,rotate around={-90:(1.5,0)}] (1.5,0) parabola (3.5,1.5);
\draw[postaction={decorate},very thick,rotate around={90:(1.5,0)}] (1.5,0) parabola (3.5,-1.5);

\draw[postaction={decorate},very thick,rotate around={30:(-0.75,1.3)}] (-0.75,1.3) parabola (1.25,2.8);
\draw[postaction={decorate},very thick,rotate around={210:(-0.75,1.3)}] (-0.75,1.3) parabola (1.25,-0.2);
\draw[postaction={decorate},very thick,rotate around={-210:(-0.75,-1.3)}] (-0.75,-1.3) parabola (1.25,0.2);
\draw[postaction={decorate},very thick,rotate around={-30:(-0.75,-1.3)}] (-0.75,-1.3) parabola (1.25,-2.8);

\draw(2.5,1.6) node[right]{$C_0^+$}; 
\draw(2.5,-1.6) node[right]{$C_0^-$};
\draw[rotate around={120:(0,0)}] (2.5,1.6) node[above]{$C_1^+$}; 
\draw[rotate around={120:(0,0)}] (2.5,-1.6) node[left]{$C_1^-$};

\draw[rotate around={-120:(0,0)}] (2.5,1.6) node[left]{$C_2^+$}; 
\draw[rotate around={-120:(0,0)}] (2.5,-1.6) node[below]{$C_2^-$};

\draw[postaction={decorate},very thick](0.25,0.4)--(0.4,0.4)--(0.6,0.35)--(0.8,0.2)--(1,0);
\draw[postaction={decorate},very thick](0.25,-0.4)--(0.4,-0.4)--(0.6,-0.35)--(0.8,-0.2)--(1,0);
\draw[postaction={decorate},very thick, rotate around={120:(0,0)}](0.25,0.4)--(0.4,0.4)--(0.6,0.35)--(0.8,0.2)--(1,0);
\draw[postaction={decorate},very thick, rotate around={120:(0,0)}](0.25,-0.4)--(0.4,-0.4)--(0.6,-0.35)--(0.8,-0.2)--(1,0);
\draw[postaction={decorate},very thick, rotate around={-120:(0,0)}](0.25,0.4)--(0.4,0.4)--(0.6,0.35)--(0.8,0.2)--(1,0);
\draw[postaction={decorate},very thick, rotate around={-120:(0,0)}](0.25,-0.4)--(0.4,-0.4)--(0.6,-0.35)--(0.8,-0.2)--(1,0);

\draw[postaction={decorate},  very thick] (0.2,0)--(2,2.4);
\draw[postaction={decorate}, very thick] (0.2,0)--(2,-2.4);
\draw[postaction={decorate},  very thick, rotate around={120:(0,0)}] (0.2,0)--(2,2.4);
\draw[postaction={decorate},  very thick, rotate around={120:(0,0)}] (0.2,0)--(2,-2.4);
\draw[postaction={decorate},  very thick, rotate around={-120:(0,0)}] (0.2,0)--(2,2.4);
\draw[postaction={decorate},  very thick, rotate around={-120:(0,0)}] (0.2,0)--(2,-2.4);

\draw(0.3,0.17) node[right]{\tiny $L_1^+$};
\draw(0.3,-0.18) node[right]{\tiny $L_1^-$};
\draw(-0.4,0.17) node[above]{\tiny $L_1^+$};
\draw(-0.1,0.35) node[above]{\tiny $L_1^-$};
\draw(-0.15,-0.35) node[below]{\tiny $L_1^+$};
\draw(-0.4,-0.17) node[below]{\tiny $L_1^-$};

\draw(1.4,2.5) node[left]{$L_2^+$};
\draw(1.8,2.2) node[left]{$L_2^-$};
\draw(-2.8,0.0) node[below]{$L_2^+$};
\draw(-2.8,0.5) node[below]{$L_2^-$};
\draw(1.8,-2.2) node[left]{$L_2^+$};
\draw(1.4,-2.5) node[left]{$L_2^-$};
\filldraw(1,0) circle (1pt);
\filldraw[rotate around={120:(0,0)}](1,0) circle (1pt);
\filldraw[rotate around={-120:(0,0)}](1,0) circle (1pt);
\end{tikzpicture}
\end{center}
\caption{Lenses $L_1$ and $L_2$ around $\Sigma_1$ and $\Sigma_2$ and the contours $\Sigma_S$ in
the RH problem for $S$.}
\label{fig:lenses}
\end{figure}
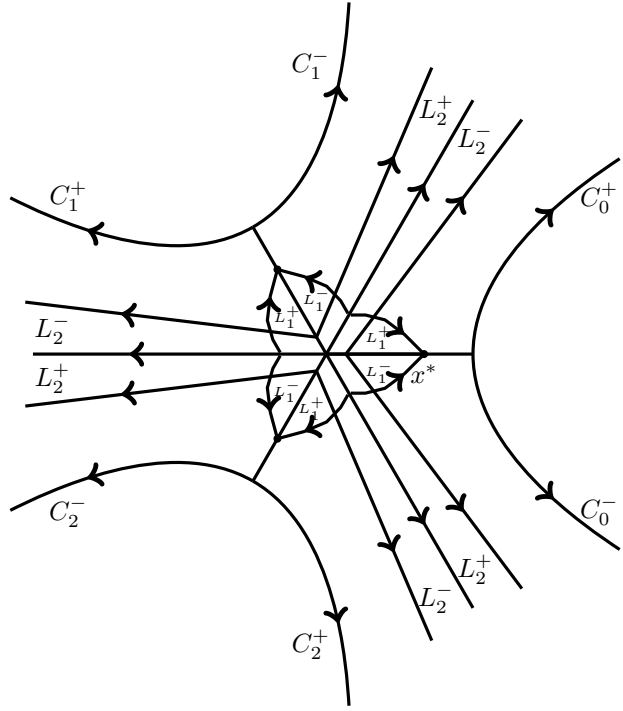

Note that the lenses $L_1$ and $L_2$ have non-empty intersection.
More precisely, $L_1^+ \cap L_1^- \neq \emptyset$ and $L_1^- \cap L_2^+ \neq \emptyset$.

\begin{definition}
We define 
\[ S(z) = T(z) \] 
for $z$ outside of the lenses $L_1$ and $L_2$.
For $z$ inside the lenses we define
\begin{align} \label{defS} 
	S(z) = T(z) \times
	\begin{cases} \begin{pmatrix} 1 & 0 & 0  \\ \mp e^{-2n \varphi_1(z)} & 1 & 0 \\ 0 & 0 & 1 \end{pmatrix},
	& z \in L_1^{\pm} \setminus L_2,  \\
			\begin{pmatrix} 1 & 0 & 0 \\ 0 & 1 & 0 \\ 0 & \mp \omega^{\mp} e^{-2n \varphi_2(z)} & 1 \end{pmatrix}, 
	& z \in L_2^{\pm} \setminus L_1,  \\
			\begin{pmatrix} 1 & 0 & 0  \\ \mp e^{-2n \varphi_1(z)} & 1 & 0 \\ 
			\omega^{\mp} e^{-2n(\varphi_1(z) + \varphi_2(z))} & \pm \omega^{\pm} e^{-2n \varphi_2(z)}  & 1 \end{pmatrix},
	& z \in L_1^{\pm} \cap L_2^{\mp},  
		\end{cases}
		\end{align}
where $\omega^+ = \omega$ and $\omega^- = \omega^{-1} = \omega^2$.	
\end{definition}

Then $S$ satisfies a RH problem on a union of contours $\Gamma_S$ that consists
of $\Gamma_T$ and the lips $\partial L_1$ and $\partial L_2$ of the lenses.
We continue to follow the convention that the contours are oriented away from the origin, and towards infinity.
For $\partial L_1$ and $\partial L_2$ this means that each part in $\partial L_1$ is 
oriented towards a branch point $\omega^j x^*$, and each part in $\partial L_2$
is oriented towards infinity, as indicated in Figure~\ref{fig:lenses}.

\begin{rhproblem} \label{RHforS}
\begin{itemize}
\item $S$ is analytic in $\mathbb C \setminus \Gamma_S$,
\item $S_+ = S_- J_S$ on $\Gamma_S$ with $J_S$ given as follows on $\Gamma_T$,
\begin{align} \label{JS1} 
	J_S(z) = \begin{cases} 
	\begin{pmatrix} 0 & 1 & 0 \\ 
	-1 & 0 & 0 \\ 0 & 0 & 1 \end{pmatrix}, &
		\quad z \in \Sigma_1, \\
	\begin{pmatrix} 1 & 0 & 0 \\ 0 & 0 & 1 \\ 0 & -1 & 0 \end{pmatrix}, &
		\quad z \in \Sigma_2, \\
		\begin{pmatrix} 1 & e^{2n \varphi_1(z)} & 0 \\ 0 & 1 & 0 \\ 0 & 0 & 1 \end{pmatrix},
			& z \in \bigcup_{j=0}^2 [\omega^j x^*, \omega^j \widehat{x}], \\
		\begin{pmatrix} 1 & \frac{1}{2} e^{2n \varphi_1(z)} & 0 \\ 
	0 & 1 & 0 \\ 0 & 0 & 1 \end{pmatrix}, & z \in  \bigcup_{j=0}^2 C_j, 
		\end{cases}
		\end{align}
and on the lips of the lenses by
\begin{align} \label{JS2} J_S(z) = 
	\begin{cases} \begin{pmatrix} 1 & 0 & 0  \\ e^{-2n \varphi_1(z)} & 1 & 0 \\ 0 & 0 & 1 \end{pmatrix},
	& z \in \partial L_1 \setminus L_2,  \\
			\begin{pmatrix} 1 & 0 & 0 \\ 0 & 1 & 0 \\ 0 & \omega^{\mp} e^{-2n \varphi_2(z)} & 1 \end{pmatrix}, 
	& z \in \partial L_2^{\pm} \setminus L_1,  \\
			\begin{pmatrix} 1 & 0 & 0  \\ e^{-2n \varphi_1(z)} & 1 & 0 \\ \pm  e^{-2n(\varphi_1(z) + \varphi_2(z))} & 0 & 1 \end{pmatrix},
	& z \in \partial L_1^{\pm} \cap L_2,  \\
    \begin{pmatrix} 1 & 0 & 0 \\ 0 & 1 & 0 \\ \omega^{\pm} e^{-2n(\varphi_1(z) + \varphi_2(z))} & \omega^{\mp} e^{-2n \varphi_2(z)} & 1 \end{pmatrix},
	& z \in \partial L_2^{\pm} \cap  L_1.
		\end{cases}
\end{align}
\item $ S(z) = (I + O(1/z)) A(z)$ as $z \to \infty$ where $A(z)$ is given by \eqref{defAz}.
\end{itemize}
\end{rhproblem}

\begin{proof}
The jump matrices $J_S$ arise from combining the jump condition \eqref{JT} for $T$ with the
definition \eqref{defS}. As an example we show how the jump matrix \eqref{JS2} on $\partial L_1^+ \cap L_2$
arises.

To compute the jump matrix on $\partial L_1^+ \cap L_2$ we
need the definition \eqref{defS} of $S$ in the parts $L_1^+ \cap L_2$ (which is on the $-$side of  $\partial L_1^+ \cap L_2$)
and $L_2^- \setminus L_1$ (which is on the right).  Since $T$ has no jump we find from \eqref{defS} that,
\begin{align*} 
	J_S(z) & = 
	\begin{pmatrix} 1 & 0 & 0 \\ - e^{-2n \varphi_1(z)} & 1 & 0 \\ \omega^2 e^{-2n (\varphi_1(z) + \varphi_2(z))} & \omega e^{-2n \varphi_2(z)} & 1 \end{pmatrix}^{-1}
	\begin{pmatrix} 1 & 0 & 0 \\ 0 & 1 &0 \\ 0 & \omega e^{-2n \varphi_2(z)} & 1 \end{pmatrix} \\
	& = 	\begin{pmatrix} 1 & 0 & 0 \\ e^{-2n \varphi_1(z)} & 1 & 0 \\ e^{-2n (\varphi_1(z) + \varphi_2(z))} & -\omega e^{-2n \varphi_2(z)} & 1 \end{pmatrix}
	\begin{pmatrix} 1 & 0 & 0 \\ 0 & 1 &0 \\ 0 & \omega e^{-2n \varphi_2(z)} & 1 \end{pmatrix},
\end{align*} 
since $1+ \omega + \omega^2 = 0$. This then easily reduces to the jump matrix $J_S$ as given in \eqref{JS2}
on $\partial L_1^+ \cap L_2$. The other jumps follow in a similar way.

The asymptotic condition for $S$ outside of the lens $L_2$ follows trivially from the asymptotic condition
in the RH problem \ref{RHforT} for $T$. Inside $L_2$, we note that by  \eqref{phi2onSigma2} we have 
\[ \varphi_{2}(z) = \pm \frac{2}{3t_0 \sqrt{t_3}} z^{3/2} + O(\log z) \qquad \text{as } z \to \infty, \, z \in L_2, \]
where the $\pm$ signs are such that $\Re \varphi_2(z) > 0$ for $z$ large enough in  $L_2$.
Then it follows from \eqref{defS} that the asymptotic condition is also valid in the lens $L_2$.
\end{proof}

We are now in a situation where all jump matrices $J_S$ tend to the identity matrix 
as $n \to \infty$, except for those on $\Sigma_1$ and $\Sigma_2$, which are constant.
We already saw this for the jumps on $\bigcup_j (\omega^j x^*, \omega_j \widehat{x}] \cup C_j)$, see  \eqref{JS1},
since $\Re \varphi_1 < 0$ on these parts of the contour. Also $L_1$ is contained in
the region where $\Re \varphi_1 > 0$, so that the entries $e^{-2n \varphi_1(z)}$ that appear in
\eqref{JS2} are indeed decaying as $n \to \infty$.

Regarding $\varphi_2$, we note that by \eqref{gfunctions} and \eqref{phi2onSigma2} we have that
$\varphi_{2,\pm}(z)$ is purely imaginary for $z \in \Sigma_2$, with 
\begin{equation} \label{phi2increasing}
\begin{aligned}  i \varphi_{2,+}(z) & = - i \varphi_{2,-}(z) \\
	  & = \pi \mu_2^*([0, z]) + \frac{2}{3 t_0 \sqrt{t_3}} |z|^{3/2} 
		+ \begin{cases} - \pi/6,  & \arg z = \pm \pi/3, \\
		  \pi/3, & \arg z = \pi.
		  \end{cases} 
		  \end{aligned}
		  \end{equation}
Thus $i \varphi_{2,+}$ and $-i \varphi_{2,-}$ are strictly increasing along each of the
rays in $\Sigma_2$. Then by the Cauchy-Riemann equations, we have that $\Re \varphi_2 > 0$
both to the left and to the right of $\Sigma_2$. We may (and do) assume that the lips of
the lens $L_2$ is contained in the region where $\Re \varphi_2 > 0$ and  it indeed follows that the jump matrices \eqref{JS2} 
tend to the identity matrix as $n \to \infty$.

\begin{remark} \label{remark3}
The asymptotic condition in the RH problem \ref{RHforS} for $S$ is valid uniformly as $z \to \infty$
in any direction in the complex plane.
This is in contrast to what happens in the RH problems  for $X$, $V$, $U$ and $T$, where
the asymptotic condition is not valid uniformly in the directions $\arg z = \pm \pi/3$, $\pi$ that
correspond to $\Sigma_2$, see Remark \ref{remark1}.
 
The term $e^{-2n \varphi_2(z)}$ that appears in  the transformation \eqref{defS} for $z \in L_2^{\pm} \setminus L_1$
is small for $z$ away from $\Sigma_2$, but becomes of order $1$ as  $z$ approaches $\Sigma_2$. 
On $\Sigma_2$ a certain cancellation takes place, which results in the asymptotic condition being valid uniformly. 
One can establish this rigorously by looking back at the transformations $Y \mapsto X \mapsto V \mapsto U \mapsto 
T \mapsto S$ near $\Sigma_2$, but we will not go into that here. 

If we do not want to rely on this detailed analysis, then we can only guarantee at this stage that
the asymptotic condition in the RH problem \ref{RHforS} for $S$ is valid uniformly as $z \to \infty$ with
\[ \arg z \in  (-\pi + \varepsilon, - \pi/3 - \varepsilon) \cup (-\pi/3 + \varepsilon, \pi/3 - \varepsilon)
	\cup (\pi/3 + \varepsilon, \pi - \varepsilon) \]
and that
\begin{equation} \label{Sasymp2} 
	S(z) = O(1) A(z) \qquad \text{ uniformly  as } z \to \infty.
\end{equation}
The latter condition is a consequence of \eqref{Xasymp2} in Remark \ref{remark1} and the transformations
leading from $X$ to $S$.

%
\end{remark}

\subsection{Global parametrix}\label{subsecglobal}

The global parametrix $M$ is a solution of the following model RH problem,
which we obtain from the RH problem \ref{RHforS} for $S$ by dropping the jump matrices
that tend to the identity matrix as $n \to \infty$.

\begin{rhproblem} \label{RHforM}
\begin{itemize}
\item $M$ is analytic in $\mathbb C \setminus \Gamma_M$ where $\Gamma_M = \Sigma_1 \cup \Sigma_2$,
\item $M_+ = M_- J_M$ on $\Gamma_M$ with
\[ J_M(z) = \begin{cases} 
	\begin{pmatrix} 0 & 1 & 0 \\ 
	-1 & 0 & 0 \\ 0 & 0 & 1 \end{pmatrix}, &
		\quad z \in \Sigma_1, \\
	\begin{pmatrix} 1 & 0 & 0 \\ 0 & 0 & 1 \\ 0 & -1 & 0 \end{pmatrix}, &
		\quad z \in \Sigma_2,
		\end{cases} \]  
\item $ M(z) = (I + O(1/z)) A(z)$ as $z \to \infty$, where $A(z)$ is given by \eqref{defAz}.  
\item $M(z) = O\left((z - \omega^j x^*)^{-1/4} \right)$ as $z \to \omega^j x^*$ for $j=0,1,2$.
\item $M(z)$ remains bounded as $z \to 0$.
\end{itemize}
\end{rhproblem}
Note that the product of the jump matrices around $0$ is the identity matrix 
and so we can indeed require that $M$ remains bounded near the origin.
Also observe that the asymptotic condition is compatible with the jumps on the unbounded contours,
since 
\begin{equation} \label{jumpA} 
	A_+ = A_- J_M \qquad \text{ on } \Sigma_2, 
	\end{equation}
which may be checked directly from the definition \eqref{defAz} of $A(z)$.

To construct $M$ we use a meromorphic differential $\Omega$ on the Riemann surface $\mathcal R$,
which we specify by its poles and residues, see also \cite{DuKuMo, KuiMo, Mo}. 
We require that $\Omega$ has simple poles at the branch points
$\omega^j x^*$, $j=0,1,2$ and at $\infty_2$ (the point at infinity that is common to the second
and third sheets) with
\[ \Res(\Omega, \omega^j x^*) = - \tfrac{1}{2}, \qquad \Res(\Omega, \infty_2) = \tfrac{3}{2} \]
and $\Omega$ is holomorphic elsewhere. The residues add up to $0$, and therefore the
meromorphic differential exists and it is also unique since the  genus is zero.

We use $\infty_1$ as base point for integration of $\Omega$ and define
\begin{equation} \label{defuj} 
	u_j(z) = \int_{\infty_1}^z \Omega, \qquad z \in \mathcal R_j, \quad j=1,2,3, 
	\end{equation}
where the path of integration is chosen according to the following rules:
\begin{itemize}
\item The path for $u_1(z)$ stays on the first sheet.
\item The path for $u_2(z)$ starts on
the first sheet and passes once through $\Sigma_1$ to go to the second sheet
and then stays on the second sheet.
\item The path for $u_3(z)$ starts on the first sheet, passes once through 
$\Sigma_1$ to go to the second sheet, and then passes once through $\Sigma_2$
to go the third sheet and then stays on the third sheet.
\item All passages from one sheet to the next  go via the $-$-side on the upper
sheet to the $+$-side on the lower sheet.
\end{itemize}

Then the functions $u_j$ are well-defined, $u_1$ is analytic on $\mathbb C \setminus \Sigma_1$,
$u_2$ is analytic on $\mathbb C \setminus(\Sigma_1 \cup \Sigma_2)$, and $u_3$ is analytic
on $\mathbb C \setminus \Sigma_2$. The functions satisfy
\begin{align*}
	u_{3,+} = u_{3,-} \qquad u_{1,-} = u_{2,+}, \qquad u_{1,+} = u_{2,-} \pm \pi i, \qquad  \text{on } \Sigma_1, \\
	u_{1,+} = u_{1,-}, \qquad u_{2,-} = u_{3,+}, \qquad u_{2,+} = u_{3,-} \pm \pi i, \qquad \text{on } \Sigma_2. 
\end{align*}
We put 
\begin{equation} \label{defvj} 
	v_j = e^{u_j}
	\end{equation} 
	and then we have
\begin{align*}
	\begin{pmatrix} v_1 & v_2 & v_3 \end{pmatrix}_+
		= 	\begin{pmatrix} v_1 & v_2 & v_3 \end{pmatrix}_- J_M \qquad \text{on } \Sigma_1 \cup \Sigma_2,
		\end{align*}
Since $v_1(z) = 1 + O(1/z)$, 
$v_2(z) = O(z^{-3/4})$, $v_3(z) = O(z^{-3/4})$ as $z \to \infty$, the vector $(v_1, v_2, v_3)$
satisfies the conditions for the first row of $M$.

The vector space of holomorphic functions on $\mathcal R \setminus \{\infty_2\}$ 
with at most a double pole at $\infty_2$ has dimension $3$. Let $f^{(1)} \equiv 1,f^{(2)}, f^{(3)}$
be a basis of this vector space. We use $f^{(i)}_j$ to denote the restriction of $f^{(i)}$ to
the sheet $\mathcal R_j$. Then it is easy to see that 
\[ B = \begin{pmatrix} v_1 & v_2 & v_3 \\
	v_1 f^{(2)}_1 & v_2 f^{(2)}_2 & v_3 f^{(2)}_3 \\
	v_1 f^{(3)}_1 & v_2 f^{(3)}_2 & v_3 f^{(3)}_3 
\end{pmatrix} \]
is defined and analytic in $\mathbb C \setminus \Gamma_M$ with jump $B_+ = B_- J_M$. 
Also $B(z) = O(z^{1/4})$ as $z \to \infty$. It then easily follows that $\det B \equiv {\rm const}$
and the constant is non-zero, since the functions $f^{(1)} \equiv 1$, $f^{(2)}$, $f^{(3)}$ are independent.

We already noted that $A_+ = A_- J_M$ on $\Sigma_2$. Since $A$ and $B$ have the same
jumps on $\Sigma_2$, it follows that $B A^{-1}$ is analytic in $\mathbb C \setminus \Sigma_1$,
and therefore has a Laurent expansion at infinity. Since both $B(z) = O(z^{1/4})$ and $A(z) = O(z^{1/4})$
as $z \to \infty$, we have  
\[ B(z) A^{-1}(z) = C + O(z^{-1}) \qquad \text{as } z \to \infty \]
for some constant matrix $C$. Since $\det B \equiv {\rm const} \neq 0$ and $\det A \equiv 1$,
we see that $\det C \neq 0$, and so $C$ is invertible. Then
\[ M(z) = C^{-1} B(z) \]
satisfies
\[ M(z) = (I + O(z^{-1})) A(z) \qquad \text{as } z \to \infty. \]
It also satisfies the jump condition $M_+ = M_- J_M$, it is bounded at $0$ and it has
at most fourth root singularities at the branch points. Thus $M$ is the global parametrix we
are looking for. 
 
\begin{remark} \label{remarkM11}
From the above construction it follows that
\[ M_{1,1}(z) = v_1(z) = e^{u_1(z)}, \qquad z \in \mathbb C \setminus \Sigma_1, \]
which in particular means that $M_{1,1}(z) \neq 0$ for $z \in \mathbb C \setminus \Sigma_1$.
\end{remark}

\subsection{Local parametrices} \label{subseclocal}

The local parametrix $P$ is defined in the disks
\[ D(\omega^j x^*, \delta) = \{ z \in \mathbb C \mid |z-\omega^j x^*| < \delta \} \]
where $\delta > 0$ is taken sufficiently small.

\begin{rhproblem} \label{RHforP}
\begin{itemize}
\item $P$ is continuous on $(\bigcup_j D(\omega^j x^*, \delta) ) \setminus \Gamma_S$ and
analytic in $(\bigcup_j D(\omega^j x^*, \delta) ) \setminus \Gamma_S$,
\item $P_+ = P_- J_S$ on $\Gamma_S \cap \bigcup_j D(\omega^j x^*, \delta)$,
where $J_S$ is as in \eqref{JS1}--\eqref{JS2}
\item $P(z)$ matches with the global parametrix $M(z)$ in the sense that
\begin{equation} \label{matching} 
	P(z) = M(z) \left(I + O(n^{-1})\right) \quad \text{ as } n \to \infty,
	\end{equation}
 uniformly for $z \in \bigcup_j \partial D(\omega_j x^*,\delta)$.
\end{itemize}
\end{rhproblem}

In the noncritical case $t_0 < t_{0,crit}$ the density of $\mu_1$ vanishes as a square root at $\omega^j x^*$,
for $j=0,1,2$, which means that we can build the local parametrix $P$ out of Airy functions in a small disk
around each of these endpoints.  The construction is done in a standard way and we do not give details here.

\subsection{Sixth transformation} \label{subsecsixth}

Now we make the final transformation. We choose a small disk $D(\omega^j, x^*,\delta)$ around each of the branch
points, and define
\begin{definition}
	We define
	\begin{equation} \label{defR} 
		R(z) = \begin{cases} 
	S(z) P(z)^{-1} & \text{ in the disks around the endpoints } \omega^j x^*, \, j =0,1,2, \\
	S(z) M(z)^{-1} & \text{ outside of the disks.}
	\end{cases} \end{equation}
\end{definition}

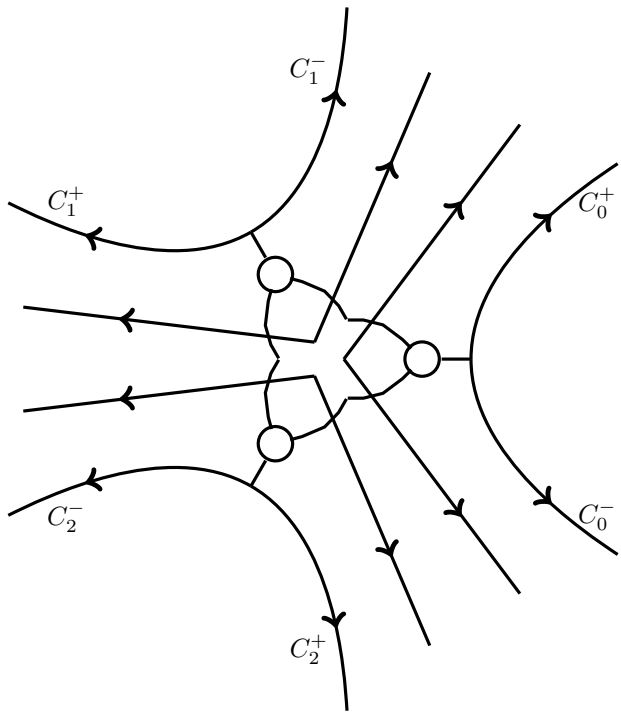
\begin{figure}[t]
\begin{center}
\begin{tikzpicture}[scale=1.3,decoration={markings,mark=at position .67 with {\arrow[black,line width=0.8mm]{>};}}]

\draw[very thick] (1.2,0)--(1.5,0); 
\draw[very thick, rotate around ={120:(0,0)}] (1.2,0)--(1.5,0); 
\draw[very thick, rotate around ={-120:(0,0)}] (1.2,0)--(1.5,0); 


\draw[postaction={decorate},very thick,rotate around={-90:(1.5,0)}] (1.5,0) parabola (3.5,1.5);
\draw[postaction={decorate},very thick,rotate around={90:(1.5,0)}] (1.5,0) parabola (3.5,-1.5);

\draw[postaction={decorate},very thick,rotate around={30:(-0.75,1.3)}] (-0.75,1.3) parabola (1.25,2.8);
\draw[postaction={decorate},very thick,rotate around={210:(-0.75,1.3)}] (-0.75,1.3) parabola (1.25,-0.2);
\draw[postaction={decorate},very thick,rotate around={-210:(-0.75,-1.3)}] (-0.75,-1.3) parabola (1.25,0.2);
\draw[postaction={decorate},very thick,rotate around={-30:(-0.75,-1.3)}] (-0.75,-1.3) parabola (1.25,-2.8);

\draw(2.5,1.6) node[right]{$C_0^+$}; 
\draw(2.5,-1.6) node[right]{$C_0^-$};
\draw[rotate around={120:(0,0)}] (2.5,1.6) node[above]{$C_1^+$}; 
\draw[rotate around={120:(0,0)}] (2.5,-1.6) node[left]{$C_1^-$};

\draw[rotate around={-120:(0,0)}] (2.5,1.6) node[left]{$C_2^+$}; 
\draw[rotate around={-120:(0,0)}] (2.5,-1.6) node[below]{$C_2^-$};

\draw[very thick](0.24,0.4)--(0.4,0.4)--(0.6,0.35)--(0.8,0.2)--(0.89,0.11); 
\draw[very thick](0.24,-0.4)--(0.4,-0.4)--(0.6,-0.35)--(0.8,-0.2)--(0.89,-0.11); 
\draw[very thick, rotate around={120:(0,0)}](0.24,0.4)--(0.4,0.4)--(0.6,0.35)--(0.8,0.2)--(0.89,0.11); 
\draw[very thick, rotate around={120:(0,0)}](0.24,-0.4)--(0.4,-0.4)--(0.6,-0.35)--(0.8,-0.2)--(0.89,-0.11); 
\draw[very thick, rotate around={-120:(0,0)}](0.24,0.4)--(0.4,0.4)--(0.6,0.35)--(0.8,0.2)--(0.89,0.11); 
\draw[very thick, rotate around={-120:(0,0)}](0.24,-0.4)--(0.4,-0.4)--(0.6,-0.35)--(0.8,-0.2)--(0.89,-0.11); 

\draw[postaction={decorate},  very thick] (0.2,0)--(2,2.4);
\draw[postaction={decorate}, very thick] (0.2,0)--(2,-2.4);
\draw[postaction={decorate},  very thick, rotate around={120:(0,0)}] (0.2,0)--(2,2.4);
\draw[postaction={decorate},  very thick, rotate around={120:(0,0)}] (0.2,0)--(2,-2.4);
\draw[postaction={decorate},  very thick, rotate around={-120:(0,0)}] (0.2,0)--(2,2.4);
\draw[postaction={decorate},  very thick, rotate around={-120:(0,0)}] (0.2,0)--(2,-2.4);


\draw[very thick](1,0) circle (5pt);
\draw[very thick,rotate around={120:(0,0)}](1,0) circle (5pt);
\draw[very thick,rotate around={-120:(0,0)}](1,0) circle (5pt);
\end{tikzpicture}
\end{center}
\caption{Contour $\Gamma_R$ for the RH problem for $R$.}
\label{fig:contourR}
\end{figure}

Then $R$ is defined and analytic outside of the union of $\Gamma_S$ with
the  circles around the branch points.
Since the jump matrices of $S$ and $M$ agree on $\Sigma_1 \cup \Sigma_2$, we see that $R$
has an analytic continuation across $\Sigma_2$ and across the part of $\Sigma_1$ that
is outside of the three disks. Similarly, the jump matrices on $S$ and $P$ agree inside
the disks, and so $R$ also has analytic continuous inside the three disks.
It follows that $R$ is analytic in $\mathbb C \setminus \Gamma_R$,
where $\Gamma_R$ is the contour shown in Figure \ref{fig:contourR}. It consists
of the part of $\Sigma_S \setminus (\Sigma_1 \cup \Sigma_2)$ that is outside of
the three disks, together with the three circles around the branch points.

The RH problem for $R$ is thus as follows.
\begin{rhproblem} \label{RHforR}
\begin{itemize}
\item $R$ is analytic in $\mathbb C \setminus \Gamma_R$,
\item $R_+ = R_- J_R$ on $\Gamma_R$ where
\begin{align} \label{JR}
	J_R(z) & = \begin{cases} M(z)^{-1} P(z) & \text{for $z$ on the circles}, \\
	  M(z)^{-1} J_S(z) M(z) & \text{elsewhere on $\Gamma_R$.}
	  \end{cases}
\end{align}
\item $R(z) = I + O(z^{-1})$ as $z \to \infty$.
\end{itemize}
\end{rhproblem}

\begin{remark} \label{remark4}
Following Remarks \ref{remark1} and \ref{remark3} we note that the asymptotic 
condition in the RH problem \ref{RHforR}  is valid uniformly as $z \to \infty$ 
with
\[ \arg z \in (-\pi + \varepsilon, - \pi/3 - \varepsilon) \cup (-\pi/3 + \varepsilon, \pi/3 - \varepsilon)
	\cup (\pi/3 + \varepsilon, \pi - \varepsilon) \]
for any $\varepsilon > 0$.
From Remark \ref{remark3} and the transformation \eqref{defR} it also follows
that $R(z) = O(1)$ as $z \to \infty$ uniformly.
Since the jump matrix $J_R(z)$ for $z \in \Gamma_R$ tends
to the identity matrix as $z \to \infty$, see also \eqref{JRestimate2} below, it is then
easy to show that the asymptotic condition $R(z) = I + O(z^{-1})$ as $z \to \infty$
is in fact valid uniformly.
\end{remark}

All jump matrices $J_R$ in the RH problem \ref{RHforR} tend 
to the identity matrix as $n \to \infty$. Indeed, because of the
matching condition \eqref{matching} we have
\begin{equation} \label{JRestimate1} 
	J_R(z) = I + O(n^{-1}), \qquad \text{for } z \in \partial D(\omega^j x^*, \delta) 
	\end{equation}
while the jump matrices on the remaining parts of $\Gamma_R$ are
exponentially close to 
\begin{equation} \label{JRestimate2} 
	J_R(z) = I + O(e^{- c n |z|^3}), \qquad \text{elsewhere on $\Gamma_R$.} 
	\end{equation}

We conclude two things from \eqref{JRestimate1}--\eqref{JRestimate2}.
The first thing is that the RH problem \ref{RHforR} has a solution
if $n$ is large enough. This indeed follows from \eqref{JRestimate1}--\eqref{JRestimate2}
since for jump matrices $J_R$ close enough to the identity matrix the RH problem \ref{RHforR}
has a unique solution which can be written down as a Neumann series. 
Since the transformations 
\begin{equation} \label{transformations}
	Y \mapsto X \mapsto V \mapsto U \mapsto T \mapsto S \mapsto R
	\end{equation}
	are invertible it then also follows that the solution $Y$ to the
	RH problem \ref{RHforY2} uniquely exists. So in particular the $(1,1)$
	entry exists, which is the orthogonal polynomial $P_{n,n}$. 

The second thing we obtain from \eqref{JRestimate1}--\eqref{JRestimate2}
is that the solution $R$ of the RH problem \ref{RHforR} not only exists
but is also close to the identity matrix as $n \to \infty$. There is an estimate
\begin{equation} \label{Restimate} 
	R(z) = I + O \left( \frac{1}{n(1 + |z|)} \right)  \qquad \text{as } n \to \infty,
	\end{equation}
uniformly for $z \in \mathbb C \setminus \Gamma_R$.
The estimate \eqref{Restimate} completes the steepest descent analysis of the
RH problem.

\subsection{Proof of Lemma \ref{lemmaRH}}

The proof of Lemma \ref{lemmaRH} now follows by unravelling the transformations
\eqref{transformations} and then use \eqref{Restimate}
to see the effect on the orthogonal polynomial, since 
\[ P_{n,n}(z) = Y_{1,1}(z), \]
see \eqref{Y11}.

\begin{proof}

From the definitions \eqref{defX1}--\eqref{defX2}, \eqref{defV1}--\eqref{defV3}, \eqref{defU},
	\eqref{defT1}--\eqref{defT2} we easily see that 
\begin{equation} \label{T11}
	\begin{aligned} 
	P_{n,n}(z) & = X_{1,1}(z) = V_{1,1}(z)  = U_{1,1}(z) e^{n g_1(z)} \\
		& = T_{1,1}(z) e^{n g_1(z)}, \qquad z \in \mathbb C \setminus \Sigma_1. 
		\end{aligned}
	\end{equation}
Using \eqref{T11} and \eqref{defS} we find 
\begin{equation} \label{S11} 
	P_{n,n}(z) = S_{1,1}(z) e^{n g_1(z)}, \qquad z \in \mathbb C \setminus L_1. 
	\end{equation}
For $z$ outside of the disks $D(\omega^j x^*, \delta)$ we have $S = RM$ by \eqref{defR} 
and so by \eqref{Restimate}
\[ S_{1,1}(z) = (1 + O(1/n)) M_{1,1}(z) + O(1/n) \]
which, since $M_{1,1}(z)$ is analytic with no zeros in $\overline{\mathbb C} \setminus \Sigma_1$, see Remark \ref{remarkM11},
can also be written as
\[ S_{1,1}(z) = (1 + O(1/n)) M_{1,1}(z), \qquad z \in  \mathbb C \setminus (L_1 \cup \bigcup_j D(\omega^j x^*, \delta)), \]
and the $O$-term is uniform.
Inserting this into \eqref{S11} we obtain \eqref{Pnnasymptotics} 
uniformly for $z \in \mathbb C \setminus (L_1 \cup \bigcup_j D(\omega^j x^*,\delta))$.
The lense $L_1$ around $\Sigma_1$ and the disks $D(\omega^j x^, \delta)$ around the branch points 
can be made as small as we like.  
It follows that \eqref{Pnnasymptotics} holds uniformly for $z$ in compact subsets of $\overline{ \mathbb C} \setminus \Sigma_1$
and the lemma is proved.
\end{proof}

\subsection{Proofs of Theorems \ref{theorem2} and \ref{theorem3}} \label{subsec:theorem23proof}

\begin{proof}
For $t_0, t_3 > 0$ with $t_0 < t_{0,crit}$ we proved that the orthogonal polynomial
$P_{n,n}$ exists for large enough $n$ and satisfies \eqref{Pnnasymptotics}.
Since $M_{1,1}(z) \neq 0$ for $z \in \mathbb C \setminus \Sigma_1$, the zeros of $P_{n,n}$
accumulate on $\Sigma_1$ as $n \to \infty$, and this proves Theorem~\ref{theorem2}.

From \eqref{Pnnasymptotics} and \eqref{gfunctions} it follows that, uniformly for $z \in \mathbb C \setminus \Sigma_1$,
\[ \lim_{n \to \infty} \frac{1}{n} \log |P_{n,n}(z)| = \Re g_1(z) = \int \log |z-s| d\mu_1^*(s). \]
Standard arguments from logarithmic potential theory, see e.g.\ \cite[Theorem III.4.1]{SafTot} then imply
that $\mu_1^*$ is the weak limit of the normalized zero counting measures of the polynomials
$P_{n,n}$ as $n \to \infty$. By construction, $\mu_1^*$ is the first component of the
minimizer of the vector equilibrium problem, see also Lemma~\ref{lemma41}, and Theorem~\ref{theorem3}
follows.
\end{proof}

\subsection*{Acknowledgements}
The first author is supported in part by the National Science Foundation (NSF) Grants DMS-0652005
and DMS-0969254.

The second author is supported in part by FWO-Flanders projects G.0427.09 and G.0641.11, by K.U.~Leuven
research grant OT/08/33, by the Belgian Interuniversity Attraction Pole P06/02, and by grant MTM2008-06689-C02-01 of the
Spanish Ministry of Science and Innovation.

This work was started when both authors were visiting the Mathematical Sciences Research Institute in
Berkeley in the fall of 2010 during the program ``Random Matrices, Interacting Particle Systems and
Integrable Systems''. We thank the MSRI for its hospitality and wonderful research conditions.

\end{document}